\documentclass[reqno]{amsart}

\usepackage[margin=1in]{geometry}
\usepackage{amsmath}%
\usepackage{amsfonts}%
\usepackage{amssymb}%
\usepackage{graphicx}
\usepackage{subfig}
\usepackage{color}
\usepackage{enumitem}
\newtheorem{theorem}{Theorem}
\theoremstyle{plain}

\newtheorem{definition}{Definition}

\newtheorem{lemma}{Lemma}

\newtheorem{remark}{Remark}

\numberwithin{equation}{section}

\newcommand{\argmax}[1]{\underset{#1}{\operatorname{argmax}}}

\newcommand{\matR}{\ensuremath{\mathbb{R}}}
\newcommand{\tanps}{\ensuremath{T_{P}S}}
\newcommand{\tanpps}{\ensuremath{T_{P}S^{\perp}}}

\newcommand{\neighp}{\ensuremath{\mathcal{N}_{\varepsilon}(P)}}

\newcommand{\sampwidth}{\ensuremath{\nu}}

\newcommand{\kfmax}{\ensuremath{\mathcal{K}_{max}}}

\newcommand{\kflj}{\ensuremath{\mathcal{K}_{l,j}}}

\newcommand{\kflonej}{\ensuremath{\mathcal{K}_{l,j}}}
\newcommand{\kfltwoj}{\ensuremath{\mathcal{K}_{k,j}}}
\newcommand{\kflk}{\ensuremath{\mathcal{K}_{l,k}}}

\newcommand{\norm}[1]{\parallel{#1}\parallel}
\newcommand{\smoothconst}{\ensuremath{C_{s}}}

\newcommand{\abs}[1]{|{#1}|}
\newcommand{\set}[1]{\left\{{#1}\right\}}

\begin{document}

\title[Tangent space estimation for smooth manifolds]{Tangent space estimation for smooth embeddings of Riemannian manifolds}
\author{Hemant Tyagi, Elif Vural and Pascal Frossard}
\thanks{H. Tyagi is with the Institute of Theoretical Computer Science, ETH Zurich. E. Vural and P. Frossard are with the Ecole Polytechnique F\'{e}d\'{e}rale de Lausanne (EPFL), Signal Processing Laboratory (LTS4), CH-1015 Lausanne, Switzerland. Emails: htyagi@inf.ethz.ch, elif.vural@epfl.ch, pascal.frossard@epfl.ch.}
\thanks{This work has been mostly performed while the first author was with the Signal Processing Laboratory LTS4 at EPFL. It has been partly supported by the Swiss National Science Foundation under Grant 200020-132772.}
\thanks{ACCEPTED FOR PUBLICATION IN: Information and Inference: A Journal of the IMA}

%
\begin{abstract}
Numerous dimensionality reduction problems in data analysis involve the recovery of low-dimensional models or the learning of manifolds underlying sets of data. Many manifold learning methods require the estimation of the tangent space of the manifold at a point from locally available data samples. Local sampling conditions such as (i) the size of the neighborhood (\textit{sampling width}) and (ii) the number of samples in the neighborhood (\textit{sampling density}) affect the performance of learning algorithms. In this work, we propose a theoretical analysis of local sampling conditions for the estimation of the tangent space at a point $P$ lying on a $m$-dimensional Riemannian manifold $S$ in $\mathbb{R}^n$. Assuming a smooth embedding of $S$ in $\mathbb{R}^n$, we estimate the tangent space $\tanps$ by performing a Principal Component Analysis (PCA) on points sampled from the neighborhood of $P$ on $S$. Our analysis explicitly takes into account the second order properties of the manifold at $P$, namely the \textit{principal curvatures} as well as the higher order terms.
We consider a random sampling framework and leverage recent results from random matrix theory to derive conditions on the \textit{sampling width} and the local \textit{sampling density} for an accurate estimation of tangent subspaces. We measure the estimation accuracy by the angle between the estimated tangent space $\widehat{T}_P S$ and the true tangent space $T_P S$ and we give conditions for this angle to be bounded with high probability. In particular, we observe that the local sampling conditions are highly dependent on the correlation between the components in the second-order local approximation of the manifold. We finally provide numerical simulations to validate our theoretical findings.
\end{abstract}

\maketitle

%
\section{Introduction}
\label{sec:intro}
\noindent A data set that resides in a high-dimensional ambient space and that is locally homeomorphic to a lower-dimensional Euclidean space constitutes a manifold. For example, a set of signals that is representable by a parametric model, such as parametrizable visual signals or acoustic signals form a manifold. Data manifolds are however rarely given in an explicit form. The recovery of low-dimensional structures underlying a set of data, also known as manifold learning, has thus been a popular research problem in the recent years. This is typically achieved by constructing a mapping from the original data in the high-dimensional space to a space of much lower dimension. Importantly, most manifold learning methods rely on the assumption that the data has a locally linear structure. Of course, for such an assumption to be valid at some reference point on the manifold, one has to take into account (i) the size of the neighborhood from which the samples are chosen and also, (ii) the number of neighborhood points. For instance, if the manifold is a linear subspace, then the neighborhood can be chosen to be arbitrarily large and the number of samples needs to be simply greater than the dimension of the manifold. However, most manifolds are typically nonlinear, which prevents the selection of an arbitrarily large neighborhood size. Hence, one might expect the existence of an upper bound on the neighborhood size. Furthermore, the number of necessary samples is likely to vary according to the local characteristics of the manifold.

The purpose of this work is to analyze the relation between the sampling conditions of a manifold and the validity of the local linearity assumption of the data sampled from the manifold. We characterize the local linearity of the data with the accuracy of the tangent space estimation. We do a local analysis around a point $P$ on a manifold $S$. We examine the deviation between the tangent space $\widehat{T}_P S$ estimated using manifold samples in a neighborhood of $P$, and the true tangent space $T_P S$ at $P$. This deviation is related to the local geometric properties of the manifold around $P$ and the local sampling conditions. In this paper, $S$ is assumed to be an $m$-dimensional Riemannian manifold in $\mathbb{R}^n$ that can be locally represented with smooth ($\mathcal{C}^r$, $r > 2$) mappings, where $m < n$. We consider a random sampling where the orthogonal projections of the samples to $T_P S$ in a neighborhood of $P$ is uniform. We derive bounds on the size of the neighborhood and on the number of samples such that the deviation (i.e., the angle) between $\widehat{T}_P S$ and $T_P S$ is upper bounded with high probability. In particular, our analysis captures the dependency of the sampling conditions on the second-order properties of the manifold, namely the local curvature of $S$ at $P$, and on the higher-order terms. Thus, broadly speaking, this work consists of a theoretical analysis of the manifold sampling problem that relates the \textit{local} sampling conditions to the accuracy of the local linearity assumption. This paper builds on our preliminary work \cite{Tyagi2011}, where the sampling of manifolds represented with quadratic embeddings is examined, and extends the analysis to arbitrary smooth embeddings. We envisage two main applications where our study can prove to be useful. Firstly, our results can be used for deducing performance guarantees or for determining a good local subset of data samples that gives an accurate estimation of the tangent space in \textit{manifold learning} applications. Secondly, our analysis can also be used in \textit{manifold sampling} applications, i.e., for choosing samples from a manifold with a known parametric model. The discretization of a manifold can be achieved in various ways depending on the target application (see for example \cite{Vural2011}); however, in certain cases one may want to sample the manifold in such a way that the local linearity of the data is preserved and the tangent space can be correctly recovered from data samples.


The manifold learning problem has been largely studied and we provide now a brief overview of the literature, with a special focus on locally linear approximation methods. The manifold structure of data can be retrieved in various ways, from a global parameterization based on geodesic distances as in ISOMAP ~\cite{Tenenbaum2000}, or via locally linear representations as in LLE ~\cite{Roweis2000} and Hessian Eigenmaps ~\cite{Donoho2003}. The LLE algorithm considers the locally linear structure of the manifold, where each data sample is approximated by a weighted linear combination of its nearest neighbors. Then, the key idea in computing a mapping of the data is the preservation of these weights in the embedded low-dimensional space. Moreover, there are other algorithms such as \cite{Lin2006} which employ the locally linearity of the data by expressing the tangent plane as a linear combination of the manifold samples in a local neighborhood. The Hessian Eigenmaps algorithm is similar to LLE in the sense that it is based on locally linear approximations of the manifold. However, it has been seen to be more robust than LLE as it also takes more detailed geometric characteristics of the manifold into account. With similar ideas, an adaptive manifold learning algorithm is presented in \cite{Zhang2011}, where the authors propose an adaptive local neighborhood size selection strategy. Finally, the work in \cite{Zha2009} examines the conditions under which manifold learning algorithms are able to recover true global parameters from local structures computed with data samples. In particular, the authors show that the error in the global parameterization depends on the local approximation errors, as well as the null space and eigenvalue separation properties of the global parameterization.


Among the dimensionality reduction methods, one can find many examples of algorithms such as \cite{Donoho2003}, \cite{Zhang2002}, \cite{Yang2010}, \cite{Zhan2008}, which apply a local Principal Component Analysis (PCA) for the computation of the tangent space of the manifold like we do in this work. In other words, the tangent space is estimated by computing the eigenvectors of the covariance of the data matrix, where the data samples come from a set of neighbor points on the manifold. This step can be seen as an analysis of PCA under data perturbations, where the perturbation of the data is caused by the nonlinear geometry of the embedding, i.e., the deviation of the manifold samples from the tangent space as a result of nonzero curvature. The performance of Singular Value Decomposition (SVD) or PCA in case of stochastic perturbations is a well-studied topic. There are many results in the literature that examine the perturbation on the singular vectors of a data matrix in the presence of noise. The Davis-Kahan theorem \cite{Davis1970} is a classical result that examines how much the subspace spanned by the eigenvectors of a Hermitian matrix is rotated upon the perturbation of the matrix. The Wedin theorem \cite{Wedin1972} generalizes the analysis to non-Hermitian operators by bounding the angle between the estimated and true singular vectors in terms of the separation between the eigenvalues of the data matrix. A recent result in \cite{Vu2011} addresses the singular vector estimation problem under assumptions of random perturbation noise and low-rank matrix. Finally, the work in \cite{Faber1995} examines the bias of random measurement error on PCA and relates the bias to the SNR of the observed data. However, above studies do not involve the geometric structure of the data. There are also many studies that analyze the performance of PCA for a set of data generated by a specific model. For instance, the works such as \cite{Anderson1963}, \cite{Lawley1956}, \cite{Girshick1939} address the analysis of the eigenvalues and eigenvectors of the covariance matrix of some data conforming to a multivariate normal distribution. These works however do not specifically consider any manifold data model either. 

Only a few recent works have studied the relation between the PCA performance and the data geometry. The work in \cite{Singer2011} presents an interesting study that generalizes the idea of diffusion maps in dimensionality reduction \cite{Coifman2005} to vector diffusion maps, where the new vector diffusion distance involves the similarity between the tangent spaces on different manifold points. In their analysis, the authors also provide a soft bound for the deviation of the locally estimated tangent space at a reference point (using local PCA) from the true tangent space, for a probabilistic sampling of the manifold. In particular, it shows that, when the size $\varepsilon$ of the local area for tangent estimation is set to $\varepsilon = O(K^{-\frac{1}{m+2}})$ with $K$ being the number of samples on the \textit{whole} manifold, the deviation between the estimated and the true tangent space is typically of $O(\varepsilon^{3/2})$. This work however considers a global sampling from a compact manifold while we focus on the local manifold geometry. Finally, the accuracy of tangent space estimation from noisy manifold samples is analyzed in a work parallel to ours \cite{Kaslovsky2011}. The manifold is assumed to be embedded with exactly quadratic forms (similarly to \cite{Tyagi2011}) and the data consists of manifold samples corrupted with Gaussian noise. The work optimizes the number of samples (from a fixed sets of candidates) that is used for estimating the tangent space by considering the effect of noise and curvature on the accuracy of estimation. In particular, the optimal number of samples is selected as a trade-off between the error due to noise and the error caused by the curvature that respectively decreases and increases as the number of samples grows. This study however focuses on manifolds that are embedded with exactly quadratic forms and characterized with a subset of noisy samples given a priori. On the contrary, we are interested in more generic embeddings with arbitrary smooth functions and we aim at characterizing a sampling strategy in terms of the sampling width and density for noiseless manifold samples.


In our paper, we propose to characterize the local linearity of a manifold by studying the accuracy of the tangent space estimation from a local set of randomly selected manifold samples. We propose the following contributions. First, we determine a suitable upper bound on the neighborhood size within which random manifold sampling can be done. In the derivation of this bound, we consider the asymptotic case $K \rightarrow \infty$ so that the neighborhood size purely depends on the manifold geometry. In particular, our analysis depends on (i) the maximum principal curvature of the manifold and (ii) the deviation of the manifold from its second-order approximation. Our main results are stated precisely in Lemma \ref{lemma:md_width_cond} for the quadratic embedding case and in Lemma \ref{lemma:md_angle_asymp_smooth} for the more general smooth embedding case. They show the dependency of the neighborhood size on the correlation between the components in the second-order local approximation of the manifold. Second, we compute a bound on the \textit{minimum} number of  samples for accurate tangent space estimation, given that the sampling is performed randomly in a neighborhood whose size conforms with Lemmas \ref{lemma:md_width_cond} and \ref{lemma:md_angle_asymp_smooth}. We utilize recent results from random matrix theory \cite{Gittens2011}, \cite{Tropp2011} in our analysis. We state the precise expression for this bound on the number of samples in Lemma \ref{lemma:md_k_bound_eps}. Combining the two above results, we give a complete characterization of the local sampling conditions in the form of main theorems, namely Theorem \ref{thm:main_thm_md_quad} for the quadratic embedding case and Theorem \ref{thm:main_thm_md_smooth} for the more general smooth embedding case. We finally discuss potential applications of the new theoretical results proposed in this paper, in respectively manifold learning and manifold sampling problems. 

The rest of the paper is organized as follows. In Section \ref{sec:problem_setup}, we first define the notations used in the paper and then give a formal statement of the problem along with the assumptions made. For ease of readability, the main results of the paper are presented in Section \ref{sec:main_results}. We then present in Section \ref{sec:mDimSurfaces} a detailed analysis of the local sampling conditions for tangent space estimation at a reference point $P$ on $S$. In particular, Sections \ref{subsection:md_angle_bound_quad} and \ref{subsec:quad_sampl_compl} contain the sampling analysis for the case when the embedding is assumed to be exactly quadratic at $P$. In Sections \ref{subsection:md_angle_bound_smooth} and \ref{subsec:smooth_sampl_compl}, we analyze the more general scenario of $m$-dimensional smooth embeddings in $\mathbb{R}^n$. Section \ref{sec:simulation_results} presents simulation results on synthetically generated smooth manifolds. In Section \ref{sec:discuss_results}, we provide a discussion regarding the usage of our theoretical results in practical applications. Finally, in Section \ref{sec:manifold_conclusion}, we provide concluding remarks along with possible directions for future work.


\section{Problem Formulation} \label{sec:problem_setup}
\noindent In this section we first define the notations used in the paper. We then define the our manifold approximation framework. We finally state formally the problem of tangent space estimation that is studied in this paper.
%
\subsection{Notations} \label{subsec:notations}
Let $S \subset \mathbb{R}^n$ be a manifold and $P \in S$ be a reference point on the manifold where the local sampling analysis is performed. We denote the dimension of the manifold $S$ by $m$. The tangent space at $P \in S$ is represented by $\tanps$ and $\tanpps$ is used to denote the orthogonal complement of $\tanps$ in $\mathbb{R}^n$. The notation $\mathcal{C}^r$ is used for denoting $r$ times continuous differentiability. 

We denote the $\ell_p$-norm of a vector $\bar{x} \in \mathbb{R}^n$, $1 \leq p \leq \infty$, by $\norm{\bar{x}}_{p} := \left(\sum_{i=1}^{n} \abs{x_i}^p \right)^{1/p}$ and its $\ell_{\infty}$-norm by $\norm{\bar{x}}_{\infty} := \max_{i}\abs{x_i}$. The inner product between $\bar{x},\bar{y} \in \mathbb{R}^n$ is denoted by $\langle \bar{x},\bar{y} \rangle \ := \bar{x}^T \bar{y}$. Furthermore, we represent a canonical vector in $\mathbb{R}^n$ by $\bar{e}_j$ for $j=1,\dots,n$, where $\bar{e}_j$ has a $1$ at the $j^{\text{th}}$ position and $0$ at all other positions.

Given a matrix $X \in \mathbb{R}^{p \times q}$, we have by its (reduced) singular value decomposition (SVD) \cite{Golub1996} the factorization $X = U \Sigma V^T$ where $U \in \mathbb{R}^{p \times s}$ and $V \in \mathbb{R}^{q \times s}$ are the singular vector matrices with orthonormal columns. The dimension $s \leq \min(p,q)$ corresponds to the rank of $X$. The matrix $\Sigma = \text{diag}(\sigma_1(X),\dots,\sigma_s(X))$ is a diagonal matrix where $\sigma_1(X) \geq \dots \geq \sigma_s(X) > 0$ are the singular values of $X$. We denote the Frobenius norm of $X$ (the $\ell_2$-norm of its vector of singular values) by $\norm{X}_F := (\text{Tr}(X^T X))^{1/2}$ and its operator norm (the largest singular value) by $\norm{X}$. For any square matrix $X \in \mathbb{R}^{p \times p}$, we denote the trace by Tr$(X)$ and the determinant by $\text{det}(X)$. 

For a symmetric matrix $X \in \mathbb{R}^{p \times p}$,  $X = X^T$ we have the eigenvalue decomposition $X = U \Lambda U^T$. Here $\Lambda = \text{diag}(\lambda_1(X),\dots,\lambda_p(X))$ denotes the eigenvalue matrix with $\lambda_1(X) \geq \dots \geq \lambda_p(X)$ and $U \in \mathbb{R}^{p \times p}$ is a unitary matrix so that $UU^T = U^T U = I$. If $X$ is symmetric and positive semidefinite we then have $\lambda_i(X) \geq 0$ for $i=1,\dots, p$. We denote the spectral radius of a symmetric matrix $X$ by $\rho(X) = \max_{i}(\abs{\lambda_i(X)})$.

Throughout the paper, $\mathbb{E}[\cdot]$ is used for denoting the expectation and $\mathbb{P}(\cdot)$ for denoting the probability.

%

\subsection{Framework}\label{subsec:Framework} \noindent We consider an $m$-dimensional submanifold $S$ of $\mathbb{R}^n$ with a smooth embedding in $\matR^n$, $n \ \geq \ m+1$. Let $\neighp$ denote a $\varepsilon$-neighbourhood of $P$ for some $\varepsilon \ > \ 0$, where 
\begin{equation*}
\neighp \ = \ \set{M \in S: \ \norm{M - P}_{2} \ \leq \ \varepsilon}. 
\end{equation*}
The neighborhood of $P$ on $S$ is illustrated in Fig.~\ref{fig:illusNeighP}.
\begin{figure}[]
 \centering
  \includegraphics[scale=0.8]{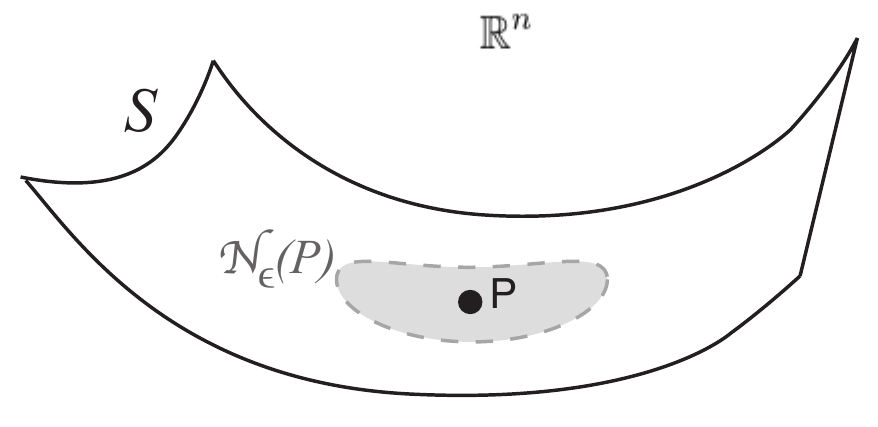}
  \caption{The manifold $S$ and the $\varepsilon$-neighborhood of a manifold point $P \in S$.}
  \label{fig:illusNeighP}
\end{figure}
In this work, as we represent points in $\neighp$ via tangent space parameterization using local functions $f_l : \tanps \rightarrow \matR$, we are interested in the mapping that orthogonally projects the manifold points in a neighborhood of $P$ to $\tanps$. In \cite{Niyogi2006}, Niyogi et al.~provide a characterization of the neighborhood of $P$ within which this mapping is one-to-one, through the condition number of the manifold. Therefore, there exists an $\varepsilon$ such that all points $M \in \neighp$ can be uniquely represented in the form
\begin{equation} \label{eq:tan_space_coor}
[\bar{x}^T  \  f_1(\bar{x}) \dots f_{n-m}(\bar{x})]^T. 
\end{equation} 
Here $\bar{x} \ = \ [x_1 \dots x_m]^T$ denotes the coordinates of the orthogonal projection of a point on \tanps. Note that, in \eqref{eq:tan_space_coor}, the coordinates are with respect to the point $P$ that is the reference point, i.e., the local origin. Furthermore, the tangent space $\tanps$ at $P$ can be represented as
\begin{equation*}
\tanps = \text{span}\set{\bar{e}_1,\dots,\bar{e}_m},
\end{equation*}
where $\bar{e}_j \in \mathbb{R}^{n}$ denote the canonical vectors. 

%
Now, we further assume the smoothness of the embedding to be $\mathcal{C}^r, \ r > 2$, implying that each
\begin{equation*}
f_l : \tanps \rightarrow \matR, \quad l=1,\dots,n-m,
\end{equation*}
is a $\mathcal{C}^r$-smooth function in the variables  $(x_1, \dots, x_m)$. Since $\nabla f_l(\bar{0}) = \bar{0}$ we have by the Taylor expansion of $f_l$ around the origin (i.e., $P$) the following identity:
\begin{equation} \label{eq:Cr_expan_f}
f_l(\bar{x}) = f_{q,l}(\bar{x}) + O(\norm{\bar{x}}_2^3); \quad l = 1,\dots,n-m
\end{equation}
where $f_{q,l}$ is a quadratic form. As a special case, we have a \textit{quadratic embedding} at $P$ when each $f_l$ is an exact quadratic form, i.e.,
\begin{equation*}
f_l(\cdot) = f_{q,l}(\cdot); \quad l=1,\dots,n-m. 
\end{equation*}

Consider the Hessian of $f_l$ at the local origin $P$, which is given as 
\begin{equation*}
\nabla^2 f_l(\bar{0}) = V_l \Lambda_l V_{l}^T,
\end{equation*}
where $\Lambda_l = \text{diag}(\mathcal{K}_{l,1},\mathcal{K}_{l,2},\dots, \mathcal{K}_{l,m})$. Here
$\mathcal{K}_{l,1},\mathcal{K}_{l,2},\dots, \mathcal{K}_{l,m}$ are the principal curvatures of the hypersurface
\[
\mathcal{S}_l=  \set{[x_1 \ \dots \ x_m \ f_l(x_1, \dots, x_m)]: \ [x_1 \ \dots \ x_m]^T \ \in \tanps} \subset \mathbb{R}^{m+1}
\]
defined by $f_l$. We then define the maximum principal curvature at $P$ as
\begin{equation*}
\kfmax := \mathcal{K}_{l^{\prime},j^{\prime}} \quad \text{where} \quad (l^{\prime},j^{\prime}) = \argmax{l,j} \abs{\mathcal{K}_{l,j}}.
\end{equation*}


%
%
We consider that the tangent space can be estimated from sample points in $\neighp$ through a PCA decomposition. More precisely, let us consider $K$ points $\set{P_i}_{i=1}^{K}$ sampled from $\neighp$. Let $M^{(K)}$ denote the local covariance matrix where 
\begin{equation*}
M^{(K)} = \sum_{i=1}^{K} \frac{1}{K} P_{i} P_{i}^T = U \Lambda U^T.
\end{equation*}
It is a common preprocessing step to subtract the empirical mean of the data from data samples in usual PCA. However, in our application, the linear subspace computed with PCA is an estimation of the tangent space, which is restricted to pass from the local origin $P$. Therefore, we omit the mean subtraction step in our analysis and assume that the principal components are computed with respect to the reference point $P$. The matrices $U$ and $\Lambda \in \mathbb{R}^{n}$ represent the eigenvector and eigenvalue matrices respectively of $M^{(K)}$ where
\begin{equation*}
U = [\bar{u}_1 \dots \bar{u}_m \dots \bar{u}_n] ; \quad \Lambda = \text{diag}(\lambda_1, \dots \lambda_m, \dots \lambda_n),
\end{equation*}
with the ordering $\lambda_1 \geq \cdots \geq \lambda_m \geq \cdots \geq \lambda_n$. The optimal $m$-dimensional linear subspace at $P$ in the least squares sense is then given by the span of the $m$ largest eigenvectors of $M^{(K)}$, i.e.,
\begin{equation*}
\widehat{T}_PS := \text{span}\set{\bar{u}_1,\dots,\bar{u}_m}.
\end{equation*}
The tangent space $T_P S$ and its estimation $\widehat{T}_PS$ are illustrated in Fig.~\ref{fig:illusTangSpace}.
\begin{figure}[]
 \centering
  \includegraphics[scale=0.7]{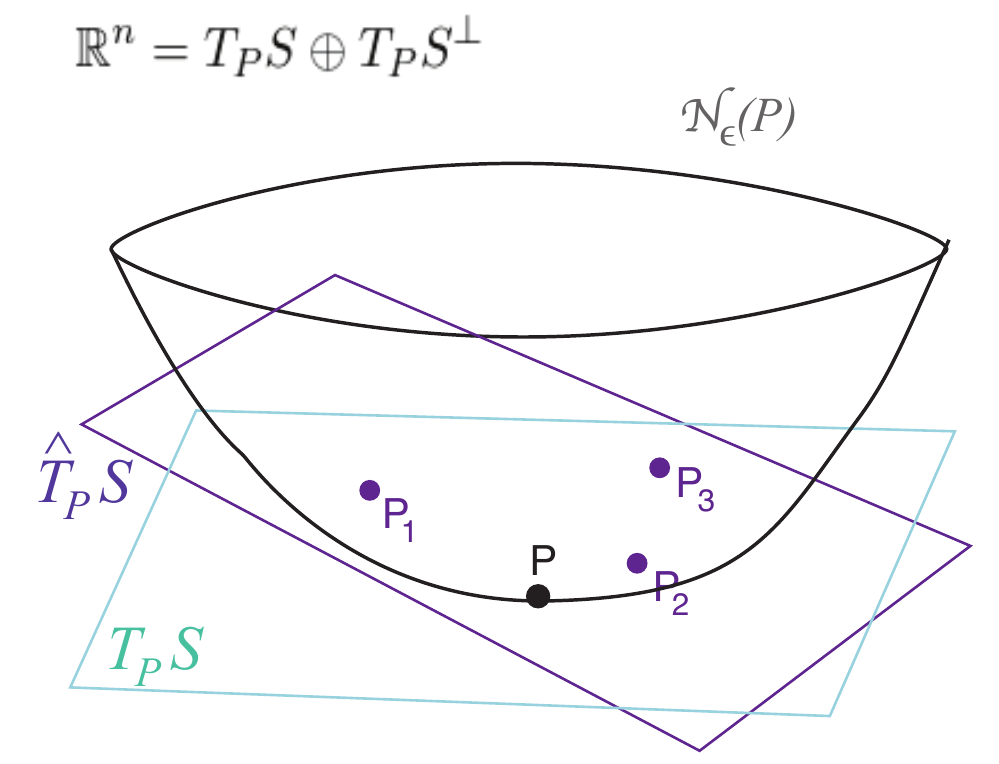}
  \caption{The true tangent space $T_PS$ and the estimated tangent space $\widehat{T}_PS$ at point $P$.}
  \label{fig:illusTangSpace}
\end{figure}
%
Finally, we characterize the accuracy of our estimation with the angle between the estimated and the true tangent spaces. The notion of `angle' between two linear subspaces as defined in ~\cite{Gunawan2005} is given in  Definition \ref{definition:subspace_angles}. 
\begin{definition} \label{definition:subspace_angles}
The angle $\angle A, B$ between two subspaces $A \ = \ \text{span}\{\bar{a}_1,\dots,\bar{a}_p\}$ and $B \ = \ \text{span}\{\bar{b}_1,\dots,\bar{b}_q\}$ of a Euclidean space $\matR^n$, where $\bar{a}_i$'s and $\bar{b}_i$'s are orthonormal vectors, is defined as
\begin{equation*}
\cos^2\angle A, B \ := \ \text{det}(W^T W),
\end{equation*}
where $[W^T]_{i,k}:= <\bar{a}_i,\bar{b}_k>$ is a $p \ \times \ q$ matrix, with $1 \leq p \leq q < \infty$.
\end{definition}
Observe that the definition can be applied to subspaces that are not necessarily of the same dimension. Geometrically speaking,
\begin{equation*}
\cos \angle A,B \ := \ \frac{V_1}{V_2}
\end{equation*}
where $V_1$ is the volume of the $p$-dimensional parallelepiped spanned by the projection of $\set{\bar{a}_1,\dots,\bar{a}_p}$ on $B$ and $V_2$ is the volume of the $p$-dimensional parallelepiped spanned by $\set{\bar{a}_1,\dots,\bar{a}_p}$. Therefore, in order to compare $\tanps$ and $\widehat{T}_PS$, one could also consider the distance between the respective projection matrices $EE^T$ and $U^{(m)} U^{(m)^T}$ through
\begin{equation} \label{eq:probsetup_proj_mat_comp}
\norm{EE^T -  U^{(m)} U^{(m)^T}}_F^2,
\end{equation}
where $E = [\bar{e}_1 \dots \bar{e}_m]$ and $U^{(m)} = [\bar{u}_1 \dots \bar{u}_m]$. Note that an upper bound on $\abs{\angle \widehat{T}_PS, \tanps}$ implies a corresponding upper bound on $\norm{EE^T -  U^{(m)} U^{(m)^T}}_F^2$. Finally, our choice of using Definition \ref{definition:subspace_angles} for estimating angles is motivated by the measure of the geometric deviation of $\widehat{T}_PS$ from $\tanps$. However, one could also work with the error criteria of Eq. \eqref{eq:probsetup_proj_mat_comp} with no change in the analysis and sampling conditions.\footnote{This is explained in more detail in Lemma \ref{lemma:md_angle_bound_cond} and Remark \ref{rem:equiv_proj_ang_defs}.}
%
\subsection{Problem statement}\label{subsec:ProblemState} Given the above settings, we want to describe the conditions on the manifold samples $\set{P_i}_{i=1}^{K}$ such that for a given error bound $\phi \in (0,\frac{\pi}{2})$ on the tangent space estimation,
\begin{equation*}
\abs{\angle \widehat{T}_PS, \tanps} \ < \ \phi \ < \frac{\displaystyle \pi}{\displaystyle 2}
\end{equation*}
is ensured. In particular, for a given error bound $\phi$, we would like to answer the following questions:
\begin{itemize}[leftmargin=*]
\item \textit{Question 1:} What would be a suitable upper bound on the \textit{sampling distance}; i.e., the distance of $P_i$ from $P$? In particular, for large embeddding dimensions $n$, what is the nature of the dependency of this bound on $n,m$ and $\kfmax$? \label{enum:prob_setup_cond_1}
\item \textit{Question 2:} Given that the points $\set{P_i}_{i=1}^{K}$ are sampled such that the sampling distance satisfies the sampling distance bound, what would be a suitable lower bound on the \textit{sampling density} $K$? In particular, for large embeddding dimensions $n$, what is the nature of the dependency of this bound on $n,m$ and $\kfmax$? \label{enum:prob_setup_cond_2}
\end{itemize}

In order to answer the above questions, we consider a random sampling framework where we assume that the coordinates of the orthogonal projections of manifold samples on $\tanps$ are distributed uniformly in the region $[-\sampwidth,\sampwidth]^m \in \tanps$. In other words, denoting the coordinates of the projection of $P_i$ on $\tanps$ by $\bar{x}_i = [x_1^{(i)} \dots x_m^{(i)}]$, we assume that 
\begin{equation*}
x_j^{(i)} \sim \mathcal{U}[-\sampwidth,\sampwidth] \quad \text{i.i.d.} \quad i=1,\dots,K; j=1,\dots,m 
\end{equation*}
where $\mathcal{U}$ denotes the uniform distribution. Therefore, we characterize the \textit{sampling distance} in Question 1 by the parameter $\sampwidth$, which we shall refer to as the sampling width in our analysis.\footnote{See Section \ref{sec:discuss_results} for a discussion on how the bound on $\nu$ relates to the distance in the ambient space.}


\section{Main results} \label{sec:main_results}
\noindent We summarize in this section the main results of the paper. We provide sampling conditions for tangent space estimation in two different cases; namely, quadratic embeddings and generic smooth embeddings. 
%
%
\subsection{Quadratic embedding at $P$} We first consider as a special case the scenario where the manifold $S$ has a quadratic embedding at $P$ in $\mathbb{R}^n$. We present the main sampling theorem in the form of Theorem \ref{thm:main_thm_md_quad} below. The main purpose of this result is to gain some intuition about the sampling conditions when the local functions $f_l$'s involved in the tangent space parametrization have a purely quadratic form and they are not `perturbed' by higher-order terms. We refer the reader to Section \ref{subsection:md_angle_bound_quad} for details regarding the proof and for a more rigorous analysis. 
\begin{theorem}[Quadratic manifold sampling] \label{thm:main_thm_md_quad}
Consider $\set{P_i}_{i=1}^{K}$ to be formed by sampling uniformly at random from the region $[-\sampwidth,\sampwidth]^m$ around $P$ in $\tanps$, i.e.,
\begin{equation*}
x_j^{(i)} \sim \mathcal{U}[-\sampwidth,\sampwidth] \quad \text{i.i.d., } \quad i=1,\dots,K,\, j=1,\dots,m.
\end{equation*}
Let $D \in \mathbb{R}^{n-m \times n-m}$ denote the local correlation matrix for the mappings $\set{f_{q,l}}_{l=1}^{n-m}$ such that
\begin{equation*}
[D]_{l,k} = \mathbb{E}[f_{q,l}(\bar{x}) f_{q,k}(\bar{x})]; \quad l, k = 1,\dots, n-m.
\end{equation*}
We then have the following sufficient sampling conditions that guarantee a bound on $\abs{\angle \widehat{T}_PS, \tanps}$.
For any $\tau \in (0,1)$, the choices
\begin{equation*}
\sampwidth = O(n^{-1/2} m^{-1} \abs{\kfmax}^{-1}) \text{  and} \quad K = O(\tau^{-2}m^2 \log n), \quad \text{ as} \ n \rightarrow \infty 
\end{equation*}
ensure that $\abs{\angle \widehat{T}_PS, \tanps} < \cos^{-1}\sqrt{(1 - \tau^2)^m}$ holds w.h.p.
\end{theorem}
\noindent \textbf{Interpretation of Theorem \ref{thm:main_thm_md_quad}.} We see that the bound on $\nu$ behaves as $O(n^{-1/2} m^{-1} \abs{\kfmax}^{-1})$, indicating that the sampling region needs to shrink with the increase in ambient space dimension. Furthermore, we observe that the bound on the sampling width depends linearly on the reciprocal $\abs{\kfmax}^{-1}$ of the maximum curvature. The decrease in $\nu$ with respect to the increase in $n$ and $\kfmax$ can be explained as follows. Assuming that $m$ is fixed, as $n$ increases, the number $(n-m)$ of normal components that increase the nonlinearity of the manifold increases, which causes the deviation of $S$ from the tangent space. Furthermore, the magnitude of each normal component increases with the increase in the curvatures associated with that normal direction. The sampling width $\nu$ must be selected sufficiently small such that the tangential components of the data have larger magnitudes than the normal components in PCA, in order to ensure the correct identification of the tangent space. Hence, the largest admissible value of the sampling width $\nu$ depends on the overall magnitude of the normal components. This is affected by both the codimension $(n-m)$ of $S$ and the curvature parameter $\kfmax$, which is used as a uniform bound on the individual curvatures in this work. In the derivation of these main results, the magnitude of the normal components is captured by the spectral norm of the $(n-m) \times (n-m)$ correlation matrix $D$, which increases with both the dimension of the ambient space, and the curvature.

Lastly, we remark that the approximation error term $\tau^2$ arises on account of finite sampling and can be interpreted as the variance error. In particular, provided that the sampling width $\nu$ is chosen to satisfy the appropriate bound, then we have that $\abs{\angle \widehat{T}_PS, \tanps} \rightarrow 0$ in the limit where $K \rightarrow \infty$. 

%
\subsection{Smooth embedding of $S$ in $\mathbb{R}^n$} We now present our main sampling theorem for the general case of smooth embeddings of $S$ in $\mathbb{R}^n$ in the form of Theorem \ref{thm:main_thm_md_smooth}. For details regarding the proof and for a rigorous analysis we refer the reader to Section \ref{subsection:md_angle_bound_smooth}. 

\begin{theorem}[Smooth manifold sampling] \label{thm:main_thm_md_smooth}
Consider $\set{P_i}_{i=1}^{K}$ to be formed by sampling uniformly at random from the region $[-\sampwidth,\sampwidth]^m$ in $\tanps$, i.e.,
\begin{equation*}
x_j^{(i)} \sim \mathcal{U}[-\sampwidth,\sampwidth] \quad \text{i.i.d., } \quad i=1,\dots,K, \, j=1,\dots,m.
\end{equation*}
Let $D \in \mathbb{R}^{n-m \times n-m}$ denote the local correlation matrix for the mappings $\set{f_{q,l}}_{l=1}^{n-m}$, such that
\begin{equation*}
[D]_{l,k} = \mathbb{E}[f_{q,l}(\bar{x}) f_{q,k}(\bar{x})],  \quad l, k = 1,\dots, n-m.
\end{equation*}
We then have the following sufficient sampling conditions that guarantee a bound on $\abs{\angle \widehat{T}_PS, \tanps}$.
For any $\tau \in (0,1)$, the choices
\begin{equation*}
\sampwidth = O(n^{-1/2} m^{-1} \abs{\kfmax}^{-1}) \text{  and} \quad K = O(\tau^{-2}m^2 \log n), \quad \text{ as} \ n \rightarrow \infty
\end{equation*}
ensure that $\abs{\angle \widehat{T}_PS, \tanps} < \cos^{-1}\sqrt{(1 - \tau^2 - O(n^{-1} m \abs{\kfmax}^{-4}))^m}$ holds w.h.p.
\end{theorem}
\noindent \textbf{Interpretation of Theorem \ref{thm:main_thm_md_smooth}.} As the manifold $S$ is now smoothly embedded in $\mathbb{R}^n$, the local functions $f_l$'s involved in the tangent space parametrization are arbitrary smooth functions of the form \eqref{eq:Cr_expan_f}. Hence, in this case, the deviation of the manifold from the tangent space is due to both the second-order terms $\set{f_{q,l}}_{l=1}^{n-m}$ and the higher-order terms in the Taylor series of $f_l$ (which are $O(\norm{\bar{x}}_2^3)$). Observe that the bound on the sampling density $K$ has a similar order of dependence on $n$, $m$, and $\abs{\kfmax}$ as in Theorem \ref{thm:main_thm_md_quad}. In the bounds on $\abs{\angle \widehat{T}_PS, \tanps}$, the error term represented by $\tau^2$ corresponds to the variance due to finite sampling as in the quadratic embedding case. On the other hand, we see that there is an additional error term of $O(n^{-1} m \abs{\kfmax}^{-4})$ for smooth embeddings, which does not exist in the quadratic embeddings. This term arises on account of the higher-order terms in the Taylor expansion of $f_l$, and can be interpreted as a bias term due to a nonzero sampling width $\sampwidth$. This bias goes to zero as $\sampwidth \rightarrow 0$. For smooth embeddings, in particular, for a fixed $\nu$ that is chosen to satisfy the appropriate bound, $| \angle \widehat{T}_PS, T_P S |$ approaches a constant bias error term as the variance error vanishes in the limit where $K \rightarrow \infty$. The reason why the tangent space estimation is non-biased for quadratic embeddings and biased for arbitrary smooth embeddings can be explained as follows. The normal components $f_{q,l}$ in quadratic embeddings have a symmetry around the origin; i.e., $f_{q,l}(\bar{x}) =f_{q,l}(-\bar{x}) $. However, for smooth embeddings we have $f_l(\bar{x}) \neq f_l(-\bar{x})$ in general because of the higher-order terms, which create an asymmetry in the orientation of the manifold points around the origin. This leads to a perturbation in the estimation of the tangent space with PCA and thus constitutes a bias.

\begin{remark}
In the above results, we have considered the general case where the functions $\set{f_{q,l}(\bar{x})}_{l=1}^{n-m}$ are all correlated; i.e., $D$ is a dense matrix with all nonzero entries. However, in a practical application, one can possibly encounter a setting where some of the $\big( f_{q,l}(\bar{x}), f_{q,k}(\bar{x}) \big)$ pairs are uncorrelated, or  weakly correlated. Therefore, $D$ may typically be a sparse matrix or most of its entries may be close to zero in certain applications. In this case, some of the restrictions on the sampling conditions can be loosened. In order to demonstrate to what extent the sampling conditions may change with respect to the correlation of the normal components of the data, we focus throughout the analysis on the two extreme cases where the random variables $\set{f_{q,l}(\bar{x})}_{l=1}^{n-m}$ are all correlated ($D$ is a dense matrix) and where $\set{f_{q,l}(\bar{x})}_{l=1}^{n-m}$ are mutually uncorrelated ($D$ is a diagonal matrix)\footnote{More details about the feasibility of $D$ being diagonal are given in Section \ref{subsection:md_angle_bound_quad}, Remark \ref{rem:corr_mat_D}.}. We now give an overview of how the above sampling conditions change when $D$ is diagonal. In this case, for quadratic embeddings, the sufficient sampling conditions that guarantee the angle bound given in Theorem \ref{thm:main_thm_md_quad} can be replaced by
\begin{equation*}
\sampwidth = O(m^{-1} \abs{\kfmax}^{-1}) \text{  and} \quad K = O(m n \tau^{-2} \log n).
\end{equation*}
Similarly, for a smooth embedding, the sampling conditions become
\begin{equation*}
\sampwidth = O(n^{-1/3} m^{-5/6} \abs{\kfmax}^{-1/3}) \text{  and} \quad K = O(n \log n),
\end{equation*}
which ensure that $\abs{\angle \widehat{T}_PS, \tanps} < \cos^{-1}\sqrt{(1 - O(n^{-1/3} m^{5/3} \abs{\kfmax}^{-4/3}))^m}$ w.h.p. as $\ n \rightarrow \infty$.

%
\begin{table}[!htp] 
\begin{tabular}[!htp]{| c | c | c |} 
\hline 
\small Correlation of $\set{f_{q,l}}_{l=1}^{n-m}$ & \small Smooth embedding & \small Quadratic embedding \\[1.5pt]
\hline
\small Correlated & \small $\sampwidth = O(n^{-1/2} m^{-1} \abs{\kfmax}^{-1})$ & \small $\sampwidth = O(n^{-1/2} m^{-1} \abs{\kfmax}^{-1})$ \\[1.5pt]
             & \small $K = O(\tau^{-2}m^2 \log n)$ & \small $K = O(\tau^{-2}m^2 \log n)$ \\[1.5pt]
\hline 
\small Uncorrelated & \small $\sampwidth = O(n^{-1/3} m^{-5/6} \abs{\kfmax}^{-1/3})$ & \small $\sampwidth = O(m^{-1} \abs{\kfmax}^{-1})$ \\[1.5pt]
 						 & \small $K = O(n \log n)$ & \small $K = O(m n \tau^{-2} \log n)$ \\[1.5pt]
\hline
\end{tabular}
\newline
\caption{\small Summary of the sampling conditions at a point $P$ on an $m$-dimensional manifold $S$ in $\mathbb{R}^n$ ($m < n$) such that $\abs{\angle \widehat{T}_PS, \tanps} < \cos^{-1}\sqrt{(1 - \tau^2)^m}$ for some $\tau \in (0,1)$. The two following cases are compared: (i) $S$ has a smooth embedding in $\mathbb{R}^n$, (ii) $S$ has a quadratic form in $\mathbb{R}^n$ w.r.t. the point $P$.}
\label{tab:comp_main_results}
\end{table}
These dependences are summarized in Table \ref{tab:comp_main_results}  in comparison with the results obtained for the general case where the normal components are correlated. One can observe the following. For both quadratic and smooth embeddings, the admissible sampling width can be chosen larger if the functions $\set{f_{q,l}(\bar{x})}_{l=1}^{n-m}$ are uncorrelated; i.e., the dependence of $\nu$ on $n$, $m$ and $\kfmax$ is loosened if $D$ is diagonal. This can be intuitively explained as follows. If $D$ is dense and there is a high correlation between two functions $f_{q,l}(\bar{x})$ and $f_{q,k}(\bar{x})$, the projection $[0  \, \dots \, 0  \, f_{q,l}(\bar{x}) \, 0 \dots  \, f_{q,k}(\bar{x}) \,  \dots 0] \subset T_PS^{\perp}$ of $S$ onto the normal plane generated by the normal directions $l$ and $k$ has a strong orientation along a certain direction on this normal plane. This creates a dominant normal direction along which the data is concentrated. Therefore, when the mutual correlations of the functions $\set{f_{q,l}(\bar{x})}_{l=1}^{n-m}$ are high, more dominant normal directions are generated. This creates a bigger challenge for the recovery of the tangent directions with PCA and necessitates the selection of a smaller sampling width. However, if $\set{f_{q,l}(\bar{x})}_{l=1}^{n-m}$ are uncorrelated or weakly correlated, there are less dominant normal directions. This results in looser constraints on the sampling width. This phenomenon reveals itself in the derivations through the spectral norm of the $D$ matrix. When the nonzero entries of the $D$ matrix are restricted to the diagonals, the spectral norm of the matrix grows at a slower rate with respect to the increase in $n$ and $\kfmax$, in comparison with the case where $D$ is dense.

Meanwhile, sampling from a wider region requires the selection of more samples; i.e., since $\nu$ is greater for the uncorrelated case, the required sampling density $K$ is higher when $D$ is diagonal. The bound on $K$ is logarithmic in $n$ for $D$ dense, and loglinear in $n$ for $D$ diagonal. This makes sense, since reducing the sampling width $\nu$ causes $\neighp$ to be more linear and hence loosens the restrictions on the number of samples required to achieve a given approximation bound on $\abs{\angle \widehat{T}_PS, \tanps}$. Of course, the case where $D$ is dense is general and hence the corresponding bounds for $\nu$ and $K$ can be used even if $D$ is in fact diagonal. If $D$ is diagonal, one can however afford to sample the manifold from a larger neighborhood. 

\end{remark}

\section{Analysis} \label{sec:mDimSurfaces}
\noindent We now present a detailed analysis of our local sampling results for smooth $m$-dimensional Riemannian manifolds in $\mathbb{R}^n$. To begin with, we first define the framework for our analysis by introducing the tangent space parameterization for points in $\neighp$.
%
\subsection{Framework for tangent space estimation} We discuss first the parametrization that we use in our analysis. Let $ [x_1 \ \dots \ x_m \ f_1(\bar{x}) \ \dots \ f_{n-m}(\bar{x})]^T$ be a point in $\neighp$ and let $\bar{x} = [x_1 \ \dots \ x_m]^T $ denote its orthogonal projection on $\tanps$. The region $\neighp$ can be represented in terms of $(n-m)$ hypersurfaces of dimension $m$ in $\mathbb{R}^{m+1}$, where the $l^{th}$ hypersurface is given by
\begin{equation*}
\mathcal{S}_l=  \set{[x_1 \ \dots \ x_m \ f_l(x_1, \dots, x_m)]: \ [x_1 \ \dots \ x_m]^T \ \in \tanps} \subset \mathbb{R}^{m+1}.
\end{equation*}
Due to the assumption that the embedding is $\mathcal{C}^r$, where $r>2$, the functions $f_l$ have the following form $\forall l \ = \ 1, \ \dots, \ n-m$
\begin{align*}
f_l(\bar{x}) &= f_l(\bar{0}) + \nabla f_l(\bar{0})^T\bar{x} + \frac{1}{2}\bar{x}^T \nabla^2 f_l(\bar{0})\bar{x} + R_l(\bar{\xi}_l), \\
&= \bar{0} + \bar{0}^T \bar{x} + \frac{1}{2}\bar{x}^T V_l \Lambda_l V_{l}^T \bar{x} + R_l(\bar{\xi}_l), \\
&= \frac{1}{2}\sum_{j=1}^{m}\left(<\bar{x},\bar{v}_{l,j}>^2 \mathcal{K}_{l,j}\right) + R_l(\bar{\xi}_l), \\
&= \frac{1}{2}\norm{\bar{x}}_2^2 \mathcal{K}_{l}(\bar{x}) + R_l(\bar{\xi}_l) = f_{q,l}(\bar{x}) +  R_l(\bar{\xi}_l), 
\end{align*}
where $\bar{\xi}_l \in (\bar{0},\bar{x})$ depends on $\bar{x}$. Here, $f_{q,l}$ denotes the quadratic approximation of $f_l$, and $R_l(\bar{\xi}_l) = O(\norm{\bar{x}}_2^3)$ represents the higher-order remainder terms in its Taylor series. The Hessian of $f_l$ at the origin is represented by $\nabla^2 f_l(\bar{0})$, and
\begin{equation*}
V_l = \begin{bmatrix}
\bar{v}_{l,1} & \bar{v}_{l,2} & \cdots & \bar{v}_{l,m}
\end{bmatrix}_{m \times m} ,  \qquad
\Lambda_l = \text{diag}(\mathcal{K}_{l,1},\mathcal{K}_{l,2},\dots, \mathcal{K}_{l,m})
\end{equation*}
denote respectively the eigenvector and eigenvalue matrices of $\nabla^2 f_l(\bar{0})$. Geometrically, $\mathcal{K}_{l}(\bar{x})$ represents the curvature at point $P$ of the geodesic curve on $\mathcal{S}_l$ from $P$ to $[\bar{x}^T \ f_1(\bar{x}) \ \dots \ f_{n-m}(\bar{x})]^T$, where
\begin{equation*}
\mathcal{K}_{l}(\bar{x}) = \sum_{j=1}^{m}\left(\frac{<\bar{x},\bar{v}_{l,j}>^2}{\norm{\bar{x}}_2^2} \mathcal{K}_{l,j}\right) .
\end{equation*}
Given the above setting, recall from Section \ref{subsec:ProblemState} that
\begin{equation*}
\kfmax = \mathcal{K}_{l^{\prime},j^{\prime}} \quad \text{where} \quad (l^{\prime},j^{\prime}) := \argmax{l,j} \abs{\mathcal{K}_{l,j}}.
\end{equation*}
Here, $\abs{ \kfmax}$ is the largest absolute value of the principal curvatures among the hypersurfaces $\mathcal{S}_l$, for $l=1, \dots , n-m $. We remark that the sampling conditions derived throughout our analysis capture the second-order properties of the manifold at $P$ in terms of the maximum curvature $\kfmax$.
%
Equipped with the above parametrization, we can now describe the estimation of the tangent space. Let us consider $K$ points,
\begin{equation*}
\set{P_i}_{i=1}^{K} \ = \ \set{[x^{(i)}_1 \ x^{(i)}_2 \ \dots \ x^{(i)}_m \ f_1(\bar{x}_{i}) \ \dots \ f_{n-m}(\bar{x}_{i})]^T}_{i=1}^{K}
\end{equation*}
in $\neighp$. Denoting  the coordinates of the orthogonal projection of $P_i$ on $\tanps$ as $\bar{x}_i = [x^{(i)}_1 \ x^{(i)}_2 \ \dots \ x^{(i)}_m]^T$, for $i=1,\dots,K$, we represent the points by the matrix $X^{(K)}$ as follows.
\begin{equation*}
X^{(K)} = \begin{bmatrix}
x^{(1)}_1 & \cdots & x^{(K)}_1 \\
\vdots &  & \vdots \\
x^{(1)}_m & \cdots & x^{(K)}_m \\
f_1(\bar{x}_1) &\cdots & f_1(\bar{x}_K) \\
\vdots & & \vdots \\
f_{n-m}(\bar{x}_1) &\cdots & f_{n-m}(\bar{x}_K) 
\end{bmatrix}
\end{equation*}
where each $f_l$ has the following form:
\begin{equation} \label{eq:exp_for_f}
f_l(\bar{x}) \ = \ \frac{1}{2}\sum_{j=1}^{m}<\bar{x},\bar{v}_{l,j}>^2 \kflj + O(\norm{\bar{x}}_2^3)\ ; \quad l = 1,\dots,n-m.
\end{equation}
The optimal $m$-dimensional linear subspace, in the least squares sense, passing through $P$ will be the one spanned by the eigenvectors corresponding to the $m$ largest eigenvalues of
\begin{equation*}
M^{(K)} = \frac{1}{K}XX^{T^{(K)}} = \begin{bmatrix}
A^{(K)} & B^{(K)} \\
B^{(K)^T} & D^{(K)} \\
\end{bmatrix} = U \Lambda U^T,
\end{equation*}
where the individual submatrices have the following form.
\begin{align*}
A^{(K)} = \begin{bmatrix}
\frac{1}{K}\sum_{i} (x_1^{(i)})^2 & \cdots & \frac{1}{K}\sum_{i} x_1^{(i)} x_m^{(i)} \\
\vdots & & \vdots \\
\frac{1}{K}\sum_{i} x_m^{(i)}x_1^{(i)} & \cdots & \frac{1}{K}\sum_{i} (x_m^{(i)})^2 \\
\end{bmatrix}, \quad B^{(K)} = \begin{bmatrix}
\frac{1}{K}\sum_{i} x_1^{(i)} f_1(\bar{x}_i) & \cdots & \frac{1}{K}\sum_{i} x_1^{(i)} f_{n-m}(\bar{x}_i) \\
\vdots & & \vdots \\
\frac{1}{K}\sum_{i} x_m^{(i)} f_1(\bar{x}_i) & \cdots & \frac{1}{K}\sum_{i} x_m^{(i)} f_{n-m}(\bar{x}_i) \\
\end{bmatrix}
\end{align*}
and
\begin{align*}
D^{(K)} = \begin{bmatrix}
\frac{1}{K}\sum_{i} f^2_1(\bar{x}_i) & \cdots & \frac{1}{K}\sum_{i} f_1(\bar{x}_i) f_{n-m}(\bar{x}_i) \\
\vdots & & \vdots \\
\frac{1}{K}\sum_{i} f_{n-m}(\bar{x}_i)f_1(\bar{x}_i) & \cdots & \frac{1}{K}\sum_{i} f^2_{n-m}(\bar{x}_i) \\
\end{bmatrix}.
\end{align*}
Furthermore,
\begin{equation*}
U = \begin{bmatrix}
\bar{u}_{1} & \cdots & \bar{u}_{m} & \bar{u}_{m+1} & \cdots & \bar{u}_{n}
\end{bmatrix} \,\, \text{ and   } \, \,
\Lambda = \text{diag}(\lambda_1, \lambda_2,\dots, \lambda_n)
\end{equation*}
are respectively the eigenvector and eigenvalue matrices of $\frac{1}{K}XX^{T^{(K)}}$ with $U^T U = UU^T = I_{n}$. Assume the ordering $\lambda_1 \geq \cdots \geq \lambda_m \geq \cdots \geq \lambda_n$. We then have
\begin{align*}
\widehat{T}_{P}S = \text{span}\{\bar{u}_1,\dots,\bar{u}_m\} \ \text{and} \quad \tanps \ = \text{span}\{\bar{e}_1,\dots,\bar{e}_m\},
\end{align*}
where $\set{\bar{e}_j}_{j=1}^{m}$ denote the first $m$ of the $n$ canonical basis vectors in $\mathbb{R}^n$. Now the angle between $\widehat{T}_{P}S$ and $\tanps$ as per Definition \ref{definition:subspace_angles} is given by
\begin{equation*}
\cos^2(\angle \widehat{T}_PS,\tanps) \ := \ \text{det}(W^T W),
\end{equation*}
where $[W^T]_{i,j} = <\bar{u}_i, \bar{e}_j>$ for $1 \leq i,j \leq m$. Let us denote the first $m$ columns of $U$ by $U^{(m)}$ where
\begin{equation*}
U^{(m)} = \begin{bmatrix}
U_1 \\ U_2
\end{bmatrix} ; \quad U_1 \in \mathbb{R}^{m \times m}, U_2 \in \mathbb{R}^{(n-m) \times m}.
\end{equation*}
Lemma ~\ref{lemma:md_angle_bound_cond} states the condition on $\norm{U_2}_F$ that guarantees a bound on $\abs{\angle \widehat{T}_PS,\tanps}$.
%
\begin{lemma} \label{lemma:md_angle_bound_cond}
Consider $K \geq m$ points in \neighp \ sampled such that $\norm{U_2}_F < \tau < 1$ for some $0 < \tau < 1$. Then,
\begin{equation*}
\abs{\angle \widehat{T}_PS,\tanps} < \cos^{-1}\sqrt{(1-\tau^2)^m}.
\end{equation*}
\end{lemma}
\begin{proof}
Clearly $U_1^TU_1$ + $U_2^TU_2$ = $I_{m \times m}$. Let $E \ = \ [\bar{e}_1 \ \dots \ \bar{e}_m]$. We have
\begin{align*}
W^T W = (U^{{(m)}^T}E)(E^TU^{(m)}) = U^T_{1}U_{1} = I_{m \times m} - U^T_2 U_2.
\end{align*}
Denoting the eigenvalues of $U^T_2 U_2$ as $\mu_1,\dots,\mu_m$, we observe that
\begin{align*}
Tr(U^T_2 U_2) \ = \ \norm{U_2}^2_F \ \geq \ \mu_{max},
\end{align*}
where $\mu_{max} = \max_{i=1 \dots m} \mu_i$. Using this result in conjunction with Definition \ref{definition:subspace_angles}, we arrive at the following inequality.
\begin{align*}
\cos^2(\angle \widehat{T}_PS,\tanps): = \det{(I_{m \times m} - U_2^TU_2)} = \prod_{i=1}^{m}(1-\mu_i) 
\geq (1 - \mu_{max})^m.
\end{align*}
Hence, the following bound on $\abs{\angle \widehat{T}_PS,\tanps}$ clearly holds if $\norm{U_2}_F < \tau < 1$.
\begin{align*}
\cos^2(\angle \widehat{T}_PS,\tanps) > (1-\tau^2)^m \Leftrightarrow \abs{\angle \widehat{T}_PS,\tanps} < \cos^{-1}\sqrt{(1-\tau^2)^m}.
\end{align*}
\end{proof}

\begin{remark} \label{rem:equiv_proj_ang_defs}
We remark here that one can compare the column spaces of $E$ and $U^{(m)}$ by also computing the difference between their projection matrices, i.e., $\norm{EE^T - U^{(m)}U^{(m)^T}}_F^2$. It is easily verifiable that
\begin{equation}
\norm{EE^T - U^{(m)}U^{(m)^T}}_F^2 = 2 \norm{U_2}_{F}^2. \label{eq:proj_mat_def}
\end{equation}
Hence when $\norm{U_2}_F < \tau < 1$, we have $\norm{EE^T - U^{(m)}U^{(m)^T}}_F < \sqrt{2}\tau$. The core of our analysis involves deriving sampling conditions which guarantee that $\norm{U_2}_F$ is suitably upper bounded. Hence one can interchangeably use Eq. \eqref{eq:proj_mat_def} or the notion of angle in Definition \ref{definition:subspace_angles} to compare $\tanps$ and $\widehat{T}_PS$ with no change in the analysis and the sampling conditions derived later on. The only change would be in the expression for the error bound where instead of $\cos^{-1}\sqrt{(1 - \tau^2)^m}$ one would have the error term $\sqrt{2}\tau$. Our choice of using Definition \ref{definition:subspace_angles} is purely motivated by our objective of measuring the deviation of $\widehat{T}_PS$ from $\tanps$ in a geometric way.
\end{remark}
Finally, we note that, if the manifold $S$ is a linear subspace of $\mathbb{R}^n$, then we have $B^{(K)} = 0$ and $D^{(K)} = 0$ implying $U_2$ to be trivially equal to zero. In other words, we have $\angle \widehat{T}_PS, \tanps = 0$ for any $K \geq m$. However, when $S$ is a more general manifold, then its nonlinearity manifests itself in the form of error arising due to the local mappings $\set{f_l}_{l=1}^{n-m}$. Hence, in order to obtain a good locally linear approximation of $\tanps$, one intuitively expects that the points in $\neighp$ are sampled sufficiently close to $P$. In particular, one might wonder how far from $P$  points can be sampled and also how many points need to be sampled in order to achieve a good approximation guarantee on $\abs{\angle \widehat{T}_PS, \tanps}$. We now proceed to rigorously analyze these two questions in the following sections. 

%
\subsection{Accuracy of tangent space estimation for quadratic embeddings} \label{subsection:md_angle_bound_quad}
We assume first that $\neighp$ is representable in terms of quadratic forms at the reference point $P$. In other words, for any point $[x_1 \dots x_m \ f_1(\bar{x}) \dots f_{n-m}(\bar{x})]$ in $\neighp$, we have 
\begin{equation*}
f_l(\bar{x}) \ = \ f_{q,l}(\bar{x}), \quad l=1,\dots,n-m,
\end{equation*}
where $f_{q,l}(\cdot)$ denotes the second order approximation of $f_{l}(\cdot)$. 

We consider the points $\set{P_i}_{i=1}^{K}$ to be formed by sampling independently and uniformly at random in \tanps \ such that 
\begin{equation*}
x_j^{(i)} \ \sim \ \mathcal{U} [-\sampwidth,\sampwidth] \ \text{i.i.d., } \qquad \forall i = 1,\dots, K \quad \text{and} \quad j \ = \ 1,\dots,m.
\end{equation*}
To begin with, Lemma ~\ref{lemma:md_width_cond} states precisely the condition on $\sampwidth$ which guarantees that $\widehat{T}_PS = \tanps$ in the limit where $K \rightarrow \infty$.
\begin{lemma} \label{lemma:md_width_cond}
\noindent As $K \rightarrow \infty$, $[M^{(K)}]_{i,j} \ \rightarrow [M]_{i,j}$ a.s. for every $1 \ \leq \ i,j \ \leq \ n$, where
\begin{equation*}
M = \begin{bmatrix}
\frac{\displaystyle \sampwidth^2}{\displaystyle 3} I_{m \times m} & 0_{m \times (n-m)} \\
0_{(n-m) \times m} & D_{(n-m) \times (n-m)} \\
\end{bmatrix} \ , \ [D]_{l,k} = \mathbb{E}[f_{l}(\bar{x}) f_{k}(\bar{x})] = \mathbb{E}[f_{q,l}(\bar{x}) f_{q,k}(\bar{x})],
\end{equation*}
and $l, k = 1,\dots,n-m$. Furthermore, the following holds.

\begin{enumerate}[leftmargin=*]
\item Let $D$ be dense, i.e., let $\set{f_{q,l}}_{l=1}^{n-m}$ be correlated. Then, if the sampling width satisfies
\begin{equation*}
\sampwidth < \sqrt{\frac{\displaystyle 60}{\displaystyle m(n-m)(5m+4)\abs{\kfmax}^2}},
\end{equation*}
it holds that $\mathbb{P}(\abs{\angle \widehat{T}_PS,\tanps} > 0) \rightarrow 0$ as $K \ \rightarrow \ \infty$.

\item Let $D$ be diagonal, i.e., let $\set{f_{q,l}}_{l=1}^{n-m}$ be uncorrelated. Then, if the sampling width satisfies
\begin{equation*}
\sampwidth < \sqrt{\frac{\displaystyle 60}{\displaystyle m(5m+4)\abs{\kfmax}^2}},
\end{equation*}
it holds that $\mathbb{P}(\abs{\angle \widehat{T}_PS,\tanps} > 0) \rightarrow 0$ as $K \ \rightarrow \ \infty$.
\end{enumerate}
\end{lemma}

\noindent The proof of Lemma ~\ref{lemma:md_width_cond} is presented in Appendix \ref{appendix:proof_md_width_cond}. The main idea here is to observe that the eigenspace corresponding to the eigenvalue $\frac{\sampwidth^2}{3}$ is equal to the span of ${\set{\bar{e}_1,\dots,\bar{e}_m}}$, which is the same as $\tanps$. Hence, the condition on the sampling width follows from the requirement that the noiseless spectra associated with $\frac{\sampwidth^2}{3}I_m$ is separated from the noisy spectra associated with $D$ arising on account of the manifold's curvature at $P$. In other words, 
\begin{equation*}
\frac{\sampwidth^2}{3} > \rho(D)
\end{equation*}
where $\rho(D)$ denotes the spectral radius of $D$.
\noindent For the sake of brevity, let us denote the bound on $\sampwidth$ by
\begin{equation*}
\sampwidth_{\text{bound,quad}} = 1/\sqrt{3RL} 
\end{equation*}
where
\begin{equation*}
L = \frac{m(5m+4)\abs{\kfmax}^2}{180} \, \, \text{and } \, \, R = \left\{
\begin{array}{rl}
(n-m), & \text{ if D is dense}\\
1, & \text{ if D is diagonal}.
\end{array} \right .
\end{equation*}
Therefore, $\sampwidth_{\text{bound,quad}}$ depends on the structure of $D$.


\begin{remark} \label{rem:corr_mat_D}
The case where $D$ is dense is general. Therefore, the derived condition on $\sampwidth$ can be used even if $D$ is actually diagonal. Moreover, if $D$ is diagonal, then we see that the condition on $\sampwidth$ is considerably less restrictive. We note here that $\set{f_{q,l}}_{l=1}^{n-m}$ will typically be correlated if $m$ is fixed and $n$ is allowed to increase to a large value. This arises due to the requirement
\begin{equation*}
\mathbb{E}[f_{q,l}(\bar{x}) f_{q,k}(\bar{x})] = 0, \text{with } l,k = 1,\dots,n-m, \quad l \neq k,
\end{equation*}
where a large value of $n$ and a small value of $m$ result in more equations than degrees of freedom. Hence, in order to have $D$ diagonal, the manifold dimension $m$ needs to be sufficiently large. In the case that the correlation matrix $D$ is sparse, the sufficient condition on $\sampwidth$ lies in between the two bounds stated in Lemma \ref{lemma:md_width_cond}.
\end{remark}

We want now to find a lower bound on the number of samples $K$ which guarantees that the deviation $\abs{\angle \widehat{T}_PS,\tanps}$ is suitably upper bounded with high probability. Hence, we first derive a bound on $K$ that guarantees some tail bounds on the eigenvalues of the submatrices of $M^{(K)}$. This bound is precisely stated in the following Lemma.
%
\begin{lemma} \label{lemma:md_k_bound_eps}
Let the sampling width be chosen such that $\sampwidth < \sampwidth_{\text{bound,quad}} = 1/\sqrt{3RL}$. Let $s_1 \in (0,1)$, $ s_2 > e$, $s_3 > 0$ and $0 < p_1,p_2,p_3 < 1$ denote fixed constants. We define
\begin{align*}
K^{(1)}_{bound} &= \frac{6 R_M}{(1-s_1)^2} \log \left((n-m+1)/p_1\right) ,\\
K^{(2)}_{bound} &= \frac{R_D}{s_2 R L} \frac{\log((n-m)/p_2)}{\log(s_2/e)}, \\
K^{(3)}_{bound} &= \frac{\sampwidth^6 R_{\sigma} + \frac{R_B \sampwidth^3 s_3}{3}}{s_3^2/2} \log(n/p_3),
\end{align*}
where
\begin{align*}
R_M = m + \frac{1}{4}(n-m)m^2 \sampwidth^2 \abs{\kfmax}^2, \quad R_D = \frac{1}{4}(n-m) m^2 \abs{\kfmax}^2, \\
R_{\sigma} = \frac{m^2 \abs{\kfmax}^2}{12} \max \set{(n-m), \frac{R(5m+4)}{15}}, \text{and } R_B = \frac{1}{2} m^{3/2}\sqrt{n-m} \abs{\kfmax}.
\end{align*}
Then, let $K_{bound} \ = \max \{K^{(1)}_{bound}, \, K^{(2)}_{bound}, \, K^{(3)}_{bound} \}$. If the number of samples $K$ satisfies $K > K_{bound}$, then the following bounds hold true with probability at least $1-p_1-p_2-p_3$:
\begin{itemize}
\item (i) $\lambda_m(M^{(K)}) > s_1 \frac{\sampwidth^2}{3}$, 
\item (ii) $\rho(D^{(K)}) < s_2 \rho(D)$,  
\item (iii) $\norm{B^{(K)}} < s_3$.
\end{itemize}
\end{lemma}

The proof of Lemma ~\ref{lemma:md_k_bound_eps} is presented in Appendix ~\ref{appendix:proof_md_k_bound_eps}. 
The lemma defines a sufficient bound on the sampling density, which in turn guarantees probabilistic bounds on the spectral norms of the perturbation matrices $B^{(K)}$ and $D^{(K)}$. Our proof builds on the recent results \cite{Gittens2011}, \cite{Tropp2011}, which give tail bounds on the eigenvalues of sums of independent random matrices. 

We have seen earlier that, if $\sampwidth$ is chosen to satisfy the appropriate bound on the sampling width, then $\mathbb{P}(\abs{\angle \widehat{T}_PS,\tanps} > 0) \rightarrow 0$ as $K \ \rightarrow \ \infty$. We now employ Lemma \ref{lemma:md_k_bound_eps} to show that, if $\sampwidth < c \ \sampwidth_{\text{bound,quad}}$ for some $c < e^{-1/2}$, and if $s_3 > 0$ is suitably upper bounded, then for $K > K_{bound}$, we have that $\abs{\angle \widehat{T}_PS,\tanps}$ is bounded from above with high probability. This is stated precisely in Theorem ~\ref{theorem:md_angle_bound_prob}.
%
\begin{theorem} \label{theorem:md_angle_bound_prob}
Consider $K$ points randomly sampled in $\neighp$ such that their projections to $\tanps$ are independent and uniform in the region $[-\sampwidth,\sampwidth]^m$, i.e.,
\begin{equation*}
x^{(i)}_j \ \sim \ U[-\sampwidth,\sampwidth] \quad \text{i.i.d.}, \quad i=1,\dots,K, \quad j=1,\dots,m.
\end{equation*}
Under the notation defined earlier, assume that, for some fixed $s_1 \in (0,1)$ and $s_2 > e$, 
\begin{equation*}
\sampwidth < \sqrt{\frac{s_1}{s_2}}\sampwidth_{\text{bound,quad}} = \sqrt{\frac{s_1}{3 s_2 R L}}.
\end{equation*}
Then, consider that, for some $\tau \in (0,1)$, 
\begin{equation}
0 < s_3 < \frac{(s_1\frac{\sampwidth^2}{3} - s_2 RL \sampwidth^4)\tau}{\sqrt{m}} = {s_3}_{\text{bound,quad}}.
\label{eq:s3_boundquad}
\end{equation}
Finally, let $0 < p_1,p_2,p_3 < 1$. Then, if $K \ > \ K_{\text{bound}}$, we have that 
\begin{equation*}
\mathbb{P}(\abs{\angle \widehat{T}_PS,\tanps} < \cos^{-1}\sqrt{(1-\tau^2)^m}) > 1-p_1-p_2-p_3.
\end{equation*}
\end{theorem}
The proof is presented in Appendix \ref{appendix:proof_md_angle_samp_quad}. In the proof, we use the conditions derived in Lemma \ref{lemma:md_k_bound_eps} in order to obtain eigenvalue separation conditions for the correlation matrix constructed with  a finite sampling, which are then used to derive a bound on $\| U_2 \|_{F}$. Note that the error term $\tau$ is the variance error arising due to finite sampling; it goes to zero as $K \rightarrow \infty$.
%
\subsection{Analysis of the bounds for quadratic embedding} \label{subsec:quad_sampl_compl}
We now proceed to analyze the dependence of the sampling parameters $\sampwidth$ and $K$ on the manifold dimension $m$, the maximum curvature $\abs{\kfmax}$ and the ambient space dimension $n$, where we assume that $n$ is high (i.e., $n \rightarrow \infty$). We analyze this by considering two separate cases based on the structure of the matrix $D$.
\begin{enumerate}[leftmargin=*]
\item \textbf{$D$ is dense.}
When no assumption is made on the structure of $D$, we have $\sampwidth_{\text{bound,quad}} = O(n^{-1/2} m^{-1} \abs{\kfmax}^{-1})$ as $n \rightarrow \infty$. Using this, one obtains from the corresponding expression of $s_3$ that
\begin{equation*}
s_3 = O\left(n^{-1} m^{-5/2} \abs{\kfmax}^{-2} \tau \right).
\end{equation*}
For a given probability of success, we derive the sampling bound complexity as follows.
\begin{align*}
K_{bound}^{(1)} &= O(R_M \log n) = O(m \log n), \\
K_{bound}^{(2)} &= O\left(\frac{R_D}{RL} \log n\right) = O(\log n), \\
K_{bound}^{(3)} &= O\left(\displaystyle \frac{\sampwidth_{\text{bound,quad}}^6 R_{\sigma} + \frac{\displaystyle R_B \sampwidth_{\text{bound,quad}}^3 s_3}{3}}{s_3^2} \log n \right) = O(\tau^{-2} m^2  \log n).
\end{align*}
Thus $K_{bound} = O(\tau^{-2} m^2 \log n)$ as $n \rightarrow \infty$. Here, the number of samples is seen to depend quadratically on the manifold dimension and logarithmically on the ambient space dimension. Note that the dependency on $n$ is milder in this case in comparison to the case where $D$ is diagonal, which is due to the fact that the condition on the sampling width $\nu$ is stricter when $D$ is dense.
\item \textbf{$D$ is diagonal.}
We first observe that $\sampwidth_{\text{bound,quad}}$ is independent of the dimension $n$. In particular, we have $\sampwidth_{\text{bound,quad}} = O\left( m^{-1} \abs{\kfmax}^{-1}\right)$. Using this, one obtains from the corresponding expression of $s_3$ that
\begin{equation*}
s_3 = O\left(m^{-5/2} \abs{\kfmax}^{-2} \tau\right).
\end{equation*}
For a given probability of success, we derive the sampling bound complexity as follows.
\begin{align*}
K_{bound}^{(1)} &= O(R_M \log n) = O(n \log n), \\
K_{bound}^{(2)} &= O\left(\frac{R_D}{RL} \log n \right) = O(n \log n), \\
K_{bound}^{(3)} &= O\left(\frac{\sampwidth_{\text{bound,quad}}^6 R_{\sigma} + \frac{R_B \sampwidth_{\text{bound,quad}}^3 s_3}{3}}{s_3^2/2} \log n \right) = O(m n \tau^{-2}\log n).
\end{align*}
Thus $K_{bound} = O(m n \tau^{-2} \log n)$ as $n \rightarrow \infty$. Hence, the number of samples has a linear dependence on the intrinsic dimension of the manifold and a loglinear dependence on the ambient space dimension. 
\end{enumerate}
%
\subsection{Accuracy of tangent space estimation for smooth embeddings} \label{subsection:md_angle_bound_smooth}
In the previous section, we have assumed that $\neighp$ can be represented with quadratic forms. We now consider the more general scenario where the manifold is smoothly embedded in $\mathbb{R}^n$. In particular, we assume the smoothness class $\mathcal{C}^r$, where $r > 2$, in order to be able to study the influence of the local curvature of the manifold. Under this assumption we have that $f_l(\cdot) = f_{q,l}(\cdot) + R_{l}(\cdot)$ for $ \ l=1,\dots n-m$, where $f_{q,l}(\cdot)$ denotes the second order approximation of $f_{l}(\cdot)$ and $R_{l}(\cdot)$ denotes the higher-order terms. As each $f_{l}(\cdot)$ is defined over a compact domain, $R_{l}(\cdot)$ is bounded, i.e., $\abs{R_{l}(\cdot)} = O\norm{\cdot}_2^3$ for all $l=1,\dots,n-m$. Hence,

\begin{equation*}
\abs{R_{l}(\cdot)} < C_{s,l}\norm{\cdot}_2^3 \quad l=1,\dots,n-m,
\end{equation*}
where the constant $C_{s,l} > 0$ depends on the magnitude of the third order derivatives of $f_l$ in $\neighp$. We denote

\begin{equation*}
\smoothconst = \max_{l} C_{s,l}, \quad l=1,\dots,n-m.
\end{equation*}
Let us again consider the points $\set{P_i}_{i=1}^{K}$ to be formed by sampling independently and uniformly at random in \tanps \ such that 
\begin{equation*}
x_j^{(i)} \ \sim \ \mathcal{U} [-\sampwidth,\sampwidth] \ \text{i.i.d.} \qquad \forall i = 1,\dots, K \quad \text{and} \quad j \ = \ 1,\dots,m.
\end{equation*}
Using the same notation as before, when $f_l(\cdot) = f_{q,l}(\cdot) + R_{l}(\cdot), \ l=1,\dots n-m$, we arrive at the following form for the local covariance matrix $M^{(K)}$.

\begin{equation*}
M^{(K)} = M_q^{(K)} + \Delta^{(K)}
\end{equation*}
where
\begin{equation*}
M_q^{(K)} = \begin{bmatrix}
A^{(K)} & B^{(K)} \\
B^{(K)^T} & D^{(K)} \\
\end{bmatrix}
\end{equation*}
is the covariance matrix considered in the previous section, with the submatrices $B^{(K)}$ and $D^{(K)}$ representing the error on account of the manifold's curvature at $P$. Furthermore,

\begin{equation*}
\Delta^{(K)} = \begin{bmatrix}
0 & B_1^{(K)} \\
B_1^{(K)^T} & D_1^{(K)} \\
\end{bmatrix}
\end{equation*}
is an additional error term arising on account of the higher-order Taylor series terms of the mappings $\set{f_l}_{l=1}^{n-m}$ with

\begin{equation*}
B_1^{(K)} = \begin{bmatrix}
\frac{1}{K}\sum_i x_1^{(i)}R_{1}(\bar{\xi}_{1,i}) & \dots & \frac{1}{K}\sum_i x_1^{(i)}R_{{n-m}}(\bar{\xi}_{n-m,i}) \\
\vdots & & \vdots \\
\frac{1}{K}\sum_i x_m^{(i)}R_{1}(\bar{\xi}_{1,i}) & \dots & \frac{1}{K}\sum_i x_m^{(i)}R_{{n-m}}(\bar{\xi}_{n-m,i}) 
\end{bmatrix}
\end{equation*}
and

\begin{equation*}
[D_1]^{(K)}_{l,k} = \left\{
\begin{array}{rl}
\frac{1}{K}\sum_i (R_{{l}}(\bar{\xi}_{l,i})R_{{k}}(\bar{\xi}_{k,i}) + R_{{l}}(\bar{\xi}_{l,i})f_{q,k}(\bar{x}_i) 
 + R_{{k}}(\bar{\xi}_{k,i})f_{q,l}(\bar{x}_i)) & \text{if} \quad l \neq k \\
\frac{1}{K}\sum_i (R_{{l}}(\bar{\xi}_{l,i})^2 + 2f_{q,l}(\bar{x}_i)R_{{l}}(\bar{\xi}_{l,i})) & \text{if} \quad l = k.
\end{array} \right.
\end{equation*}
To begin with, let us define 
\begin{equation*}
\delta(\sampwidth) = \smoothconst m^{3/2} \sampwidth^3
\end{equation*}
where the factor $m^{3/2}$ appears since $\norm{\bar{x}_i}_2^3 < m^{3/2} \sampwidth^3$ for $i=1,\dots,K$. We then observe that each entry of $\Delta^{(K)}$ can be bounded as 

\begin{enumerate}[leftmargin=*]
\item $\abs{[B_1^{(K)}]_{j,l}} < \sampwidth\delta(\sampwidth)$ for $j=1,\dots m$; $l=1,\dots,n-m$

\item $\abs{[D_1^{(K)}]_{l,k}} < \delta(\sampwidth)^2 + \delta(\sampwidth) m \sampwidth^2\abs{\kfmax}$ for $l, k = 1,\dots,n-m$
\end{enumerate}
where we used the fact that $\abs{x_j^{(i)}} < \nu$ and $\abs{R_l(\cdot)} < \delta(\nu) $ for obtaining (1); and $\abs{f_{q,l}(\cdot)} < \frac{1}{2}m \sampwidth^2 \abs{\kfmax}$ for obtaining (2). Using the bounds on the entries, we obtain the following bounds on $\norm{B_1^{(K)}}_F$ and $\norm{D_1^{(K)}}_F$ respectively:

\begin{align*}
\norm{B_1^{(K)}}_{F} \ &< \ \sqrt{m(n-m)}\sampwidth\delta(\sampwidth) = \sqrt{m(n-m)} \smoothconst m^{3/2} \sampwidth^4 = \norm{B_1}_{F,bound}, \\
\norm{D_1^{(K)}}_{F} \ &< \ (n-m)(\delta(\sampwidth)^2 + \delta(\sampwidth) m \sampwidth^2\abs{\kfmax}) \\ 
&= (n-m)\smoothconst m^{5/2}\sampwidth^5(\smoothconst m^{1/2}\sampwidth + \abs{\kfmax}) = \norm{D_1}_{F,bound}.
\end{align*}

\noindent Now let us denote 
\begin{equation*}
B_1=\mathbb{E} [ B_1^{(K)} ], \quad D_1=\mathbb{E} [ D_1^{(K)} ], \text{ and } \Delta =\mathbb{E} [ \Delta^{(K)} ] .
\end{equation*}
Due to the ergodicity of the sampling process, we have $B_1= \lim_{K \rightarrow \infty} B_1^{(K)} $, $D_1= \lim_{K \rightarrow \infty} D_1^{(K)} $, and $\Delta = \lim_{K \rightarrow \infty} \Delta^{(K)} $. Since the bounds on the entries of the perturbation submatrices $B_1^{(K)}$ and $D_1^{(K)}$ hold for all $K$, they are also valid for the entries of $B_1$ and $D_1$. Therefore, we get $\norm{B_1}_F  <  \norm{B_1}_{F,bound}$ and $\norm{D_1}_F Ê<  \norm{D_1}_{F,bound}$.

We first consider the case $K = \infty$, where we obtain $M = \lim_{K \rightarrow \infty} M^{(K)} = M_q + \Delta$. It was shown in Lemma \ref{lemma:md_width_cond} that

\begin{equation*}
M_q = \begin{bmatrix}
\frac{\displaystyle \sampwidth^2}{\displaystyle 3} I_{m \times m} & 0_{m \times (n-m)} \\
0_{(n-m) \times m} & D_{(n-m) \times (n-m)} \\
\end{bmatrix},
\end{equation*}
where $[D]_{l,k} = \mathbb{E}[f_{q,l}(\bar{x}) f_{q,k}(\bar{x})]$, for $l,k = 1,\dots,n-m$. Given that $M_q$ is now `perturbed' by $\Delta$, Lemma \ref{lemma:md_angle_asymp_smooth} states the conditions on the sampling width $\sampwidth$ that guarantee an upper bound on the angle between $\widehat{T}_PS$ and $\tanps$.
%
\begin{lemma} \label{lemma:md_angle_asymp_smooth}
Let the sampling width satisfy
\begin{equation*}
\sampwidth < \frac{1}{[3((\beta_2 + RL) + \beta_3 \alpha + \beta_4\alpha^2)]^{1/2}}
\end{equation*}
where $\beta_2 = 4 \smoothconst m^{2}(n-m)^{1/2}$, $\beta_3 = 2(n-m)\smoothconst m^{5/2}\abs{\kfmax}$, $\beta_4 = 2(n-m)m^3 \smoothconst^2$ and
\begin{equation*}
\alpha = \min\set{(3(\beta_2 + RL))^{-1/2},\, (3\beta_3)^{-1/3}, \, (3\beta_4)^{-1/4}}.
\end{equation*}
Then, as $K \rightarrow \infty$,
\begin{equation*}
\mathbb{P}\left(\abs{\angle \widehat{T}_PS,\tanps} > \cos^{-1}\sqrt{(1 - m\sigma_{\infty}^2)^m}\right) \rightarrow 0
\end{equation*}
where
\begin{align*}
\sigma_{\infty} = \frac{\norm{B_1}_{F,bound}}{\frac{\sampwidth^2}{3} - RL \sampwidth^4 - 2(\norm{B_1}_{F,bound} + \norm{D_1}_{F,bound})}.
\end{align*}
\end{lemma}
The proof is presented in Appendix \ref{appendix:proof_md_angle_asym_smooth}. The main idea here is to ensure that the spectrum associated with $\frac{\sampwidth^2}{3}I_m$ is separated from the spectrum of the error arising due to the following factors:
\begin{enumerate}[leftmargin=*]
\item The curvature components $\set{f_{q,l}}_{l=1}^{n-m}$ which give rise to the correlation matrix $D$. 
\item The higher-order Taylor series terms of the smooth mappings $\set{f_{l}}_{l=1}^{n-m}$ giving rise to the perturbation matrix $\Delta$.
\end{enumerate}
Observe that, unlike in the case where $f_l$'s are quadratic forms, the deviation $\abs{\angle \widehat{T}_PS,\tanps}$ now does not converge to zero but to a residual bound $\cos^{-1}\sqrt{(1-m \sigma_{\infty}^2)^m}$. This is on account of the additional error associated with the matrix $\Delta$, which now perturbs the covariance matrix $M_q$. The error term $m \sigma_{\infty}^2$ can be interpreted as the bias error arising due to the nonzero sampling width. In particular, it is easily verifiable that 
\begin{equation*}
\sigma_{\infty} \rightarrow 0 \quad \text{as} \quad \sampwidth \rightarrow 0.
\end{equation*}

Also note that, had the $f_l$'s been quadratic forms, we would have $\smoothconst = 0$ resulting in $\sigma_{\infty} = 0$. This gives us the result obtained in Lemma \ref{lemma:md_width_cond}.
\begin{remark} \label{rmk:lemma_md_angle_asymp_smooth}
We remark here that the choice of the sampling width $\sampwidth$ satisfying the condition in Lemma \ref{lemma:md_angle_asymp_smooth} actually ensures the following bound:
\begin{equation*}
\frac{\sampwidth^2}{3} - RL \sampwidth^4 > 4\norm{B_1}_{F,bound} + 2\norm{D_1}_{F,bound}.
\end{equation*} 
Using this implication in the expression for $\sigma_{\infty}$, we obtain the trivial bound $\sigma_{\infty} < 1/2$. Furthermore, the residual angle bound term $\sigma_{\infty}$ can be made arbitrarily small by choosing a sufficiently small $\sampwidth$. 
\end{remark}

We now proceed to the case $K < \infty$. Theorem \ref{thm:md_angle_bound_smooth}, which is the main sampling theorem of this section, states the sufficient conditions on the sampling width $\sampwidth$ and the number of samples $K$, such that the deviation $\abs{\angle \widehat{T}_PS, \tanps}$ is suitably upper bounded with high probability.
%
\begin{theorem} \label{thm:md_angle_bound_smooth}
Consider $K$ points randomly sampled in $\neighp$ such that their projections to $\tanps$ are independent and uniform in the region $[-\sampwidth,\sampwidth]^m$, i.e.,
\begin{equation*}
x_j^{(i)} \sim \mathcal{U}[-\sampwidth,\sampwidth] \quad \text{i.i.d.,} \quad i=1,\dots,K, \text{ and } j=1,\dots,m.
\end{equation*}
Under the notation defined earlier, assume that for some fixed $s_1 \in (0,1)$ and $s_2 > e$, the following holds:
\begin{equation*}
\sampwidth < \left(\frac{s_1}{3[(\beta_2 + s_2 RL) + \beta_3\alpha + \beta_4\alpha^2]}\right)^{1/2} = \sampwidth_{\text{bound,smooth}}.
\end{equation*} 
Then, for some $\tau \in (0,1)$, let $s_3 > 0$ be chosen such that
\begin{equation}
s_3 < [(s_1\frac{\sampwidth^2}{3} - s_2RL \sampwidth^4) - 2(\norm{B_1}_{F,bound} + \norm{D_1}_{F,bound})]\left(\frac{\tau^2}{m} + \sigma_{f}^2 \right)^{1/2} - \norm{B_1}_{F,bound}
\, = \, {s_3}_{\text{bound,smooth}}
\label{eq:s3_boundsmooth}
\end{equation}
where
\begin{equation*}
\sigma_{f} = \frac{\norm{B_1}_{F,bound}}{(s_1\frac{\sampwidth^2}{3} - s_2RL \sampwidth^4) - 2(\norm{B_1}_{F,bound} + \norm{D_1}_{F,bound})}.
\end{equation*}
Finally, let $0 < p_1,p_2,p_3 < 1$. Then, if the number of samples satisfies $K > K_{bound}$, where $K_{bound}$ is as derived in Lemma \ref{lemma:md_k_bound_eps}, the following holds true
\begin{equation*}
\mathbb{P}(\abs{\angle \widehat{T}_PS, \tanps} < \cos^{-1}\sqrt{(1 - \tau^2 - m\sigma_{f}^2)^m}) > 1 - p_1 - p_2 - p_3.
\end{equation*}
\end{theorem}
\noindent The proof is presented in Appendix \ref{appendix:proof_md_angle_samp_smooth} and is built on the results of Lemma \ref{lemma:md_angle_asymp_smooth}. It uses similar ideas to those in the proof of Theorem \ref{theorem:md_angle_bound_prob}; however, the additional perturbation matrix $\Delta$ also plays a role in the derived bounds. Note that the approximation error consists of two terms - the variance term $\tau$ due to finite sampling and the bias term $\sigma_{f}$ arising due to the nonzero sampling width $\sampwidth$. 
\begin{remark}
We again remark here that the choice of the sampling width $\sampwidth$ in Theorem \ref{thm:md_angle_bound_smooth} ensures the following bound
\begin{equation*}
(s_1\frac{\sampwidth^2}{3} - s_2RL \sampwidth^4) > 4\norm{B_1}_{F,bound} + 2\norm{D_1}_{F,bound}.
\end{equation*} 
Hence, it follows trivially that $\sigma_{f} < 1/2$. Furthermore, $\sigma_{f}$ can be reduced appropriately by choosing a suitably downscaled sampling width. In particular, as shown in Section \ref{subsec:smooth_sampl_compl}, in the worst case $\sigma_{f}$ is $O(n^{-1/6})$ for large $n$. This implies that the effect of the bias error $m \sigma_{f}^2$ on the overall performance is typically mild.
\end{remark}
%
\subsection{Analysis of the bounds for smooth embedding} \label{subsec:smooth_sampl_compl}
We now analyze the complexity of the parameters involved in the sampling analysis for large $n$ (i.e., $n \rightarrow \infty$). We first observe the following for the perturbation terms $\beta_2, \beta_3$ and $\beta_4$ in Lemma \ref{lemma:md_angle_asymp_smooth}:

\begin{equation}
\beta_2 = O(n^{1/2} m^2), \ \beta_3 = O(n m^{5/2} \abs{\kfmax}), \ \beta_4 = O(n m^3). \label{eq:pert_term_orders}
\end{equation}
We now proceed by analyzing two different scenarios depending on the structure of the matrix $D$.
\begin{enumerate}[leftmargin=*]
\item \textbf{D is dense.} When the positive semidefinite matrix $D$ is dense, the sampling width has complexity
\begin{equation*}
\sampwidth_{\text{bound,smooth}} = O(n^{-1/2} m^{-1} \abs{\kfmax}^{-1}).
\end{equation*}
This is similar to the bound in the quadratic embedding case. Next, we have the following complexities for the perturbation bounds $\norm{B_1}_{F,bound}$ and $\norm{D_1}_{F,bound}$.
\begin{align*}
\norm{B_1}_{F,bound} &= O(n^{1/2} m^2 \sampwidth_{\text{bound,smooth}}^4) = O(n^{-3/2} m^{-2} \abs{\kfmax}^{-4}) \\
\norm{D_1}_{F,bound} &= O(n m^{5/2} \sampwidth_{\text{bound,smooth}}^5 \abs{\kfmax}) = O(n^{-3/2} m^{-5/2} \abs{\kfmax}^{-4})
\end{align*}
Finally, we obtain the following complexity for the `residual' angle bound term $\sigma_{f}$.
\begin{equation*}
\sigma_{f} = O\left(\frac{\norm{B_1}_{F,bound}}{\sampwidth_{\text{bound,smooth}}^2}\right) = O(n^{1/2} m^2 \sampwidth_{\text{bound,smooth}}^2) =  O(n^{-1/2} \abs{\kfmax}^{-2}).
\end{equation*} 
We observe that $\sigma_{f}$ decays at a faster rate compared to the case where $D$ is diagonal. The order of the dependency of the sampling width bound on $n$ is higher in this case, which in turn implies a stricter bound on the high-order terms in the Taylor expansion. Finally, since the dependency of $\sampwidth_{\text{bound,smooth}}$ on $n$, $m$, and $\abs{\kfmax}$ is the same as that of $\sampwidth_{\text{bound,quad}}$, we obtain the same sampling complexity as in the quadratic embedding case: 
\begin{equation*}
K_{bound} = O(\tau^{-2} m^2 \log n) \quad \text{as} \quad n \rightarrow \infty.	
\end{equation*}
Therefore, the number of samples has a quadratic dependence on the manifold dimension and a logarithmic dependence on the ambient space dimension.
\item \textbf{D is diagonal.} It can be verified that the bound on the sampling width has the complexity
\begin{equation*}
\sampwidth_{\text{bound,smooth}} = O((\beta_3\alpha)^{-1/2}) = O(n^{-1/3} m^{-5/6} \abs{\kfmax}^{-1/3}). 
\end{equation*} 
This is in contrast to the quadratic embedding case, where we have seen that $\sampwidth_{\text{bound,quad}}$ is independent of $n$. Moving on, we have the following complexities for the perturbation bounds $\norm{B_1}_{F,bound}$ and $\norm{D_1}_{F,bound}$:
\begin{align*}
\norm{B_1}_{F,bound} &= O(n^{1/2} m^2 \sampwidth_{\text{bound,smooth}}^4) = O(n^{-5/6} m^{-4/3} \abs{\kfmax}^{-4/3}), \\
\norm{D_1}_{F,bound} &= O(n m^{5/2} \sampwidth_{\text{bound,smooth}}^5 \abs{\kfmax}) = O(n^{-2/3} m^{-5/3} \abs{\kfmax}^{-2/3}).
\end{align*}
From these orders of dependency, we arrive at the following complexity for the `residual' angle bound term $\sigma_{f}$.
\begin{equation*}
\sigma_{f} = O\left(\frac{\norm{B_1}_{F,bound}}{\sampwidth_{\text{bound,smooth}}^2}\right) = O(n^{1/2} m^2 \sampwidth_{\text{bound,smooth}}^2) = O(n^{-1/6} m^{1/3} \abs{\kfmax}^{-2/3}).
\end{equation*}
Observe that $\sigma_{f}$ decays with the increase in $n$, which is due to the decrease in $\sampwidth_{\text{bound,smooth}}$. Notice also that $\sigma_{f}$ gets smaller when the maximum curvature $\abs{\kfmax}$ increases. This can be intuitively explained as follows. It has been discussed in Section \ref{sec:main_results} that, as $f_{q,l}(\bar{x}) = f_{q,l}(-\bar{x})$, the normal components $f_{q,l}$ of the second-order terms constitute a symmetry around the origin. Meanwhile, the higher-order terms do not have such a symmetry in general, causing a bias on the tangent space estimation with PCA. The residual angle bound $\sigma_{f}$ is associated with this bias resulting from the asymmetry of the normal components. As $\abs{\kfmax}$ increases, the second-order terms get more significant compared to the higher-order terms, which strengthens the symmetry of the manifold and reduces the bias term. We now study the sampling complexity by analyzing $K_{bound}^{(1)},K_{bound}^{(2)}$ and $K_{bound}^{(3)}$ separately. It can be easily verified that
\begin{align*}
K_{bound}^{(1)} &= O(R_M \log n) = O(n^{1/3} m^{1/3} \abs{\kfmax}^{4/3} \log n), \\
\text{and} \quad K_{bound}^{(2)} &= O\left(\frac{R_D}{RL} \log n\right) = O(n \log n).
\end{align*}
Furthermore, by observing that 
\begin{equation} \label{eq:s3_order_diag}
s_3 = O(\sampwidth_{\text{bound,smooth}}^2 m^{-1/2} \tau) = O(n^{-2/3} m^{-13/6} \abs{\kfmax}^{-2/3} \tau),
\end{equation}
we have
\begin{align*}
K_{bound}^{(3)} &= O\left(\frac{\displaystyle \sampwidth_{\text{bound,smooth}}^6 R_{\sigma} + \frac{R_B \sampwidth_{\text{bound,smooth}}^3 s_3}{3}}{s_3^2} \log n \right) \\
&= O(n^{1/3} m^{4/3} \abs{\kfmax}^{4/3} \log n).
\end{align*}
Since $K_{bound} = \max \set{K_{bound}^{(1)}, K_{bound}^{(2)}, K_{bound}^{(3)}}$, we have
\begin{equation*}
K_{bound} = O(K_{bound}^{(2)}) = O(n \log n) \quad \text{as} \quad n \rightarrow \infty. 
\end{equation*} 
Hence, the number of samples has a loglinear dependence on the ambient space dimension. In fact, although $K_{bound}$ is chosen according to $K_{bound}^{(2)}$, this choice of $K_{bound}$ implies that $s_3$ can be chosen up to the order $s_3 = O(n^{-1} m^{-3/2})$ by retaining the value of $K_{bound}$. Comparing this with the expression in \eqref{eq:s3_order_diag}, we see that the variance term is $\tau^2 = O(n^{-2/3} m^{4/3})$. Meanwhile, the bias term is $m \sigma_f^2 =O( n^{-1/3} m^{5/3} |\mathcal{K}_{max}|^{-4/3}) $, which shows that the decay of the variance term with the increase in $n$ is faster than the decay of the bias term. As we consider that $n$ is large, we can neglect the variance term in comparison with the bias term.   This gives 
\begin{equation*}
\tau^2 + m \sigma_{f}^2 = O(m \sigma_{f}^2) = O(n^{-1/3} m^{5/3} \abs{\kfmax}^{-4/3}).
\end{equation*}
The fact that the error resulting from finite sampling is negligible compared to the error due to the high-order Taylor terms when $D$ is diagonal can be interpreted as follows. When $D$ is diagonal, the samples can be chosen from a relatively wide region. Then, since the sampling width is large, the error in the estimation of the tangent space caused by the asymmetry in the geometric structure of the manifold dominates the error caused by finite sampling.
\end{enumerate} 

%
\section{Experimental Results} \label{sec:simulation_results}
\noindent In this section we present experimental results for the empirical validation of the sampling conditions derived in the preceding sections. For the sake of brevity, we use the notation $\theta = \angle \widehat{T}_PS, \tanps$ to describe the angle between $\tanps$ and $\widehat{T}_PS$.
Recall from Section \ref{sec:mDimSurfaces} that, for any point $P$ lying on a smooth $m$-dimensional manifold $S$ in $\mathbb{R}^n$, the points lying in the neighborhood $\neighp$ of $P$ have the following representation:
\begin{equation*}
[\bar{x}^T \ f_1(\bar{x}) \ \dots \ f_{n-m}(\bar{x})]; \quad f_l : \tanps \rightarrow \mathbb{R}.
\end{equation*}

In the experiments, we study different manifold embeddings, where the functions $f_l$ have the following form:
\begin{enumerate}
\item \textit{Quadratic form:} $f_l = \frac{1}{2}\sum_{j=1}^{m} \mathcal{K}_{l,j} x_{j}^2$ 
\item \textit{Smooth mapping 1:} $f_l = 1-\exp\left(\frac{1}{2}\sum_{j=1}^{m} \mathcal{K}_{l,j} x_{j}^2\right)$
\item \textit{Smooth mapping 2:} $f_l = \sin\left(\frac{1}{2} \sum_{j=1}^{m} \mathcal{K}_{l,j} x_{j}^2 \right)$
\item \textit{Smooth mapping 3:} $f_l = \sum_{j=1}^{m} \left(\frac{1}{2}\mathcal{K}_{l,j} x_{j}^2 + a_{l,j} x_j^3 + b_{l,j} x_j^4 + c_{l,j} x_j^5 \right)$
\end{enumerate} 

In particular, we consider the mappings $\set{f_l}_{l=1}^{n-m}$ to be all of the same form. Furthermore, we focus on the general case where $\set{f_{q,l}}_{l=1}^{n-m}$ are correlated, or equivalently $D$ is dense, which is the most generic scenario. Then, for a given value of $\kfmax$, we select the principal curvatures $(\mathcal{K}_{l,1},\dots,\mathcal{K}_{l,m})$ randomly from the interval $[0,\abs{\kfmax}]^m$ and then randomly assign the same sign ($+$ or $-$) to the elements of $\set{\mathcal{K}_{l,j}}_{j=1}^{m}$. Furthermore for Smooth mapping 3, we select the coefficients $a_{l,j},b_{l,j}$ and $c_{l,j}$ randomly from the interval $[0,10]$ for $l=1,\dots,n-m$ and $j=1,\dots,m$.

We sample the points as explained in Section \ref{sec:mDimSurfaces}. We compute the tangent space with these samples points and compare the resulting estimation with the true tangent space by measuring the angle between both subspaces. Then we analyze the results from the perspective of the theoretical bounds on the width of the sampling regions and on the number of samples, which have been derived earlier in the paper.  In particular, we consider the sampling width to have the value $\nu = \gamma \, \nu_{\text{bound,quad}}$, where
\begin{equation*}
\nu_{\text{bound,quad}} := \sqrt{\frac{60}{m(n-m)(5m+4)\abs{\kfmax}^2}}.
\end{equation*}
The choice $\nu_{\text{bound,quad}}$ for the reference sampling width is due to the fact that it can be easily computed and it also provides a basis for comparing  smooth embeddings with quadratic embeddings, made possible by varying the scale parameter $\gamma$. 

In the first set of experiments, we examine the relation between the estimation error, i.e., the deviation $\abs{\theta}$, and the sampling density $K$. We fix $m = 5$, $\kfmax = 10$ and consider different values for the dimension of the ambient space, namely $n = 100,500,1000$. For each value of $(m,n,\kfmax)$, we choose the sampling width as a scaled version of theoretical bound, i.e., $\nu = \gamma \nu_{\text{bound,quad}}$. We estimate the tangent space with $K$ samples and compute the approximation error with respect to the true tangent subspace. The results, shown in Fig. \ref{fig:md_exp1} have been averaged over 25 random trials for each value of $K$, where $K$ is varied from 100 to 2000 in steps of 100.

We first show in Figures  \ref{fig:md_exp1_quad1}-\ref{fig:md_exp1_quad3} the results obtained for a quadratic embedding. We observe that for the choice $\nu \approx 1.2 \, \nu_{\text{bound,quad}}$, $\abs{\theta}$ decreases sharply towards $0^{\circ}$ with the increasing values of $K$. On the other hand, the choice $\nu \approx 4 \, \nu_{\text{bound,quad}}$ causes $\abs{\theta}$ to increase towards $90^{\circ}$. The results are similar across different values of $n$. Furthermore, since $D$ is dense and $m = 5$ in this experiment, Theorem \ref{thm:main_thm_md_quad} states that $K = O(\tau^{-2})$ and $\abs{\theta} = O(\cos^{-1} (1-\tau^{2})^{5/2})$. Therefore, the order of the dependence of $\abs{\theta}$ on $K$ is expected as $\abs{\theta} = O(\cos^{-1} (1-K^{-1})^{5/2})$. We can see that the plotted curves are in accordance with this theoretical result. We remark that, in these experiments, the true upper bound on $\nu$ appears to be within a factor $\gamma$ of $\nu_{\text{bound,quad}}$, where $\gamma$ takes a value between 1.2 and 4.

Figures \ref{fig:md_exp1_gauss1}-\ref{fig:md_exp1_gauss3}, \ref{fig:md_exp1_sin1}-\ref{fig:md_exp1_sin3} and \ref{fig:md_exp1_poly1}-\ref{fig:md_exp1_poly3} then show the experimental results for non-quadratic embeddings, in particular, for Smooth mappings $1,2$ and $3$ respectively. Interestingly, the variation of $\abs{\theta}$ with respect to $K$ for non-quadratic mappings is almost identical to those for quadratic forms.

The theoretical bounds stated in Theorems \ref{theorem:md_angle_bound_prob} and \ref{thm:md_angle_bound_smooth} are directly implementable for this experiment where the variation of $\abs{\theta}$ with $K$ is examined. Therefore, we now provide a comparison of the theoretical bounds and the empirical results for this setup. We obtain the theoretical variation of $\abs{\theta}$ with $K$ as follows. We first choose $s_1 = 0.5$ (as $0 < s_1 < 1$) and $s_2 = 2e$ (as $s_2 > e$). Then, in the quadratic embedding, we compute $\nu = \gamma \, \nu_{bound,quad}$, where $\gamma = c \, \sqrt{s_1/s_2}  $. Here, $c<1$ is a scale parameter guaranteeing that $\nu$ is strictly less than $\sqrt{s_1/s_2} \, \nu_{bound,quad}$. In the smooth embedding, we compute $\nu = \gamma \, \nu_{bound,quad}$ where $\gamma = c \ \nu_{bound,smooth} / \nu_{bound,quad}$ with $c<1$ so that $\nu$ is strictly less than $\nu_{bound,smooth}$. In both quadratic and smooth embeddings, we fix the values of the parameters $p_1$, $p_2$, $p_3$ to $0.01$ and vary $\tau$ from $0.01$ to $0.2$ in steps of $0.01$. For each value of $\tau$, we compute $s_3$ such that it is slightly smaller than ${s_3}_{\text{bound,quad}}$ in the quadratic embedding and ${s_3}_{\text{bound,smooth}}$ in the smooth embeddings, which are respectively given in (\ref{eq:s3_boundquad}) and (\ref{eq:s3_boundsmooth}). This gives a value of $K^{(3)}_{bound}$ and hence $K_{bound}$. The angle bound $\abs{\theta} = \cos^{-1}\sqrt{(1-\tau^2)^m}$ for the quadratic embedding is computed using Theorem \ref{theorem:md_angle_bound_prob} , and the angle bound $\cos^{-1}\sqrt{(1 - \tau^2 - m\sigma_{f}^2)^m}$ for the smooth embeddings is given by Theorem \ref{thm:md_angle_bound_smooth}. Evaluating the bounds at four different values of the $c$ parameter (0.2, 0.4, 0.6, 0.8 for the quadratic embedding, and 0.1, 0.2, 0.3, 0.4 for the smooth embeddings), we obtain four subplots showing the variation of $K_{bound}$ with $\abs{\theta}$, each for a different value of $\nu$. The results are given in Fig.~\ref{fig:quad_theo_emp_exp1} for the quadratic embedding and Figures \ref{fig:smooth1_theo_emp_exp1}, \ref{fig:smooth2_theo_emp_exp1} and \ref{fig:smooth3_theo_emp_exp1} for the smooth embeddings. The theoretical plots obtained for $n= 100, \, 500, \, 1000$ are displayed respectively in (a)-(c) in these figures. The experimental curves corresponding to these theoretical plots are then obtained by sampling the manifolds in the region $\nu = \gamma \, \nu_{bound,quad}$ for the same value of $\gamma$ as in the theoretical plots, which are shown in the plots (d)-(f) of Figures \ref{fig:quad_theo_emp_exp1}-\ref{fig:smooth3_theo_emp_exp1}.

The comparison of the theoretical bounds and the experimental plots given in Figures \ref{fig:quad_theo_emp_exp1}-\ref{fig:smooth3_theo_emp_exp1} shows the following. While the numerical values of the theoretical bounds obtained for the angle error are pessimistic in comparison with the experimental values, we see that the theoretical variation of $\abs{\theta}$ with $K$ matches well the experimental one in both quadratic and smooth embeddings. Therefore, the theoretical results provide a good prediction of the dependence of the angle error on the number of samples.
\begin{figure}[!htbp]
\centering
\subfloat[\small Quadratic form]{
\begin{minipage}[c]{0.28\linewidth}
\centering
\label{fig:md_exp1_quad1} 
\noindent \includegraphics[width=1.0\linewidth]{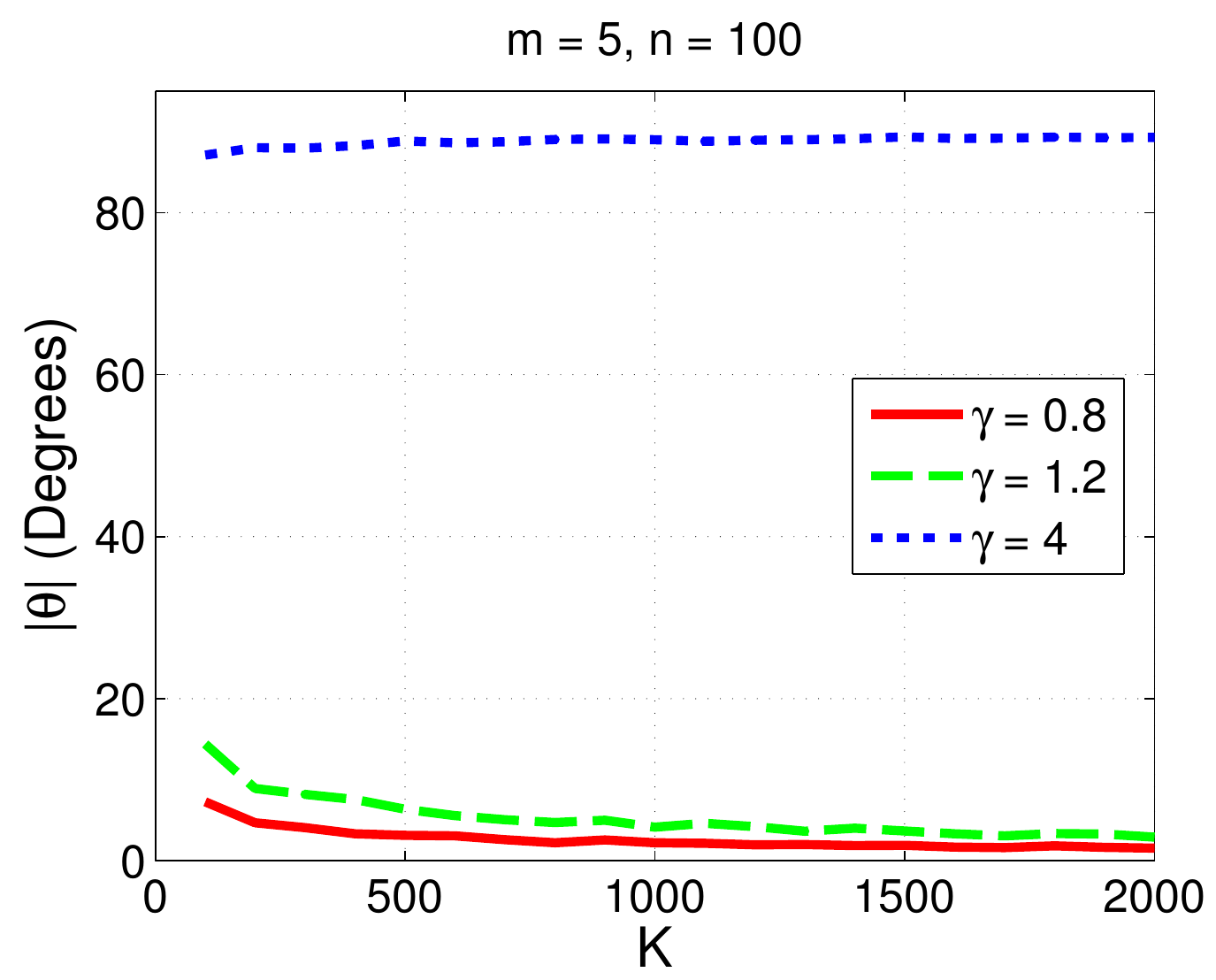}\end{minipage}}
\subfloat[\small Quadratic form]{
\begin{minipage}[c]{0.28\linewidth}
\centering
\label{fig:md_exp1_quad2} 
\noindent \includegraphics[width=1.0\linewidth]{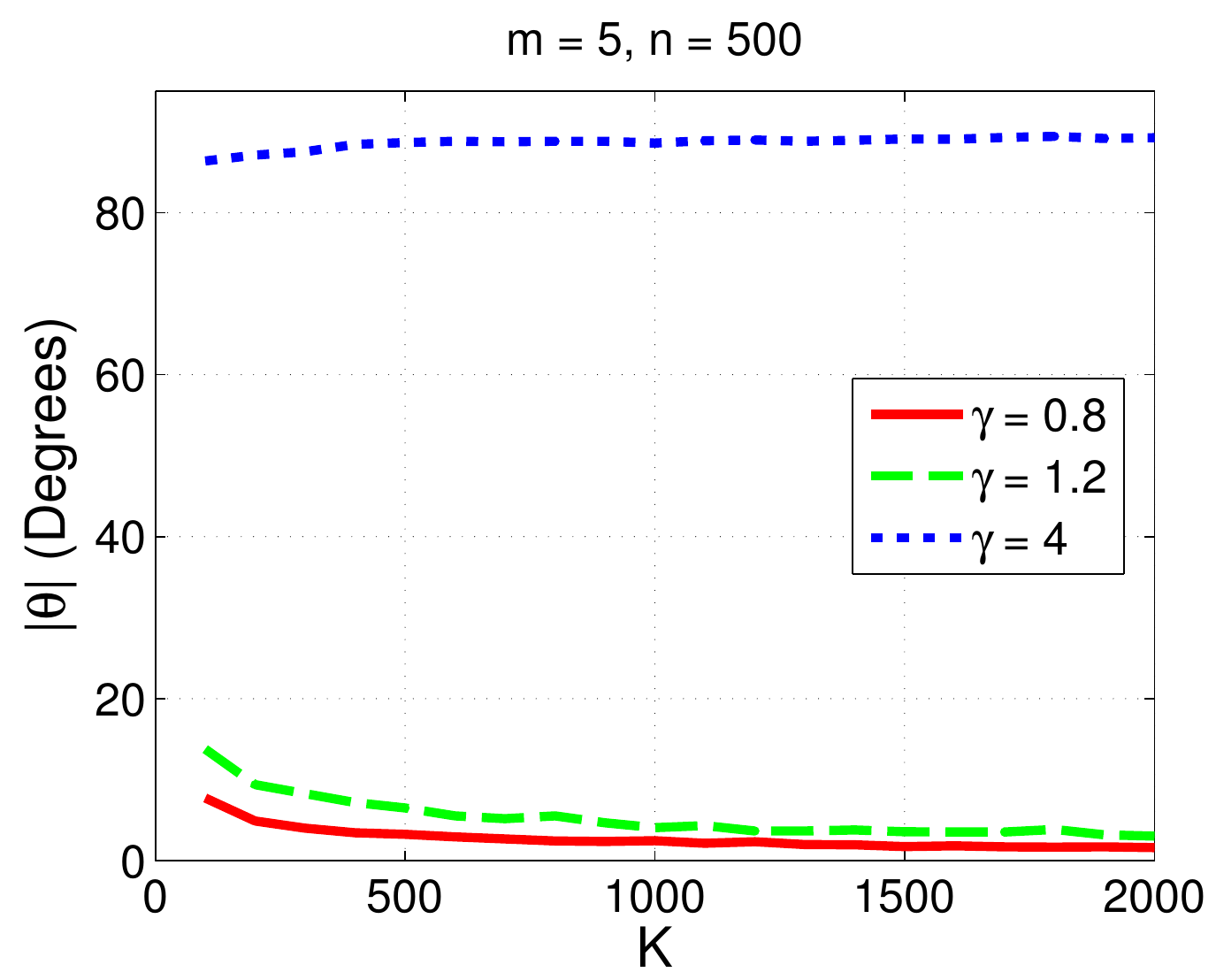}\end{minipage}}
\subfloat[\small Quadratic form]{
\begin{minipage}[c]{0.28\linewidth}
\centering
\label{fig:md_exp1_quad3} 
\noindent \includegraphics[width=1.0\linewidth]{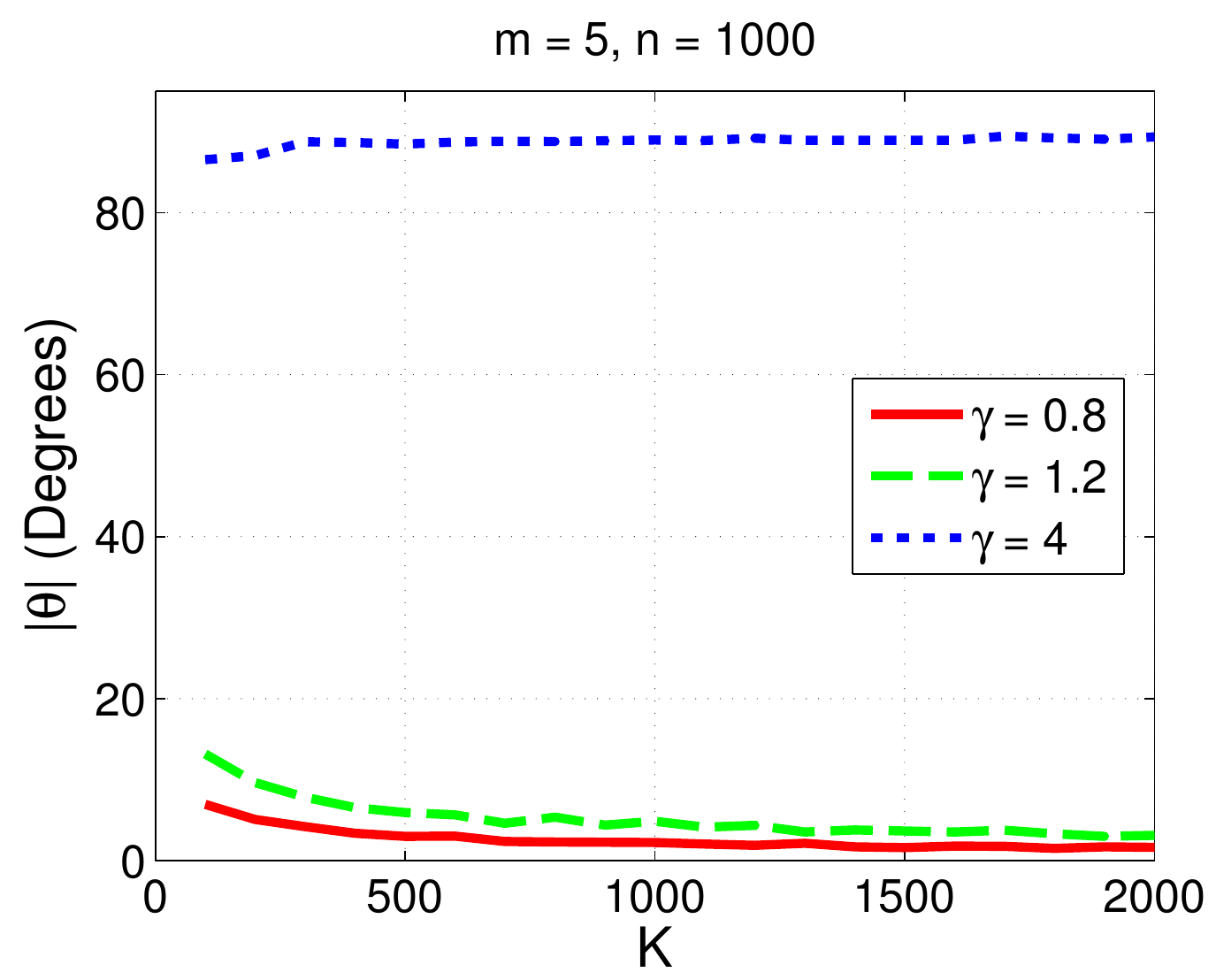}\end{minipage}}\\
\subfloat[Smooth mapping 1]{
\begin{minipage}[c]{0.28\linewidth}
\centering
\label{fig:md_exp1_gauss1} 
\includegraphics[width=1.0\linewidth]{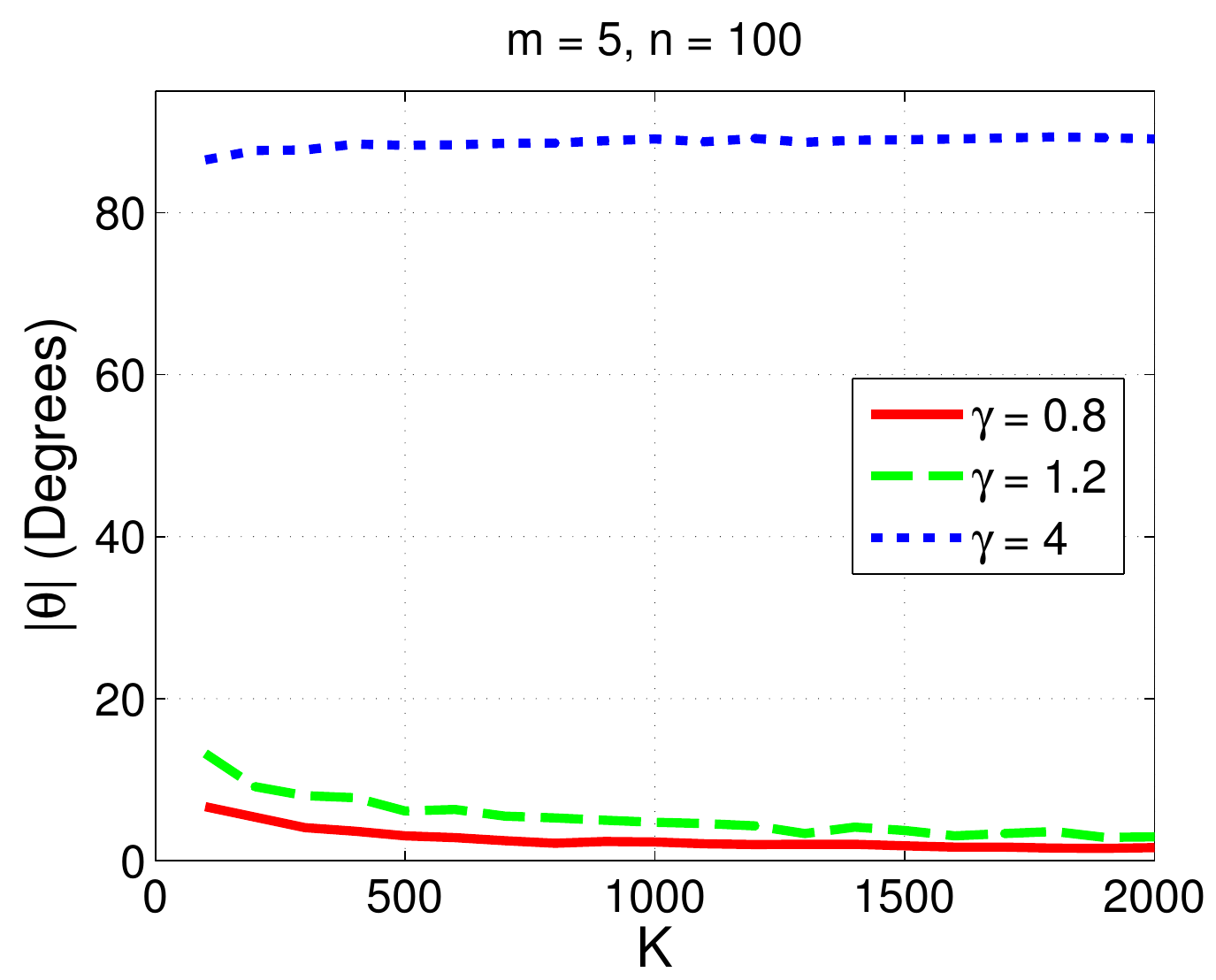} \end{minipage}}
\subfloat[Smooth mapping 1]{
\begin{minipage}[c]{0.28\linewidth}
\centering
\label{fig:md_exp1_gauss2} 
\includegraphics[width=1.0\linewidth]{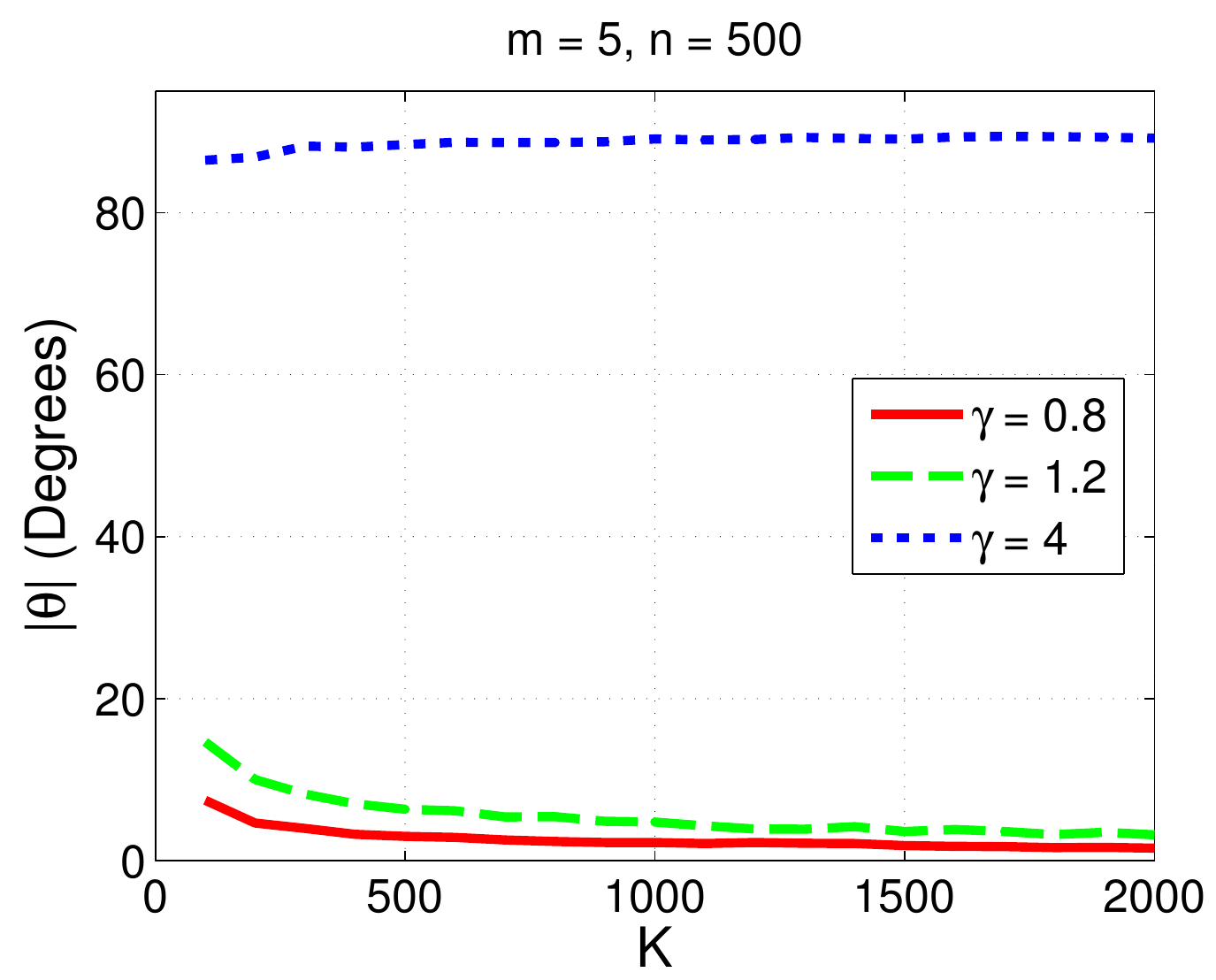} \end{minipage}}
\subfloat[Smooth mapping 1]{
\begin{minipage}[c]{0.28\linewidth}
\centering
\label{fig:md_exp1_gauss3} 
\includegraphics[width=1.0\linewidth]{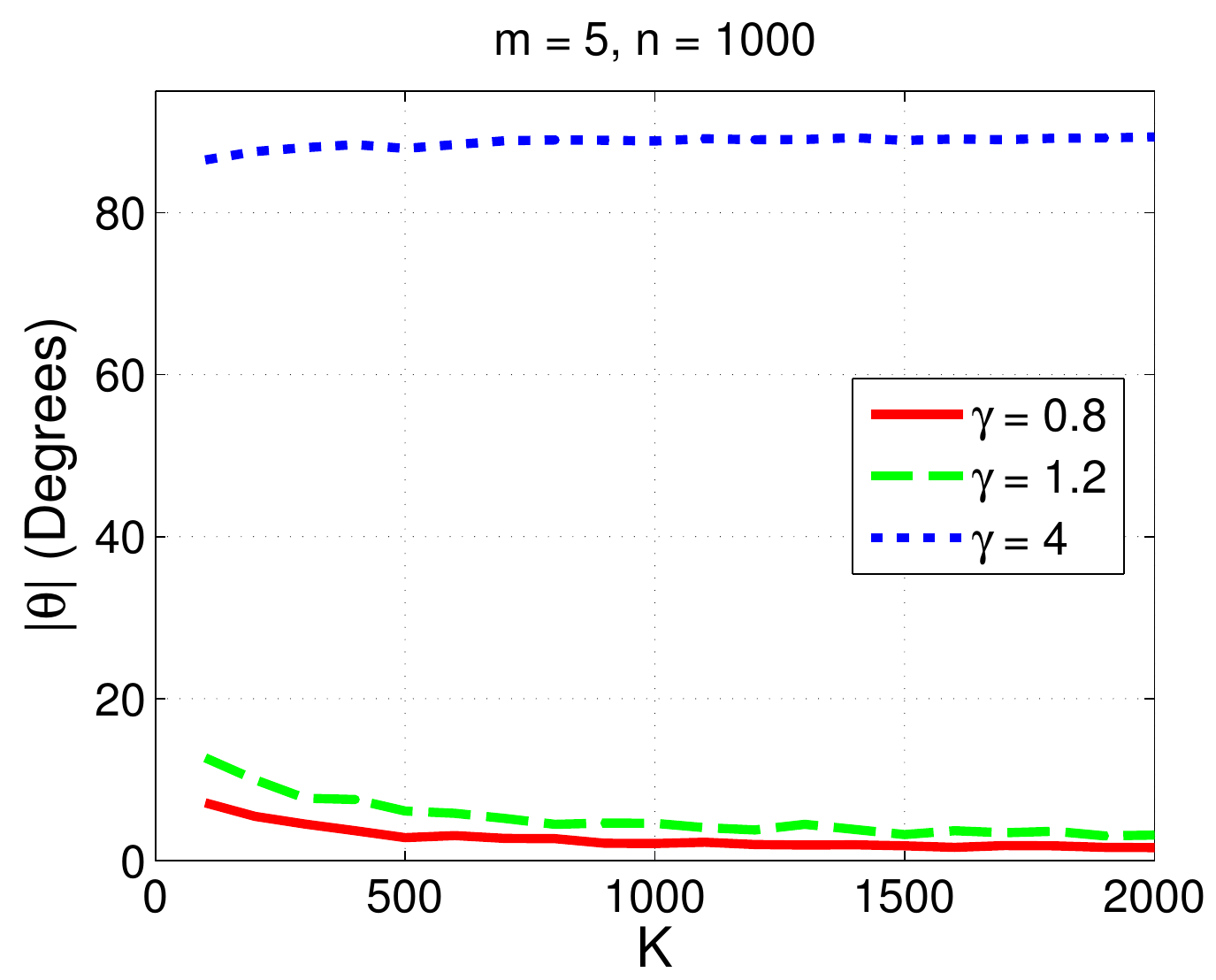} \end{minipage}}\\
\subfloat[Smooth mapping 2]{
\begin{minipage}[c]{0.28\linewidth}
\centering
\label{fig:md_exp1_sin1} 
\includegraphics[width=1.0\linewidth]{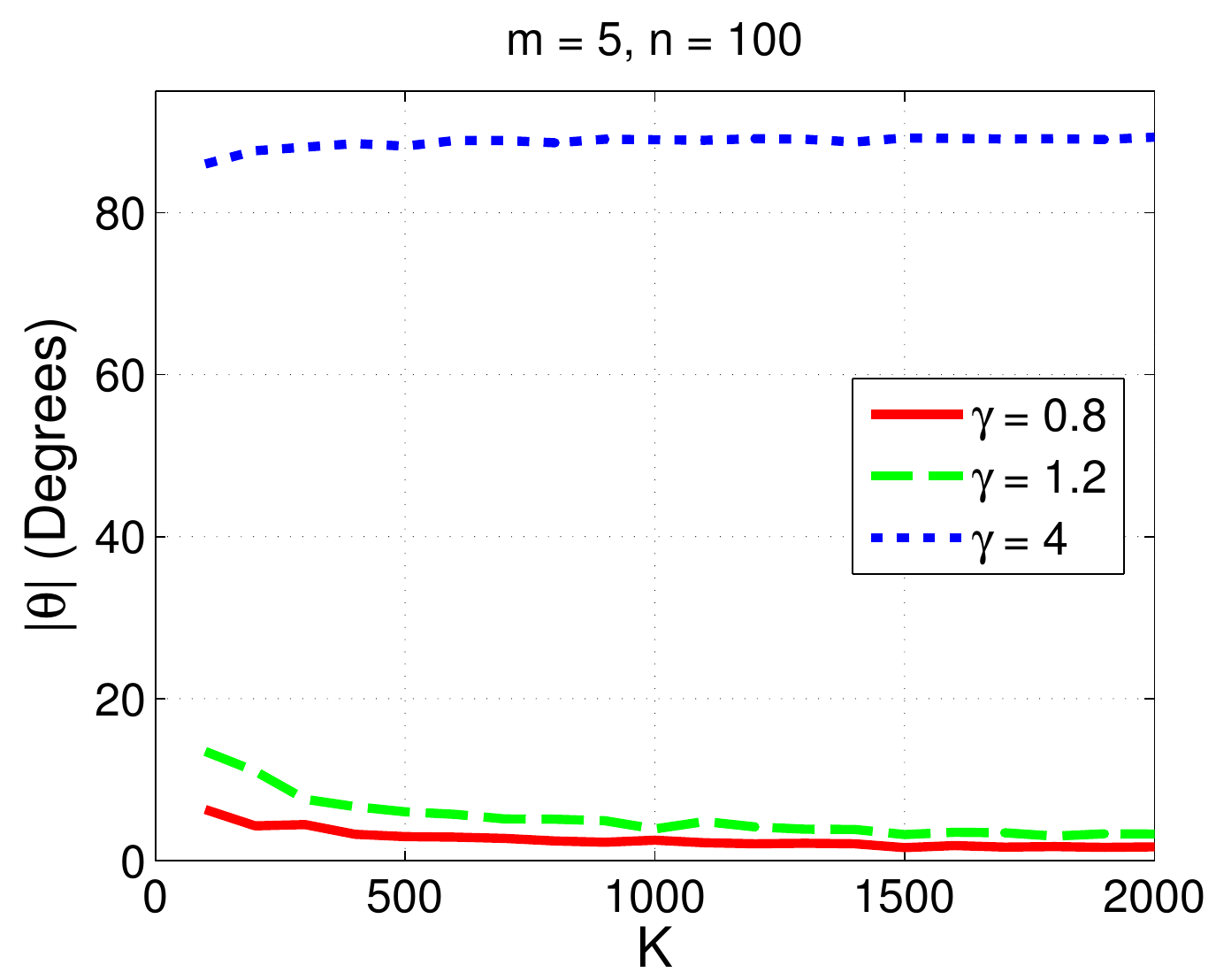} \end{minipage}}
\subfloat[Smooth mapping 2]{
\begin{minipage}[c]{0.28\linewidth}
\centering
\label{fig:md_exp1_sin2} 
\includegraphics[width=1.0\linewidth]{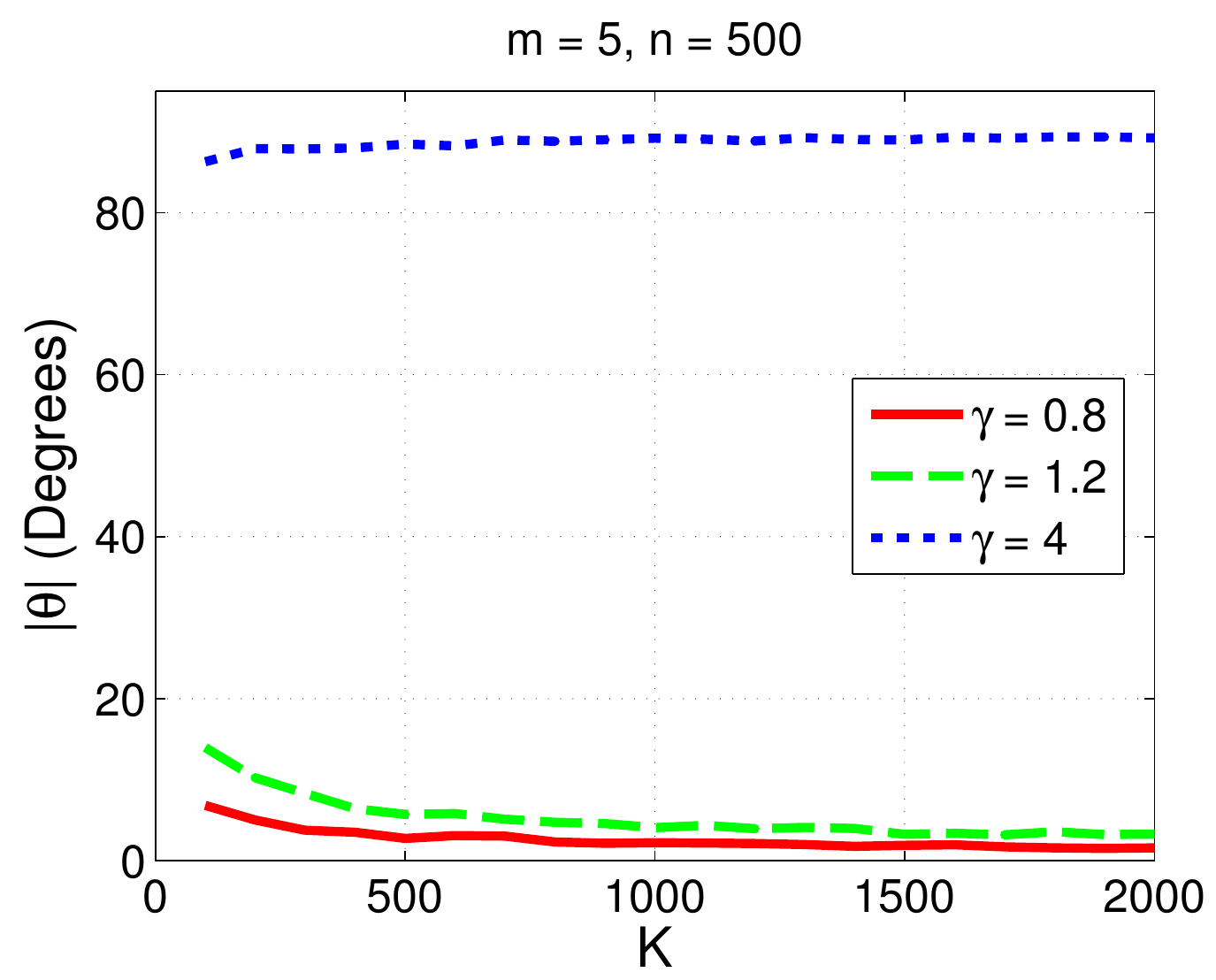} \end{minipage}}
\subfloat[Smooth mapping 2]{
\begin{minipage}[c]{0.28\linewidth}
\centering
\label{fig:md_exp1_sin3} 
\includegraphics[width=1.0\linewidth]{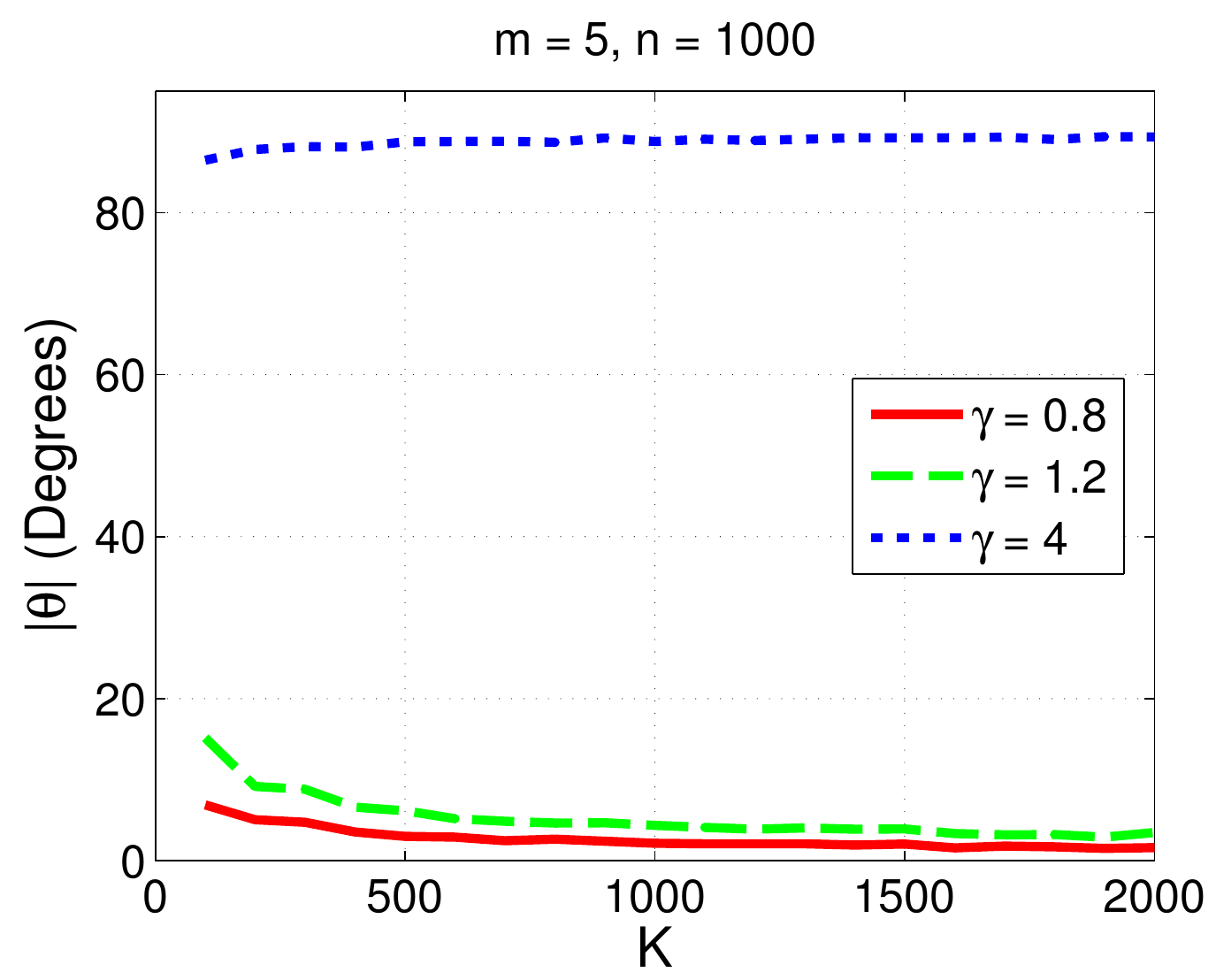} \end{minipage}}\\
\subfloat[Smooth mapping 3]{
\begin{minipage}[c]{0.28\linewidth}
\centering
\label{fig:md_exp1_poly1} 
\includegraphics[width=1.0\linewidth]{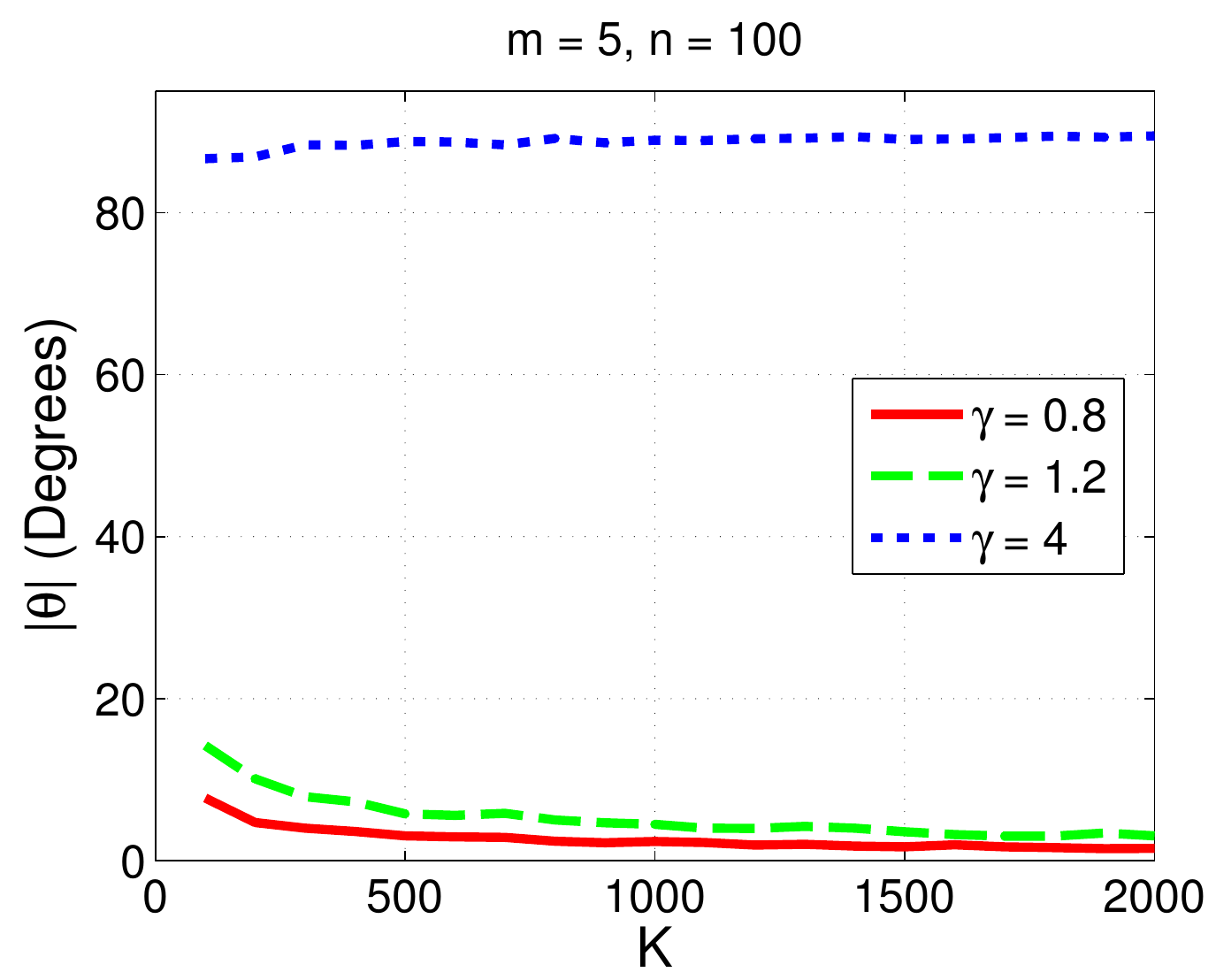} \end{minipage}}
\subfloat[Smooth mapping 3]{
\begin{minipage}[c]{0.28\linewidth}
\centering
\label{fig:md_exp1_poly2} 
\includegraphics[width=1.0\linewidth]{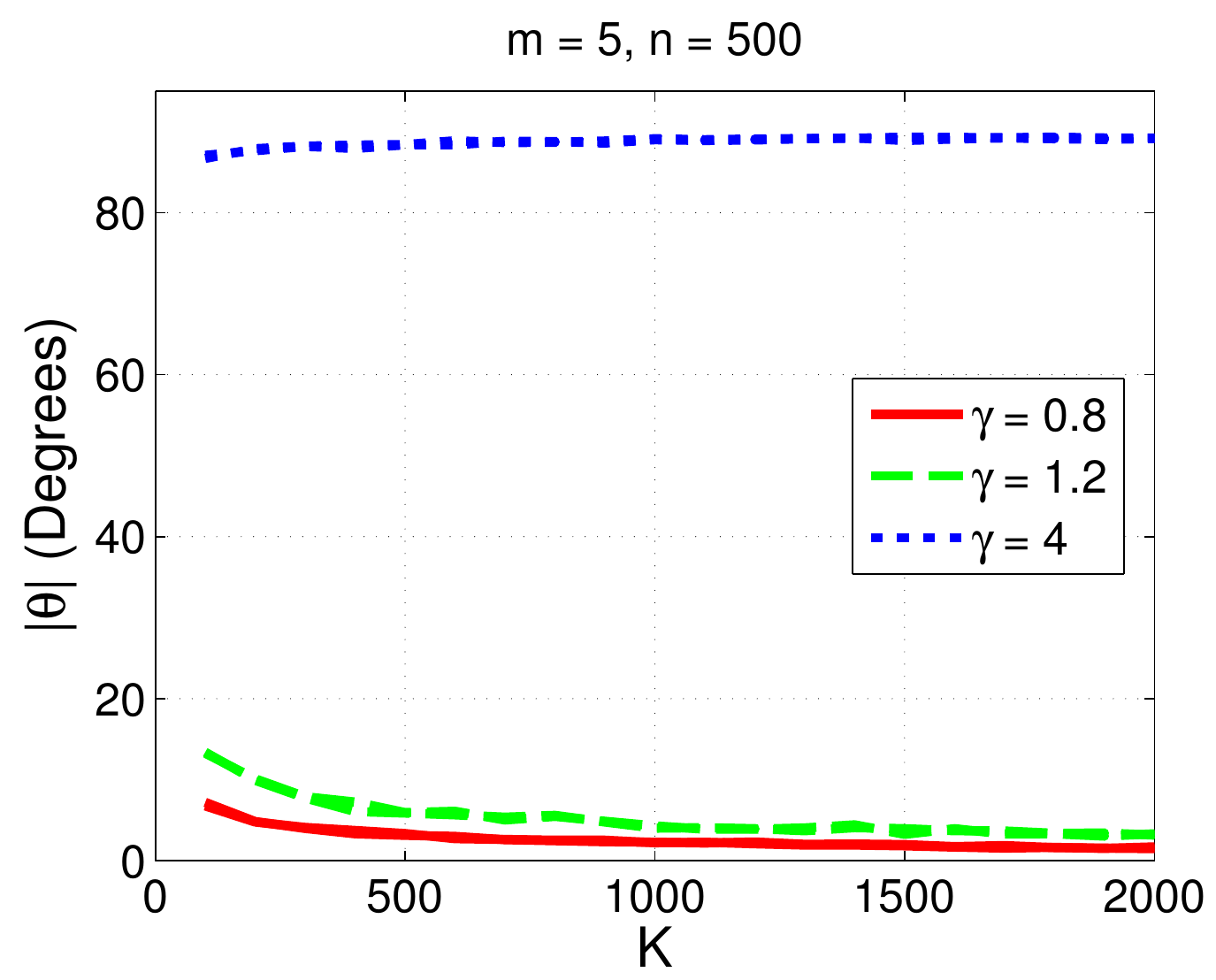} \end{minipage}}
\subfloat[Smooth mapping 3]{
\begin{minipage}[c]{0.28\linewidth}
\centering
\label{fig:md_exp1_poly3} 
\includegraphics[width=1.0\linewidth]{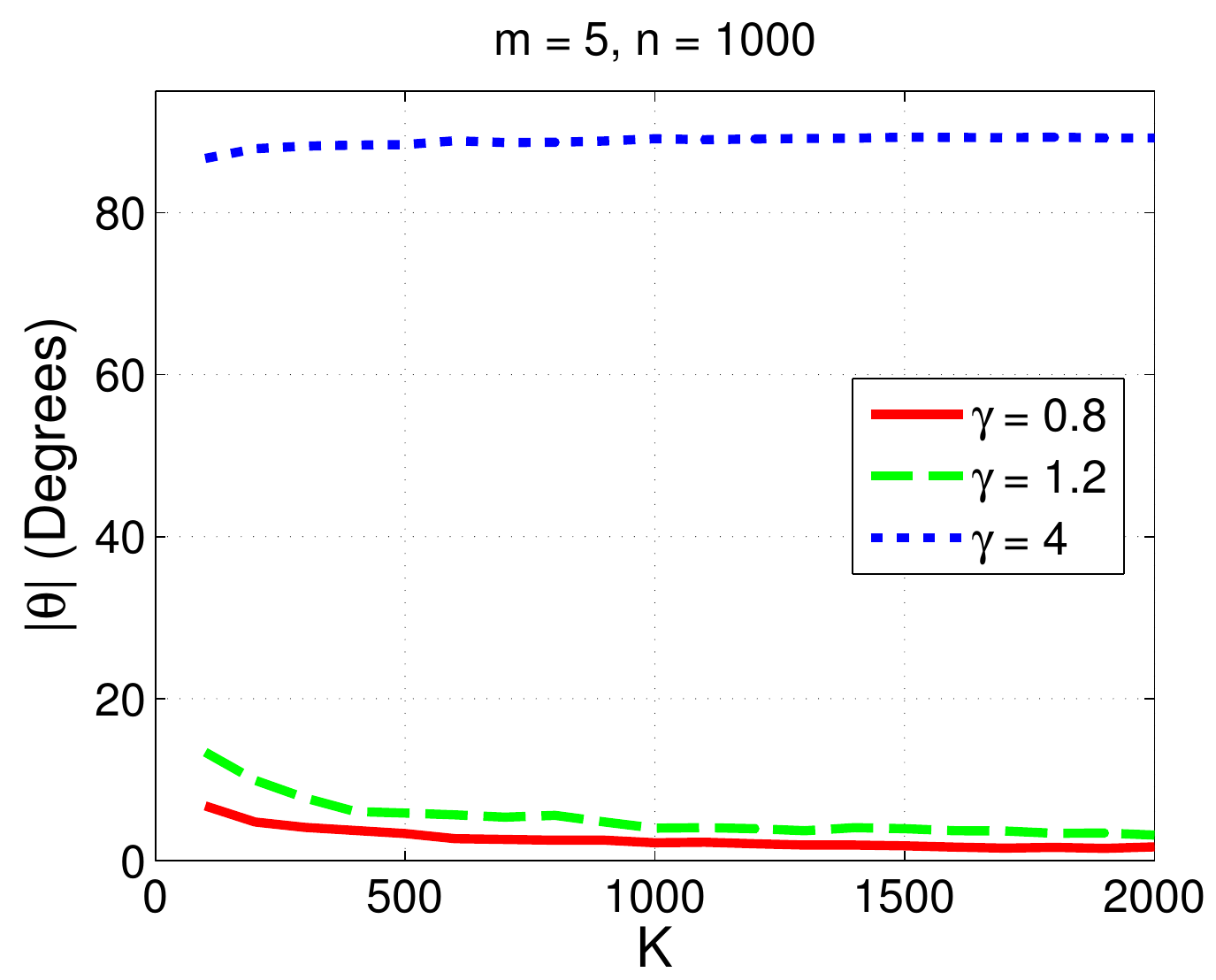} \end{minipage}}
\caption{\small Variation of the deviation $\abs{\theta}$ with respect to $K$ for different sampling widths $\nu$. For each type of mapping, $\nu = \gamma \nu_{\text{bound,quad}}$.}
\label{fig:md_exp1} 
\end{figure}

\begin{figure}[!htbp]
\centering
\subfloat[\small m=5,n=100]{
\begin{minipage}[c]{0.28\linewidth}
\centering
\label{fig:quad_theo_m5n100} 
\noindent \includegraphics[width=1.0\linewidth]{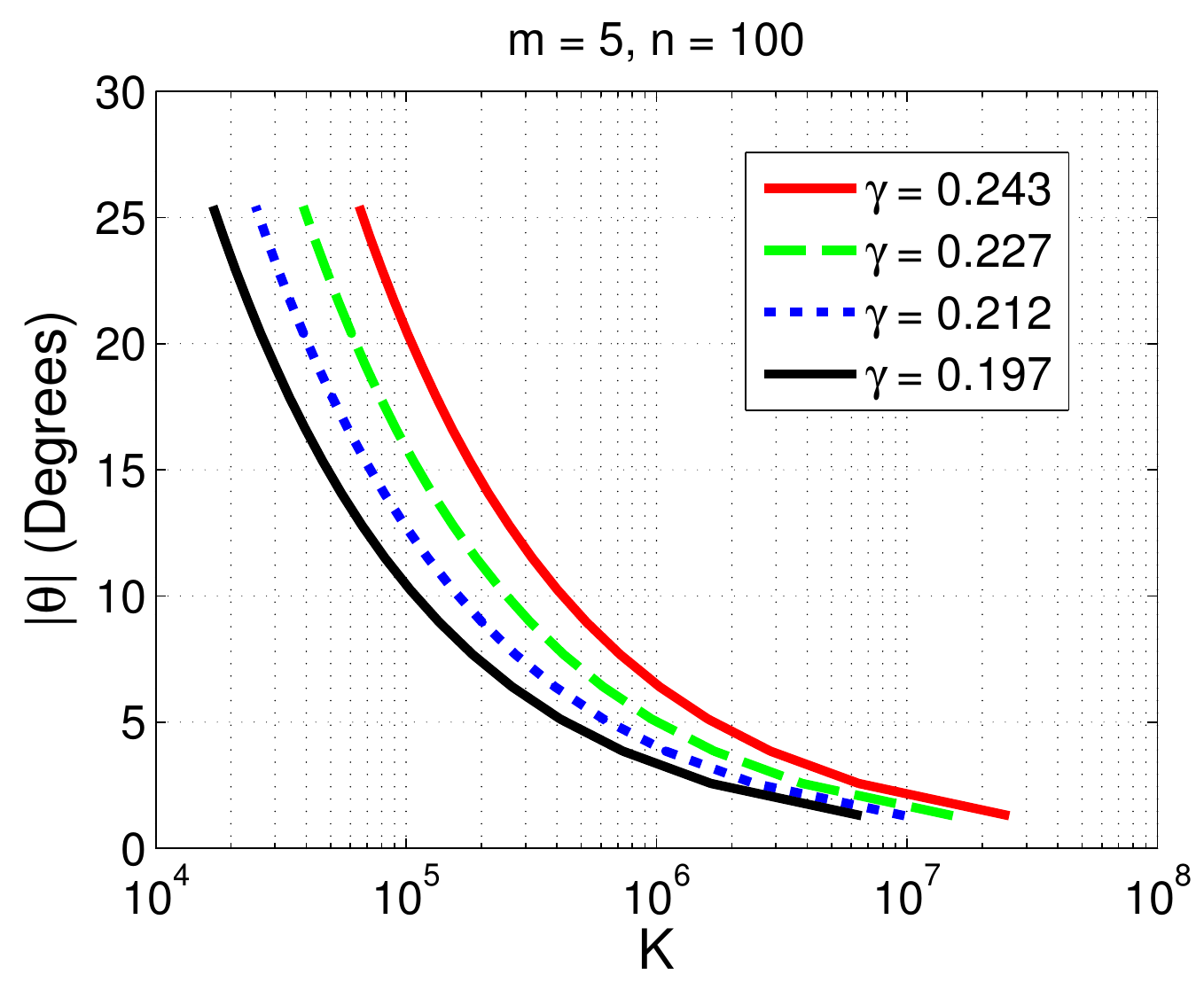}\end{minipage}}
\subfloat[\small m=5,n=500]{
\begin{minipage}[c]{0.28\linewidth}
\centering
\label{fig:quad_theo_m5n500} 
\noindent \includegraphics[width=1.0\linewidth]{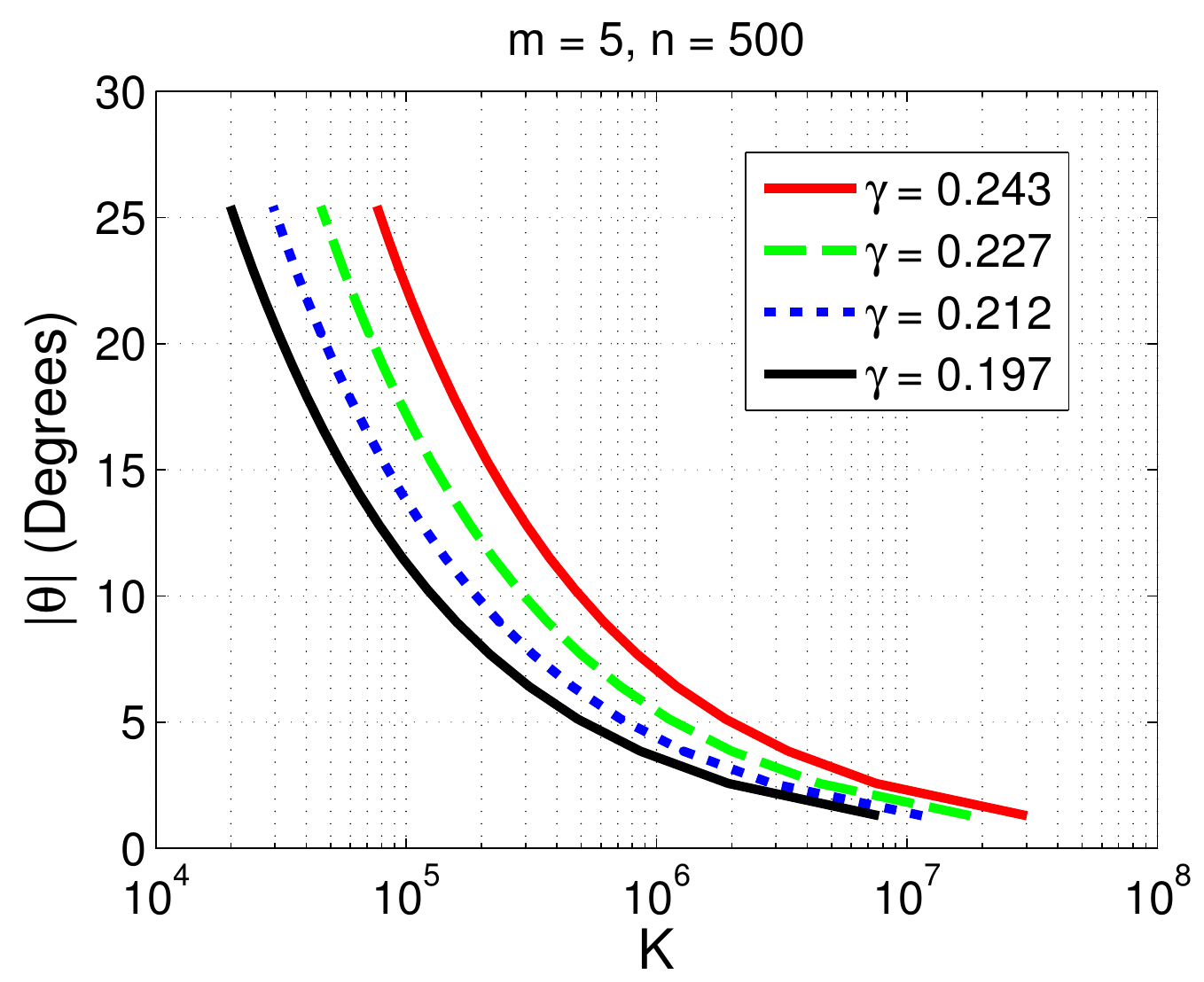}\end{minipage}}
\subfloat[\small m=5,n=1000]{
\begin{minipage}[c]{0.28\linewidth}
\centering
\label{fig:quad_theo_m5n1000} 
\noindent \includegraphics[width=1.0\linewidth]{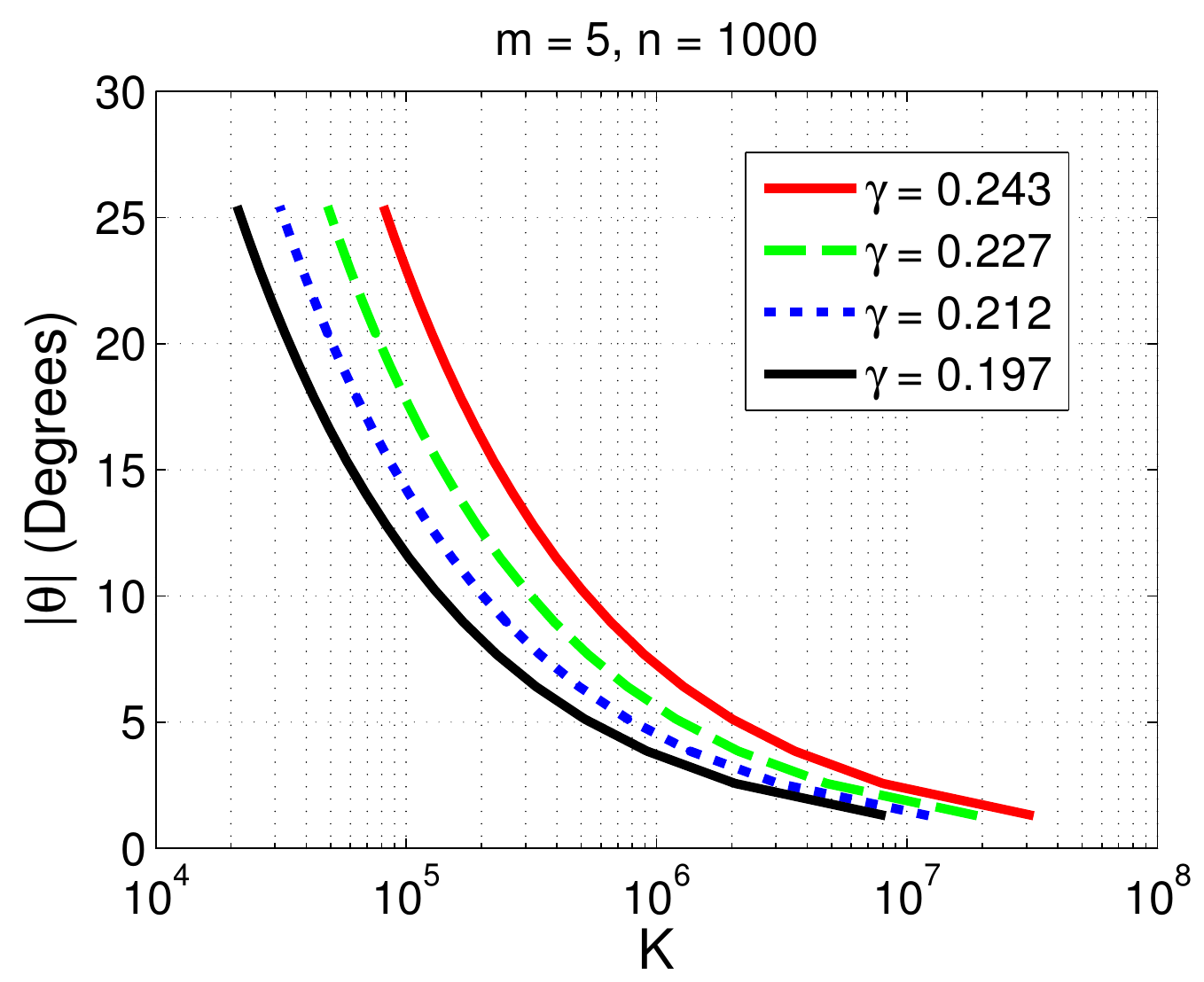}\end{minipage}}\\
\subfloat[\small m=5,n=100]{
\begin{minipage}[c]{0.28\linewidth}
\centering
\label{fig:quad_emp_m5n100} 
\noindent \includegraphics[width=1.0\linewidth]{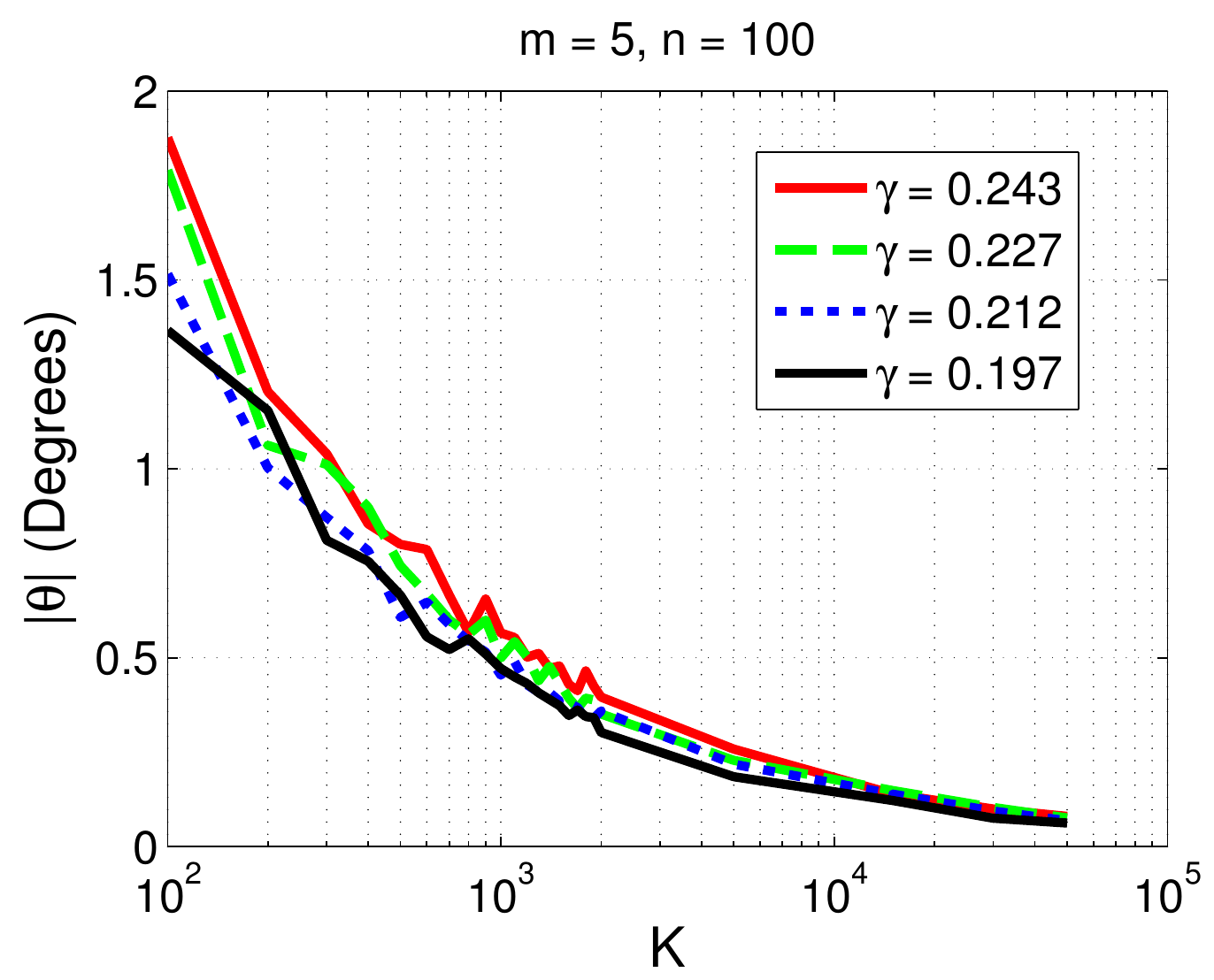}\end{minipage}}
\subfloat[\small m=5,n=500]{
\begin{minipage}[c]{0.28\linewidth}
\centering
\label{fig:quad_emp_m5n500} 
\noindent \includegraphics[width=1.0\linewidth]{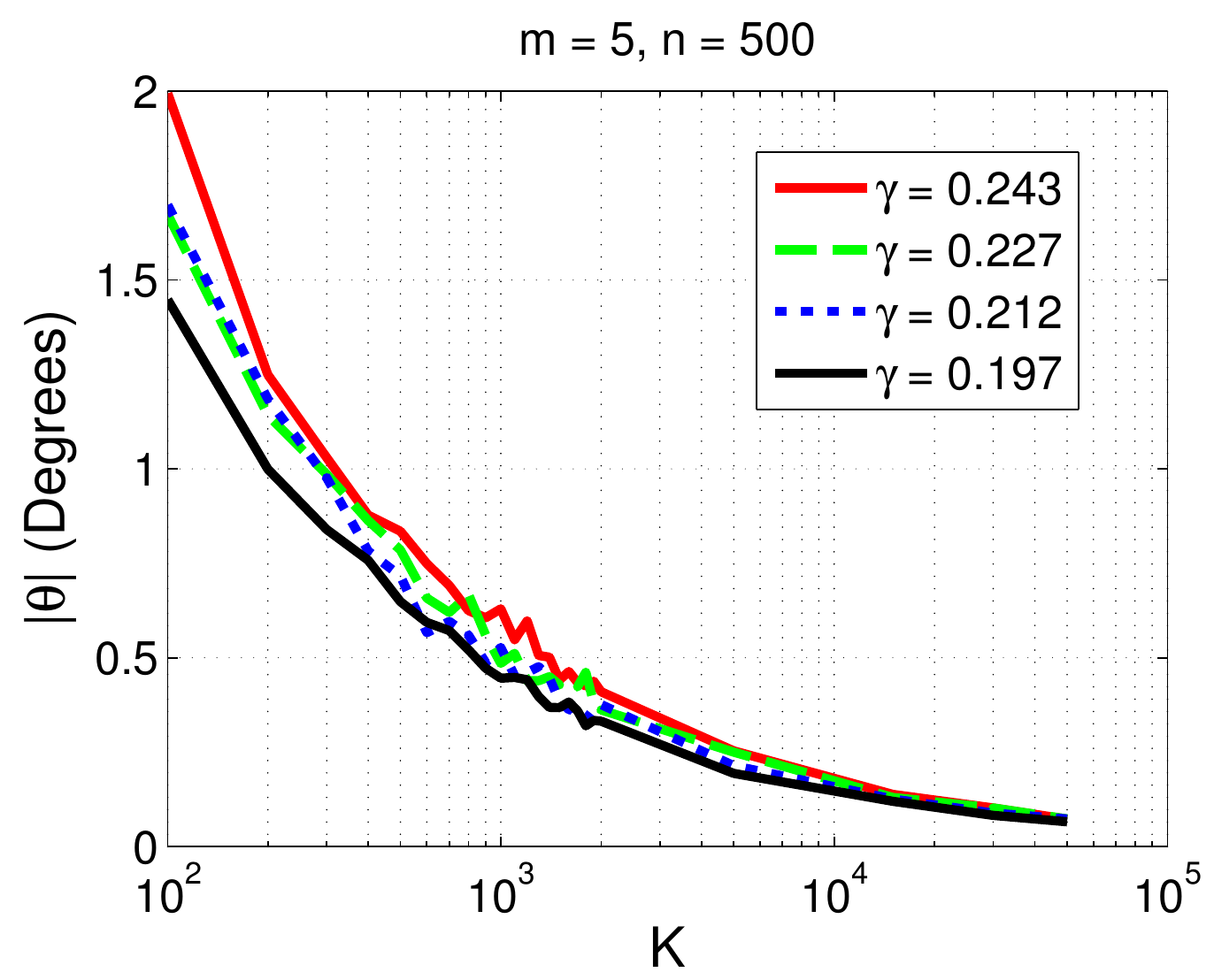}\end{minipage}}
\subfloat[\small m=5,n=1000]{
\begin{minipage}[c]{0.28\linewidth}
\centering
\label{fig:quad_emp_m5n1000} 
\noindent \includegraphics[width=1.0\linewidth]{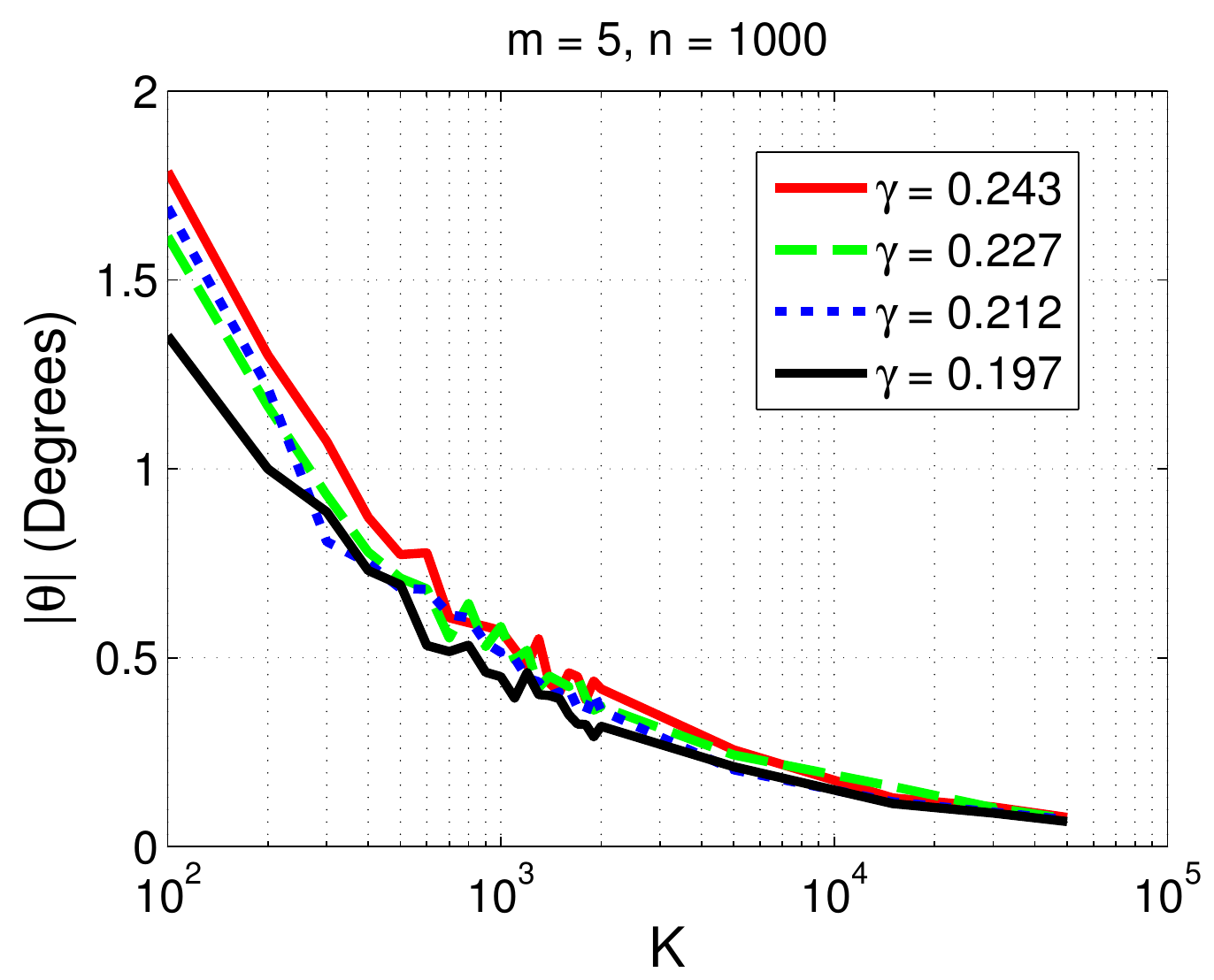}\end{minipage}}\\
\caption{\small Variation of the deviation $\abs{\theta}$ with respect to $K$ for different sampling widths $\nu= \gamma \nu_{\text{bound,quad}}$ for quadratic embedding. Figures \ref{fig:quad_theo_m5n100}-\ref{fig:quad_theo_m5n1000} show theoretical plots while Figures \ref{fig:quad_emp_m5n100}-\ref{fig:quad_emp_m5n1000} show empirical plots.}
\label{fig:quad_theo_emp_exp1} 
\end{figure}
\begin{figure}[!htbp]
\centering
\subfloat[\small m=5,n=100]{
\begin{minipage}[c]{0.28\linewidth}
\centering
\label{fig:smooth1_theo_m5n100} 
\noindent \includegraphics[width=1.0\linewidth]{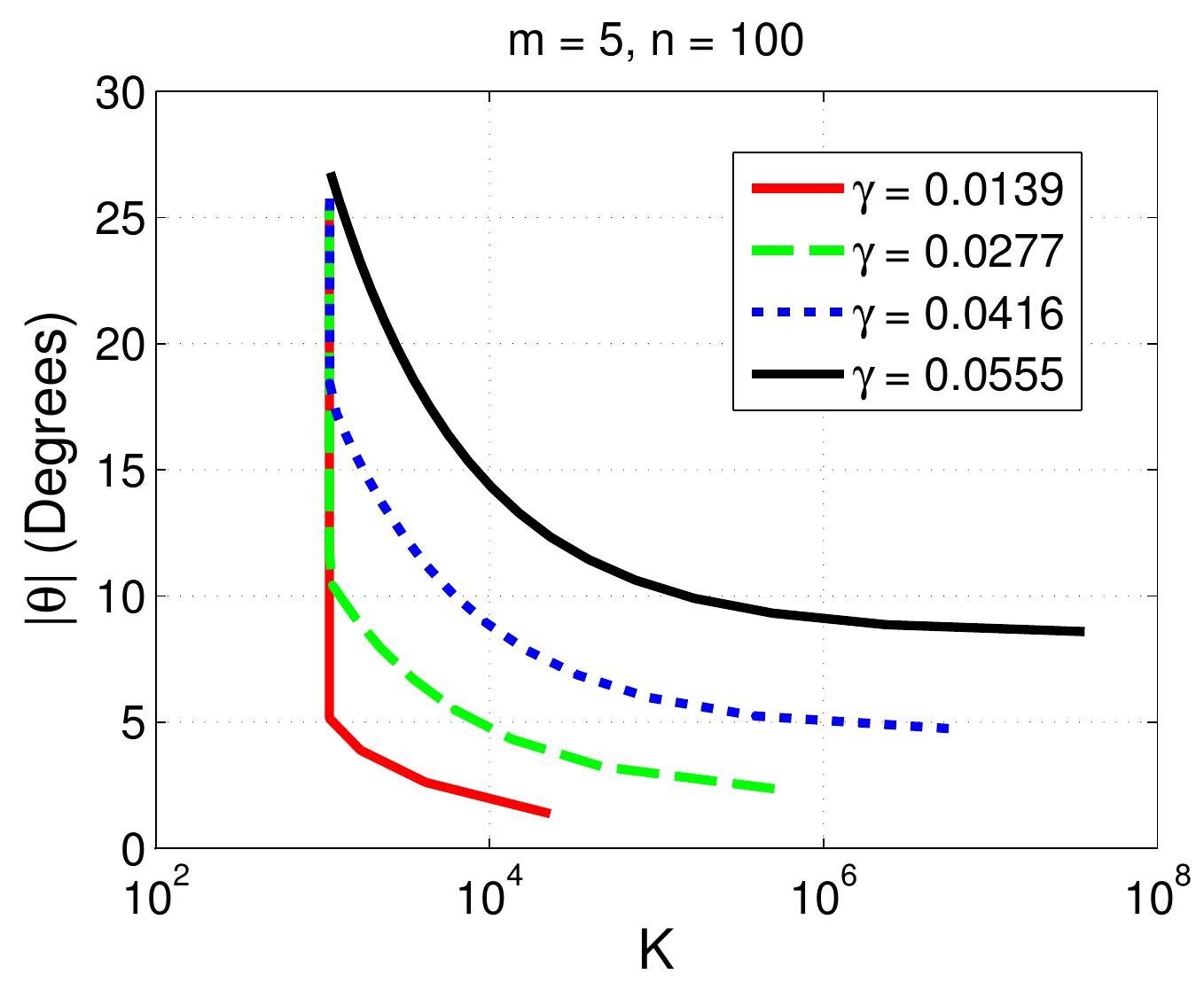}\end{minipage}}
\subfloat[\small m=5,n=500]{
\begin{minipage}[c]{0.28\linewidth}
\centering
\label{fig:smooth1_theo_m5n500} 
\noindent \includegraphics[width=1.0\linewidth]{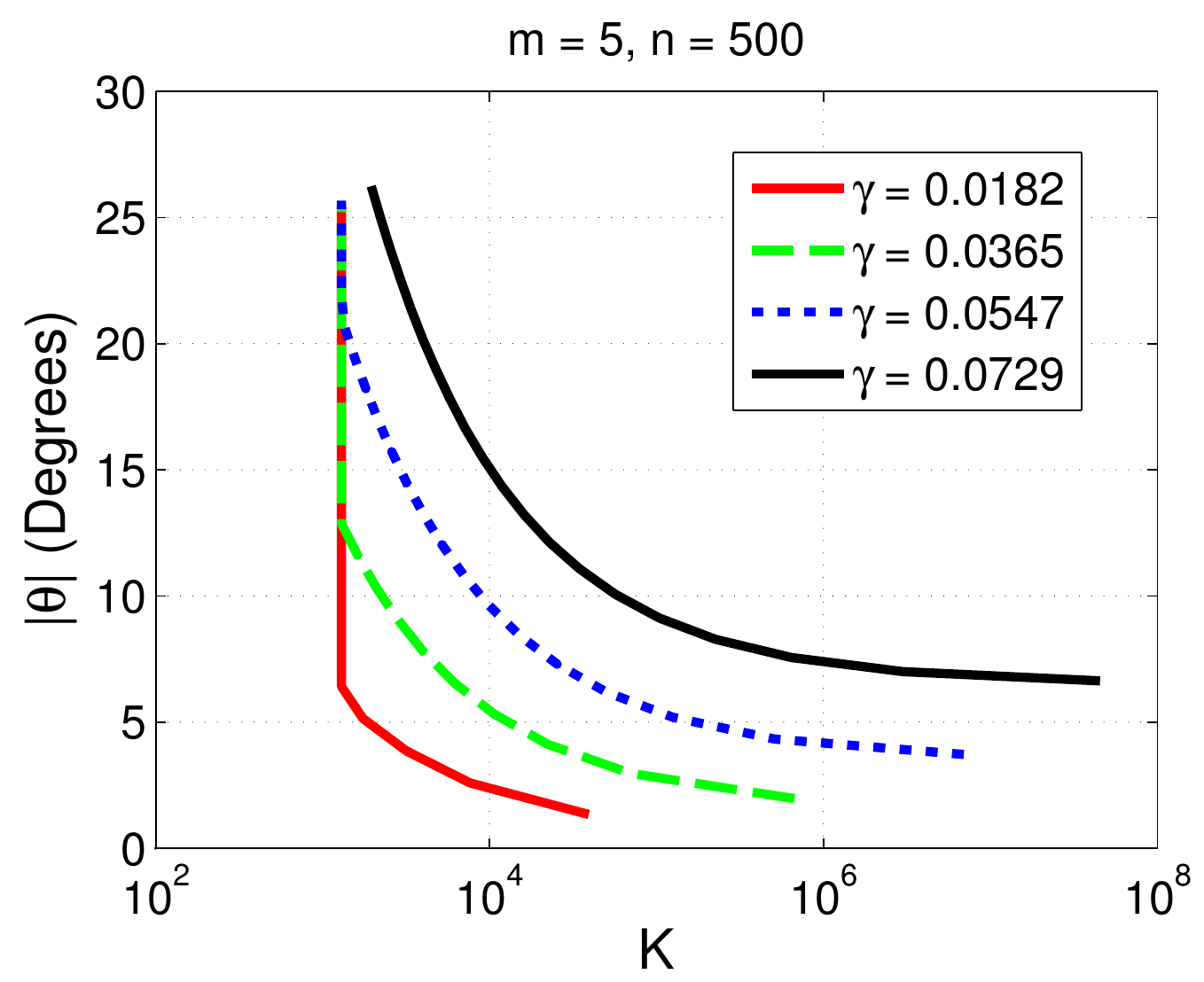}\end{minipage}}
\subfloat[\small m=5,n=1000]{
\begin{minipage}[c]{0.28\linewidth}
\centering
\label{fig:smooth1_theo_m5n1000} 
\noindent \includegraphics[width=1.0\linewidth]{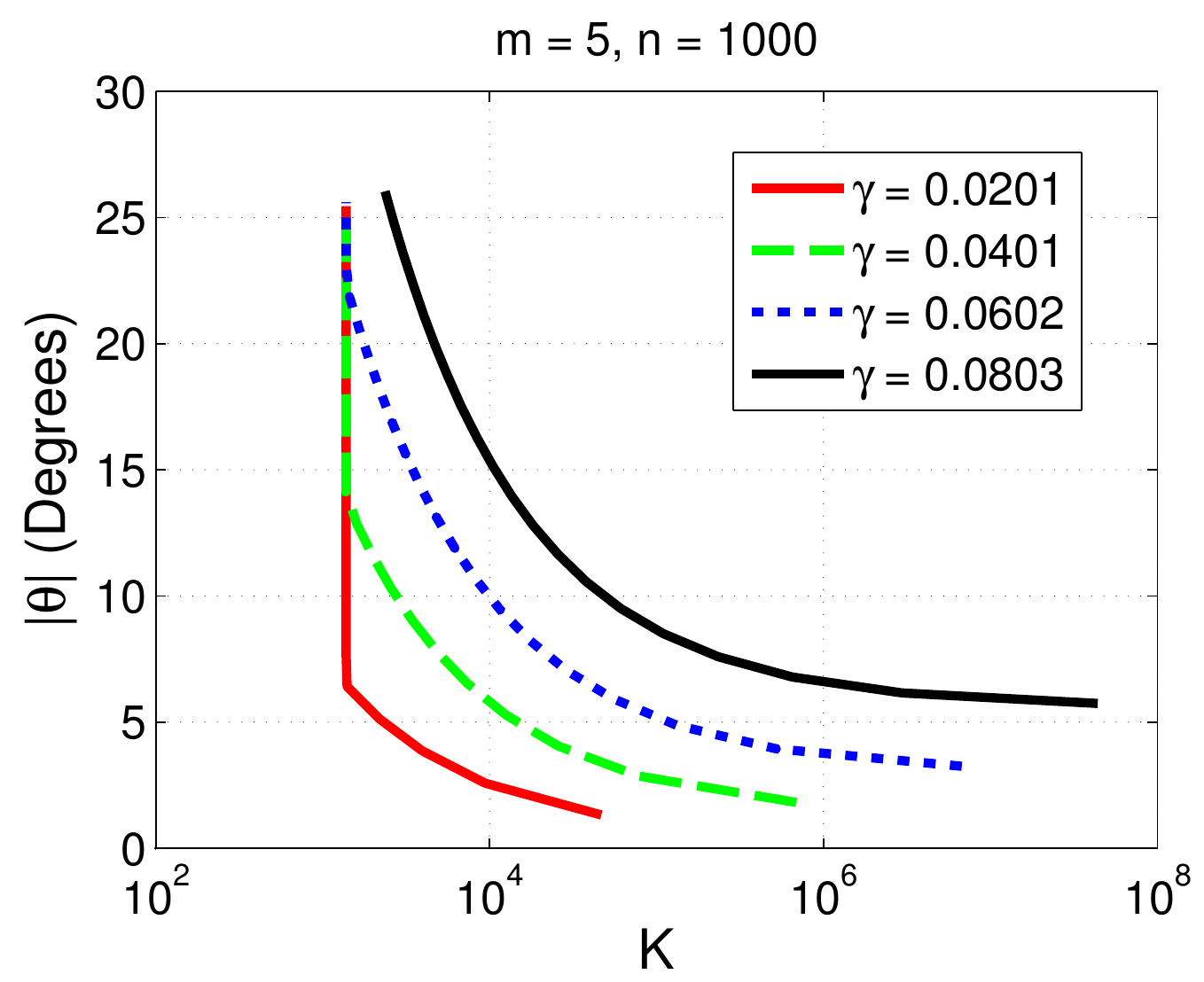}\end{minipage}}\\
\subfloat[\small m=5,n=100]{
\begin{minipage}[c]{0.28\linewidth}
\centering
\label{fig:smooth1_emp_m5n100} 
\noindent \includegraphics[width=1.0\linewidth]{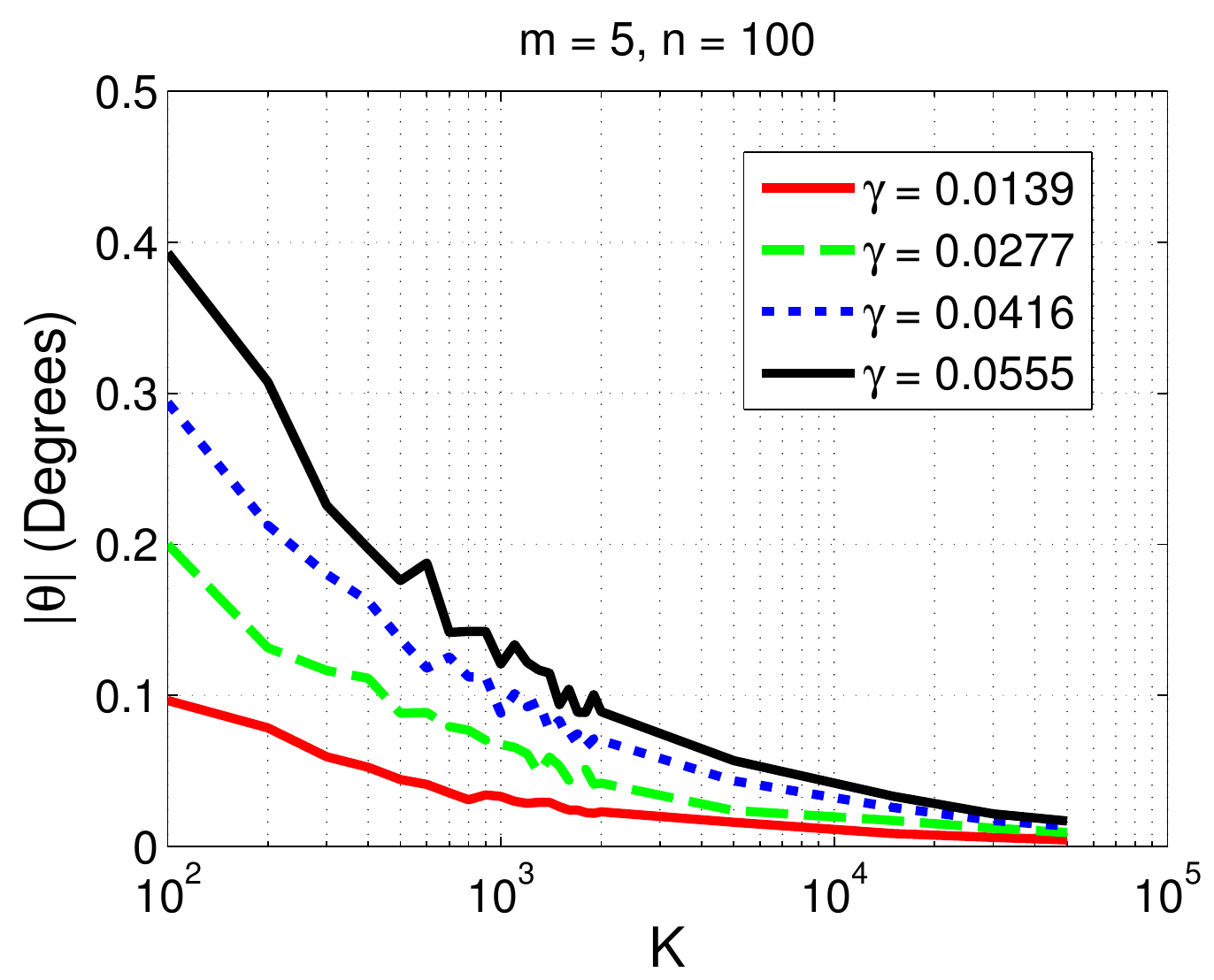}\end{minipage}}
\subfloat[\small m=5,n=500]{
\begin{minipage}[c]{0.28\linewidth}
\centering
\label{fig:smooth1_emp_m5n500} 
\noindent \includegraphics[width=1.0\linewidth]{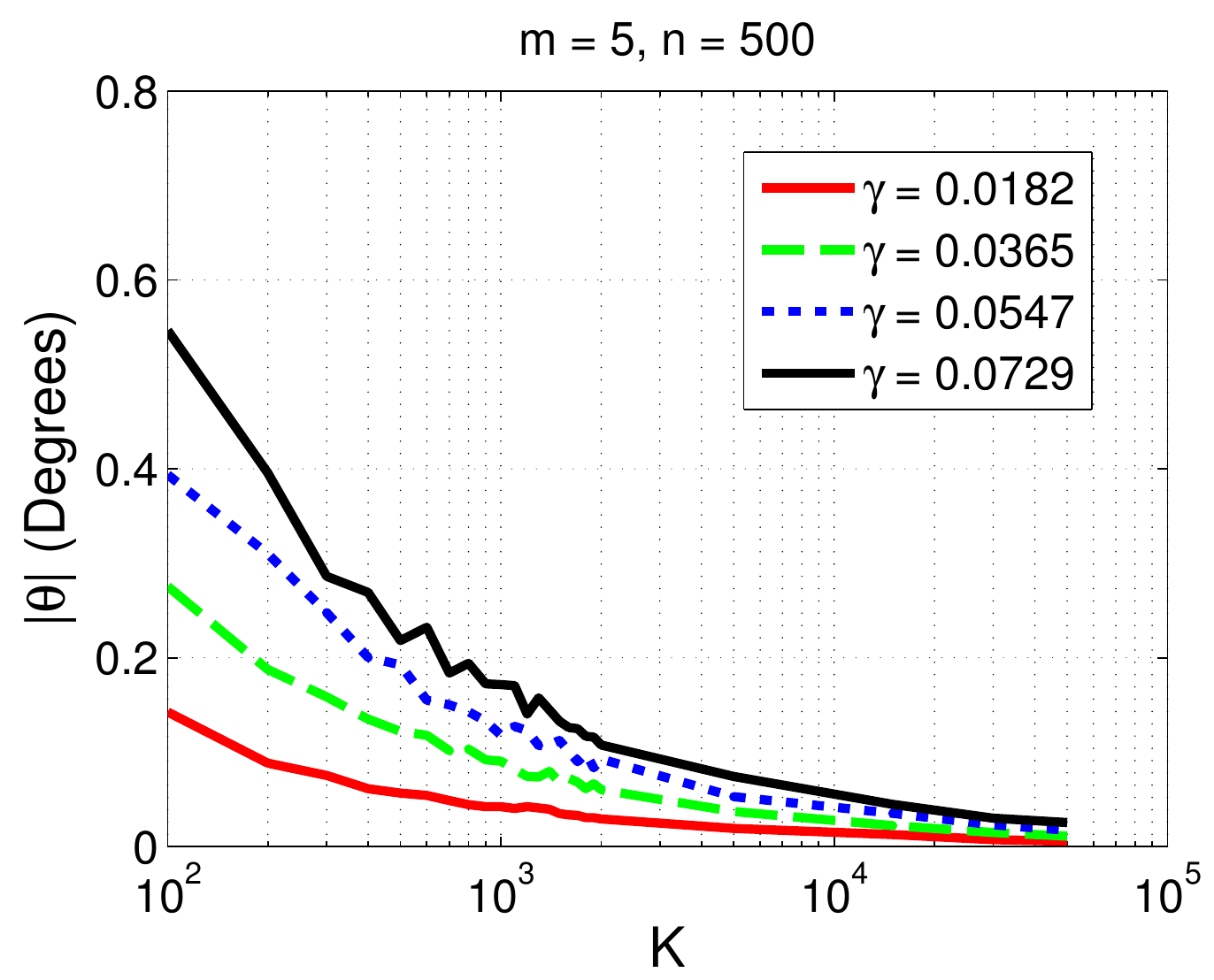}\end{minipage}}
\subfloat[\small m=5,n=1000]{
\begin{minipage}[c]{0.28\linewidth}
\centering
\label{fig:smooth1_emp_m5n1000} 
\noindent \includegraphics[width=1.0\linewidth]{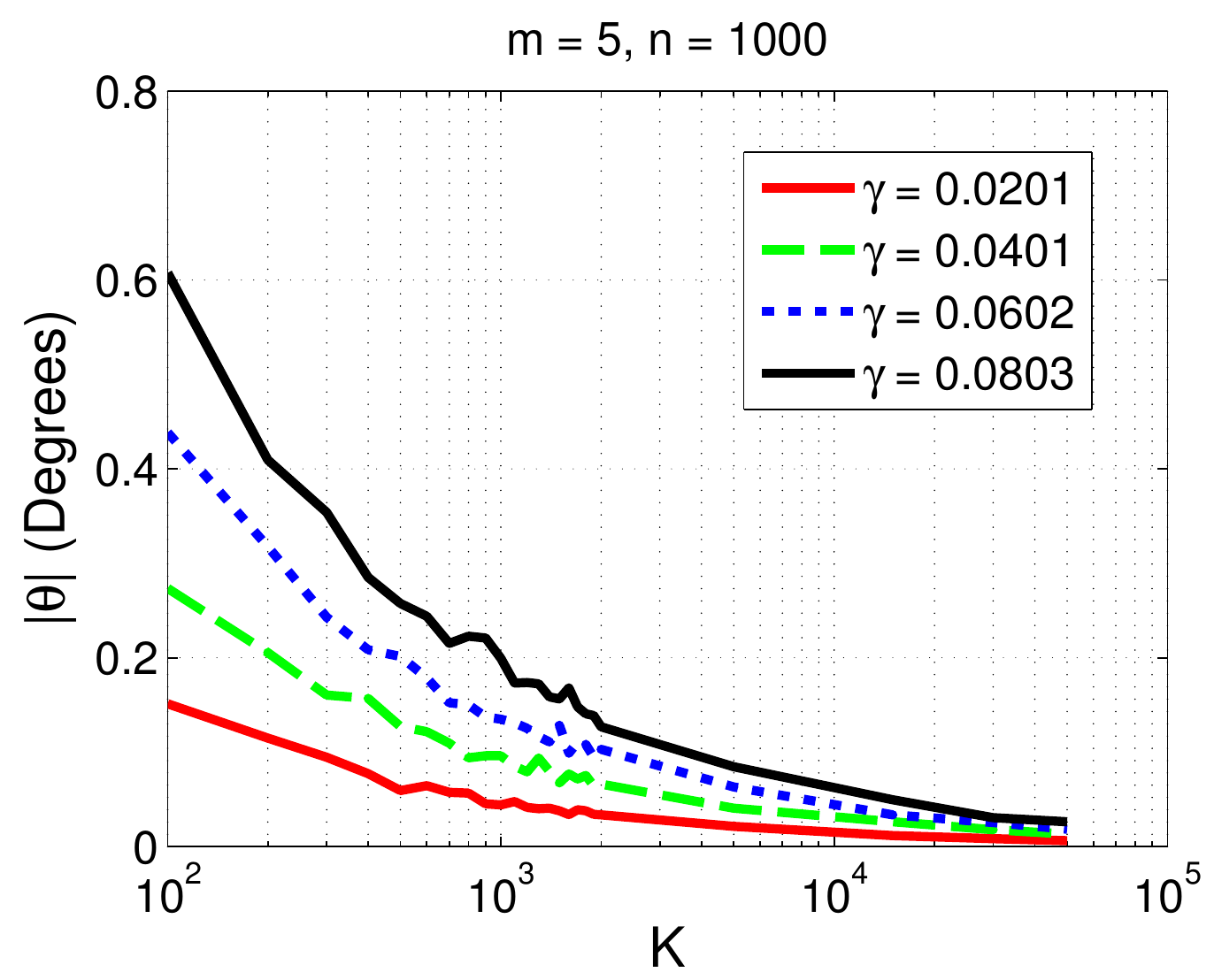}\end{minipage}}\\

\caption{\small Variation of the deviation $\abs{\theta}$ with respect to $K$ for different sampling widths $\nu = \gamma \nu_{\text{bound,quad}}$ for smooth mapping 1. Figures \ref{fig:smooth1_theo_m5n100}-\ref{fig:smooth1_theo_m5n1000} show theoretical plots while Figures \ref{fig:smooth1_emp_m5n100}-\ref{fig:smooth1_emp_m5n1000} show empirical plots.}
\label{fig:smooth1_theo_emp_exp1} 
\end{figure}
\begin{figure}[!htbp]
\centering
\subfloat[\small m=5,n=100]{
\begin{minipage}[c]{0.28\linewidth}
\centering
\label{fig:smooth2_theo_m5n100} 
\noindent \includegraphics[width=1.0\linewidth]{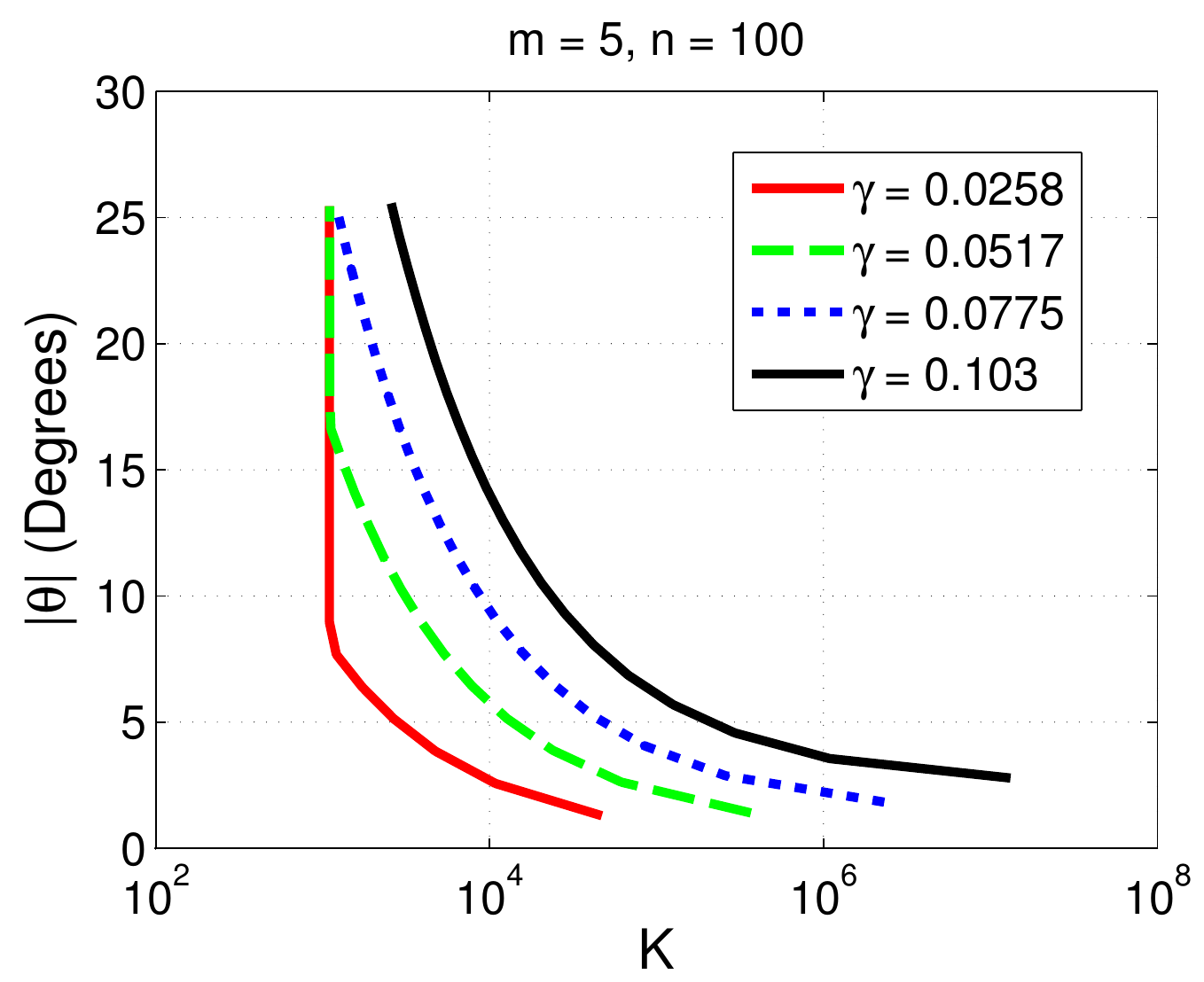}\end{minipage}}
\subfloat[\small m=5,n=500]{
\begin{minipage}[c]{0.28\linewidth}
\centering
\label{fig:smooth2_theo_m5n500} 
\noindent \includegraphics[width=1.0\linewidth]{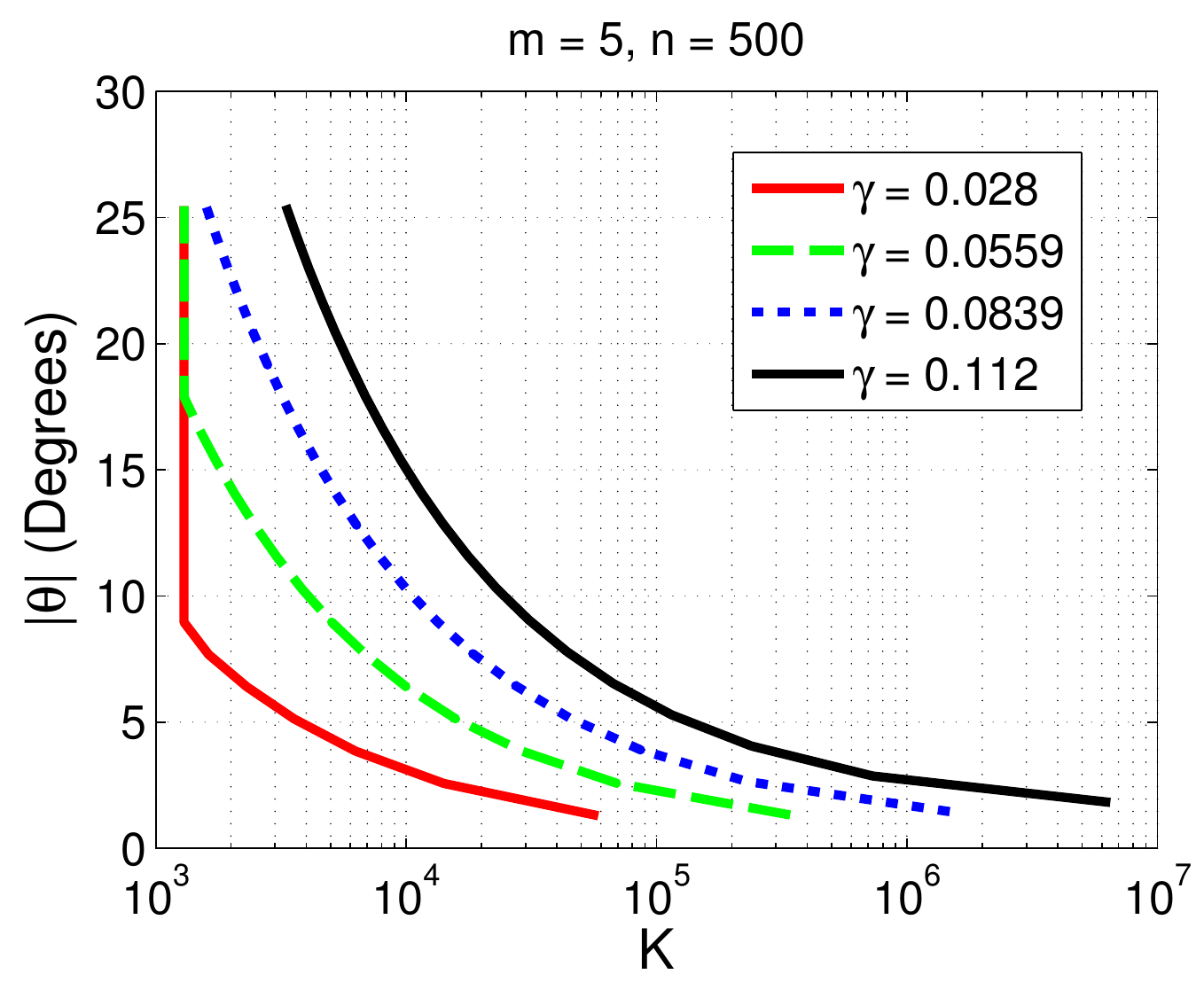}\end{minipage}}
\subfloat[\small m=5,n=1000]{
\begin{minipage}[c]{0.28\linewidth}
\centering
\label{fig:smooth2_theo_m5n1000} 
\noindent \includegraphics[width=1.0\linewidth]{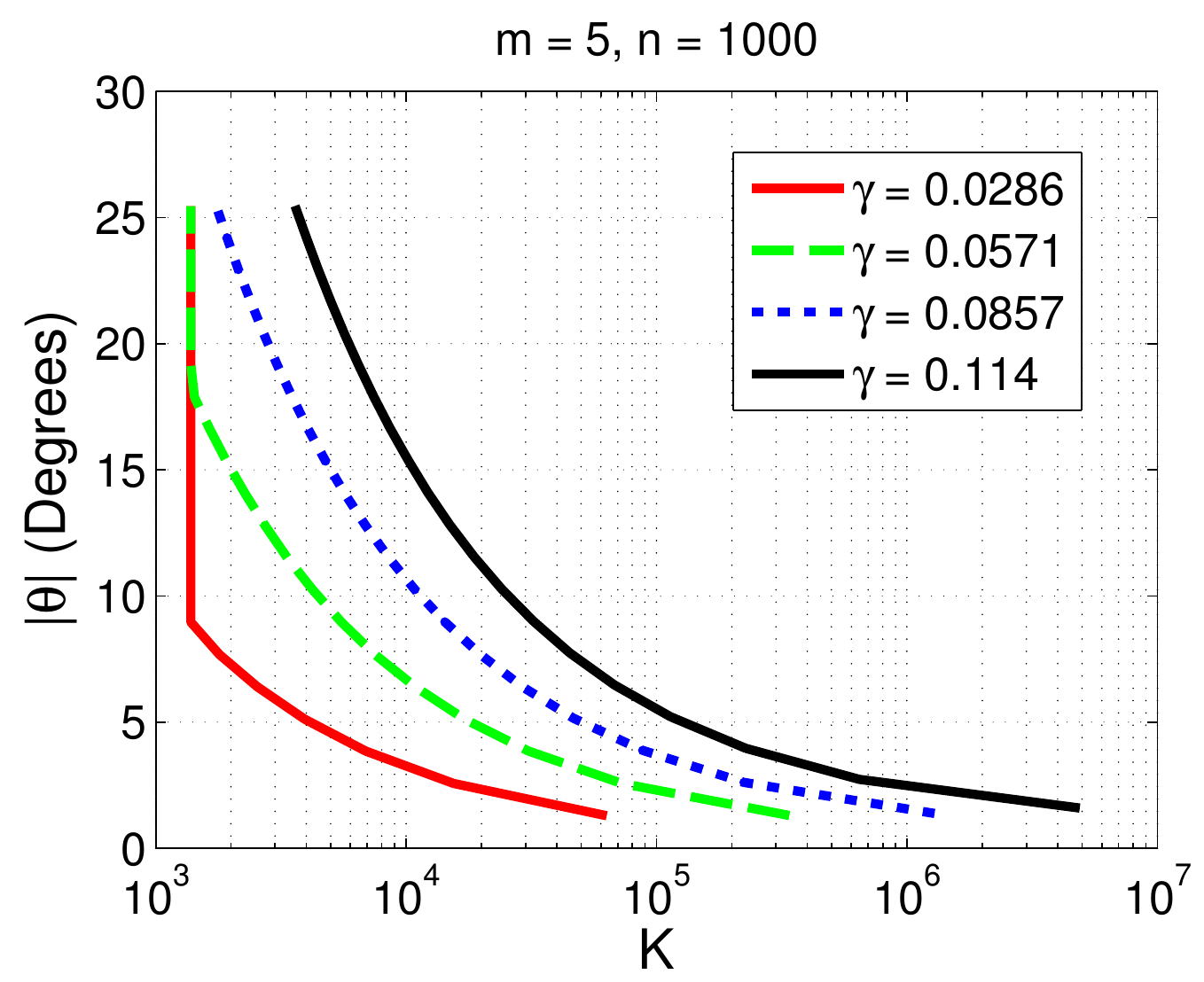}\end{minipage}}\\
\subfloat[\small m=5,n=100]{
\begin{minipage}[c]{0.28\linewidth}
\centering
\label{fig:smooth2_emp_m5n100} 
\noindent \includegraphics[width=1.0\linewidth]{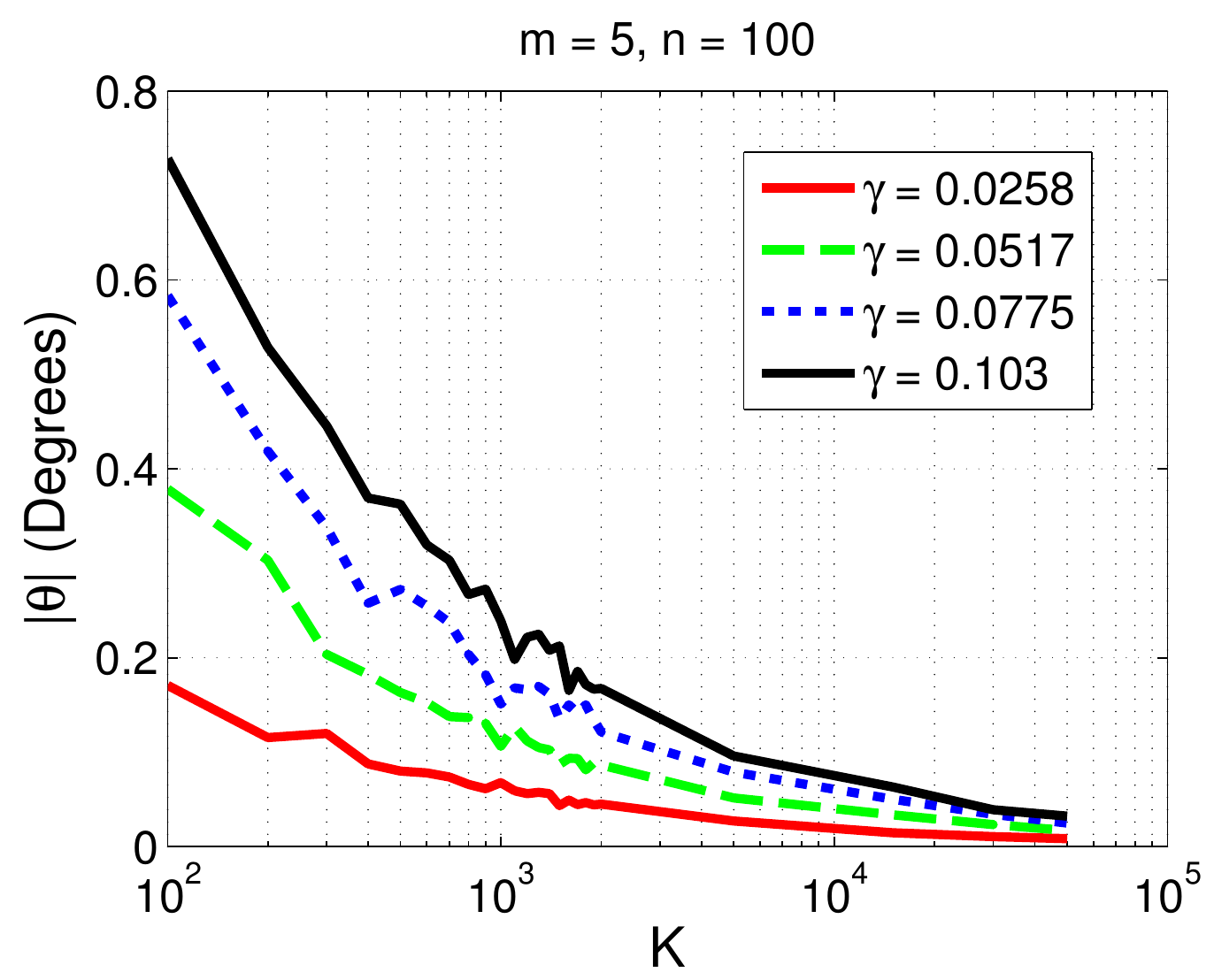}\end{minipage}}
\subfloat[\small m=5,n=500]{
\begin{minipage}[c]{0.28\linewidth}
\centering
\label{fig:smooth2_emp_m5n500} 
\noindent \includegraphics[width=1.0\linewidth]{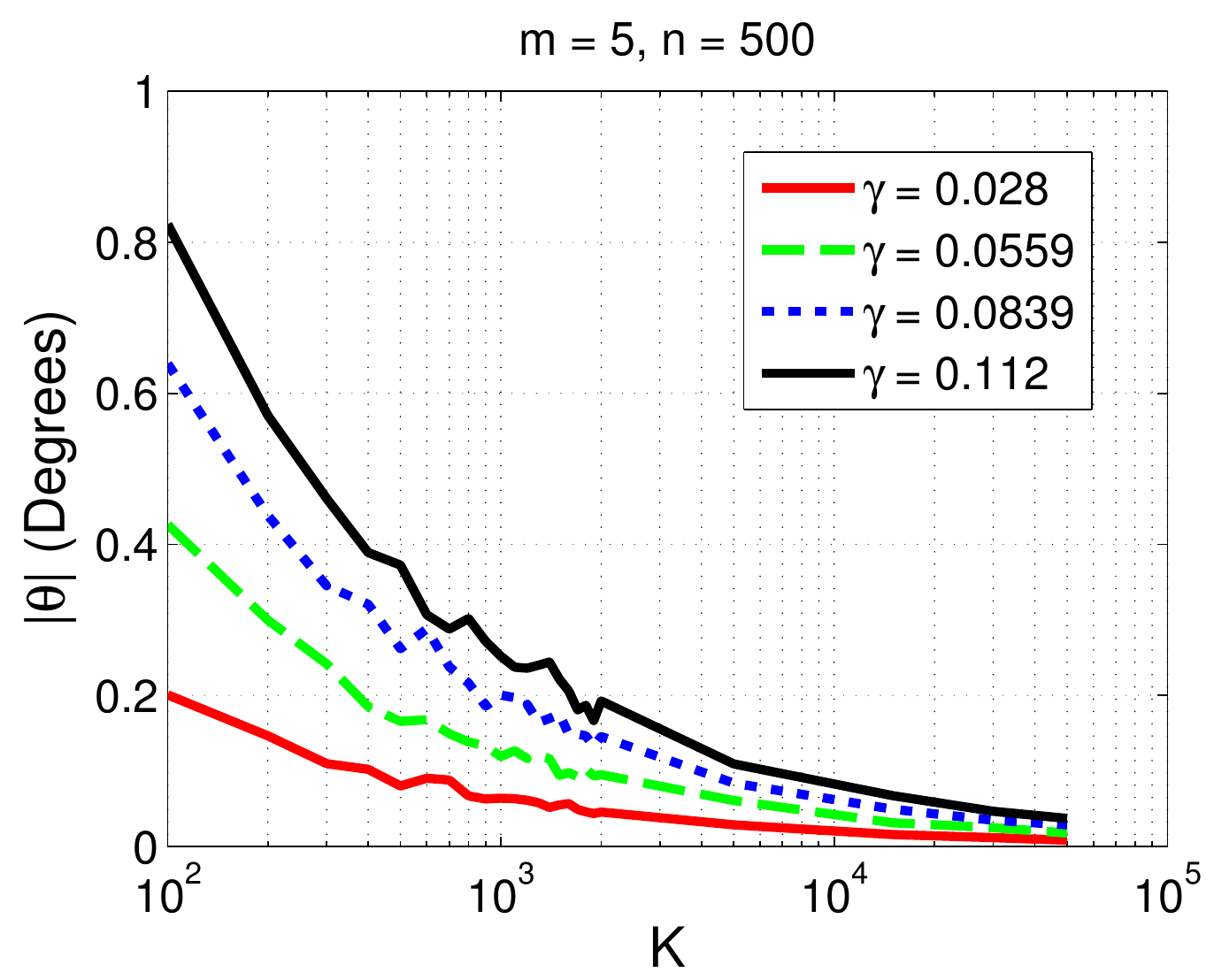}\end{minipage}}
\subfloat[\small m=5,n=1000]{
\begin{minipage}[c]{0.28\linewidth}
\centering
\label{fig:smooth2_emp_m5n1000} 
\noindent \includegraphics[width=1.0\linewidth]{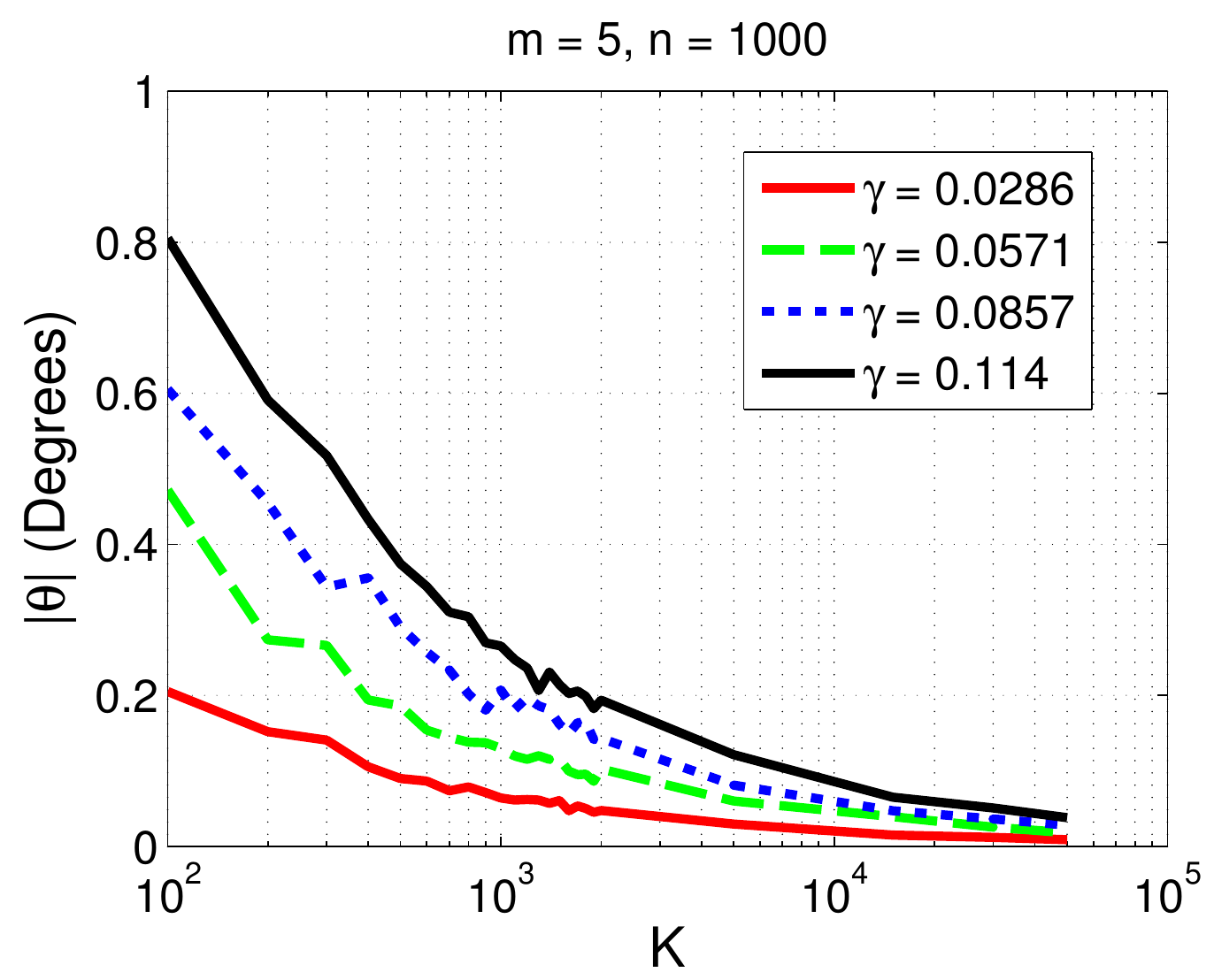}\end{minipage}}\\

\caption{\small Variation of the deviation $\abs{\theta}$ with respect to $K$ for different sampling widths $\nu = \gamma \nu_{\text{bound,quad}}$ for smooth mapping 2. Figures \ref{fig:smooth2_theo_m5n100}-\ref{fig:smooth2_theo_m5n1000} show theoretical plots while Figures \ref{fig:smooth2_emp_m5n100}-\ref{fig:smooth2_emp_m5n1000} show empirical plots.}
\label{fig:smooth2_theo_emp_exp1} 
\end{figure}
\begin{figure}[!htbp]
\centering
\subfloat[\small m=5,n=100]{
\begin{minipage}[c]{0.28\linewidth}
\centering
\label{fig:smooth3_theo_m5n100} 
\noindent \includegraphics[width=1.0\linewidth]{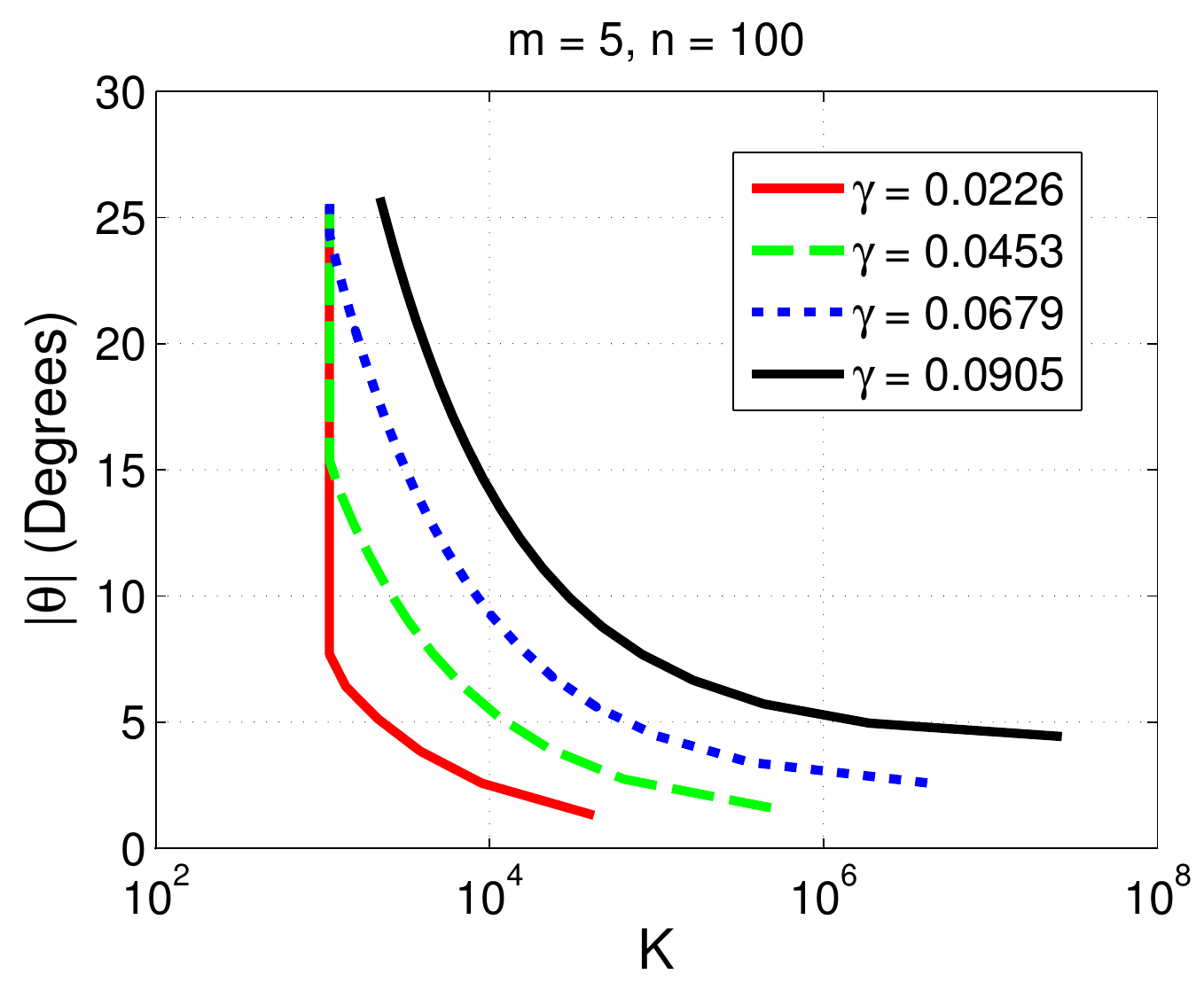}\end{minipage}}
\subfloat[\small m=5,n=500]{
\begin{minipage}[c]{0.28\linewidth}
\centering
\label{fig:smooth3_theo_m5n500} 
\noindent \includegraphics[width=1.0\linewidth]{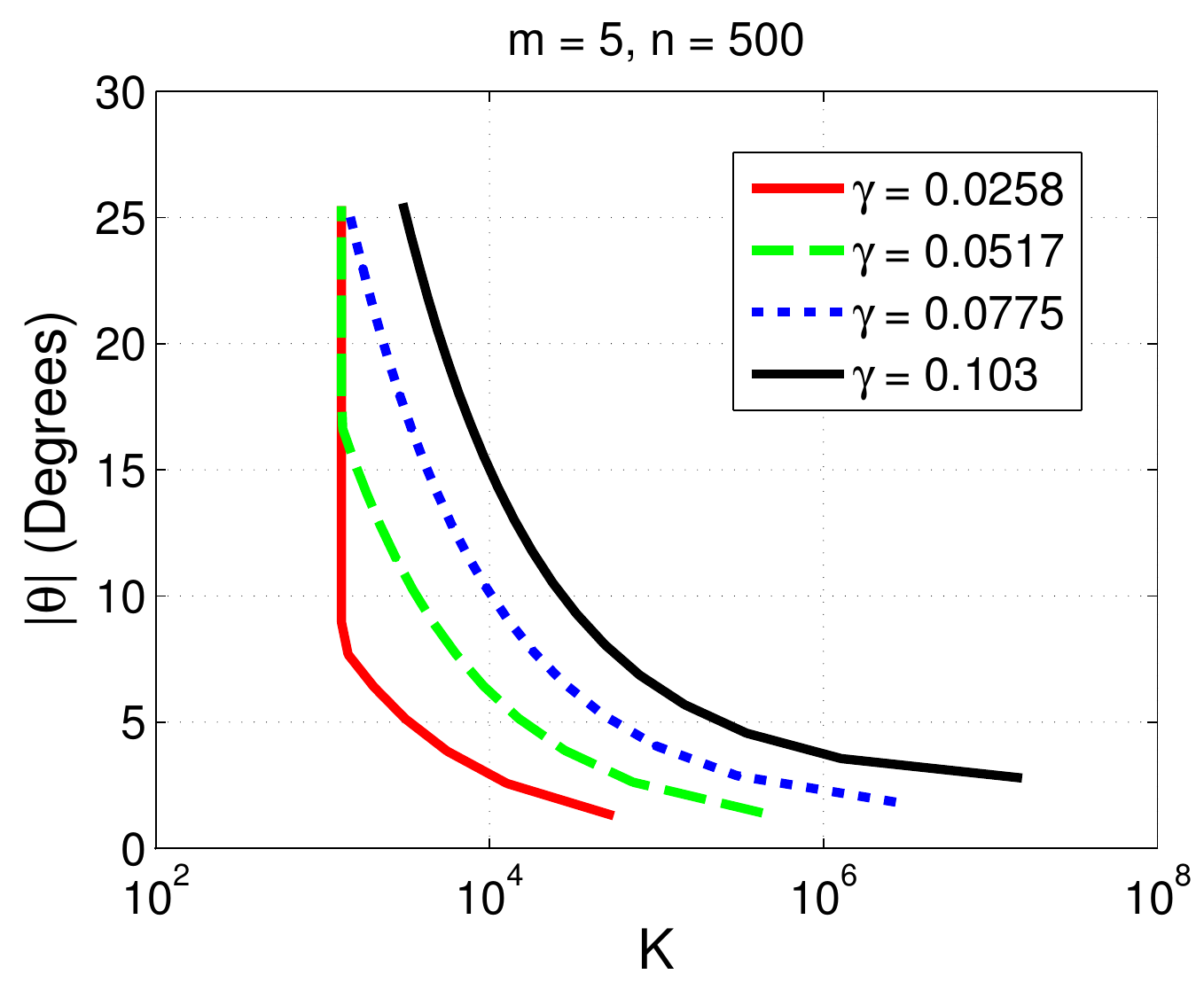}\end{minipage}}
\subfloat[\small m=5,n=1000]{
\begin{minipage}[c]{0.28\linewidth}
\centering
\label{fig:smooth3_theo_m5n1000} 
\noindent \includegraphics[width=1.0\linewidth]{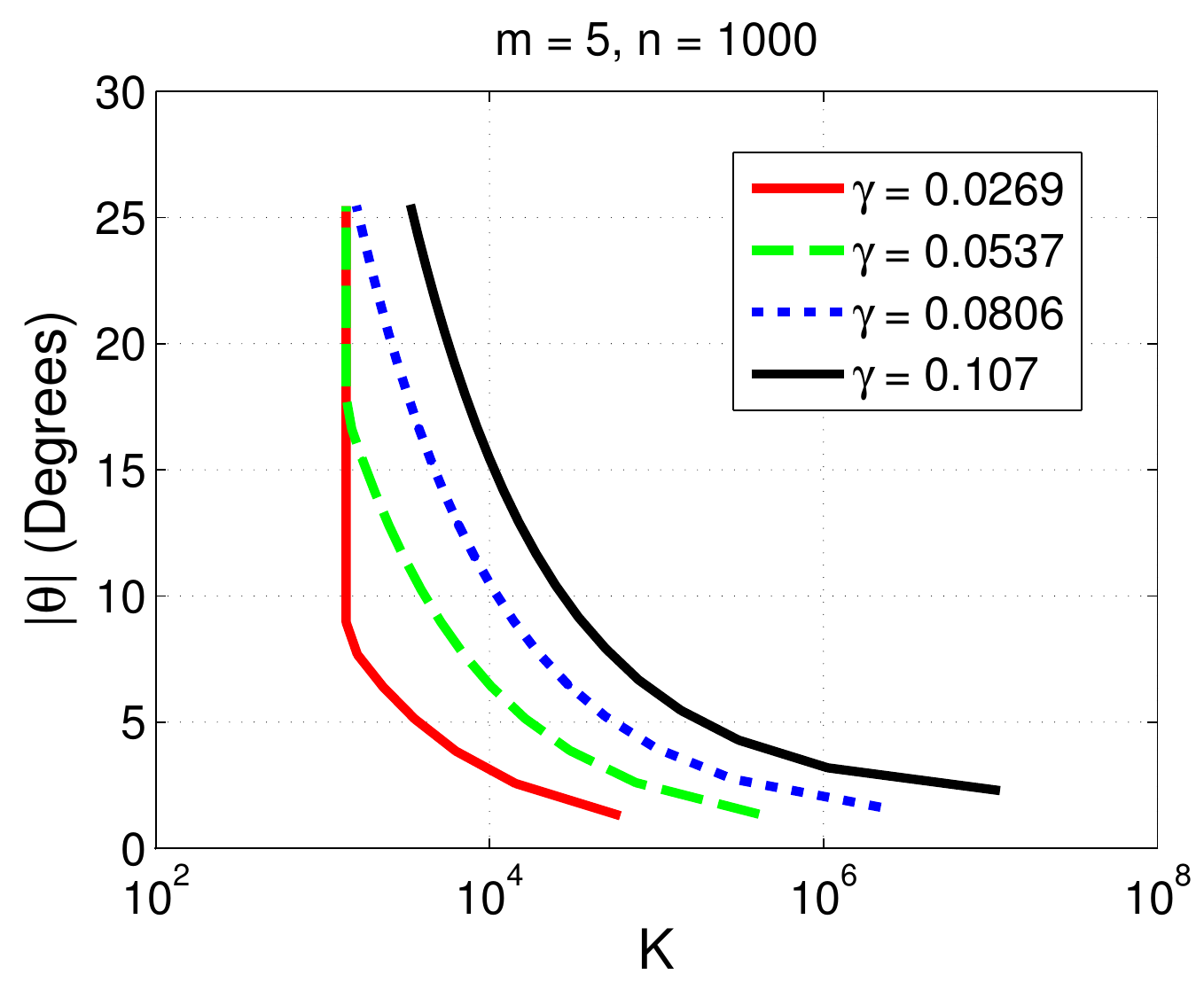}\end{minipage}}\\
\subfloat[\small m=5,n=100]{
\begin{minipage}[c]{0.28\linewidth}
\centering
\label{fig:smooth3_emp_m5n100} 
\noindent \includegraphics[width=1.0\linewidth]{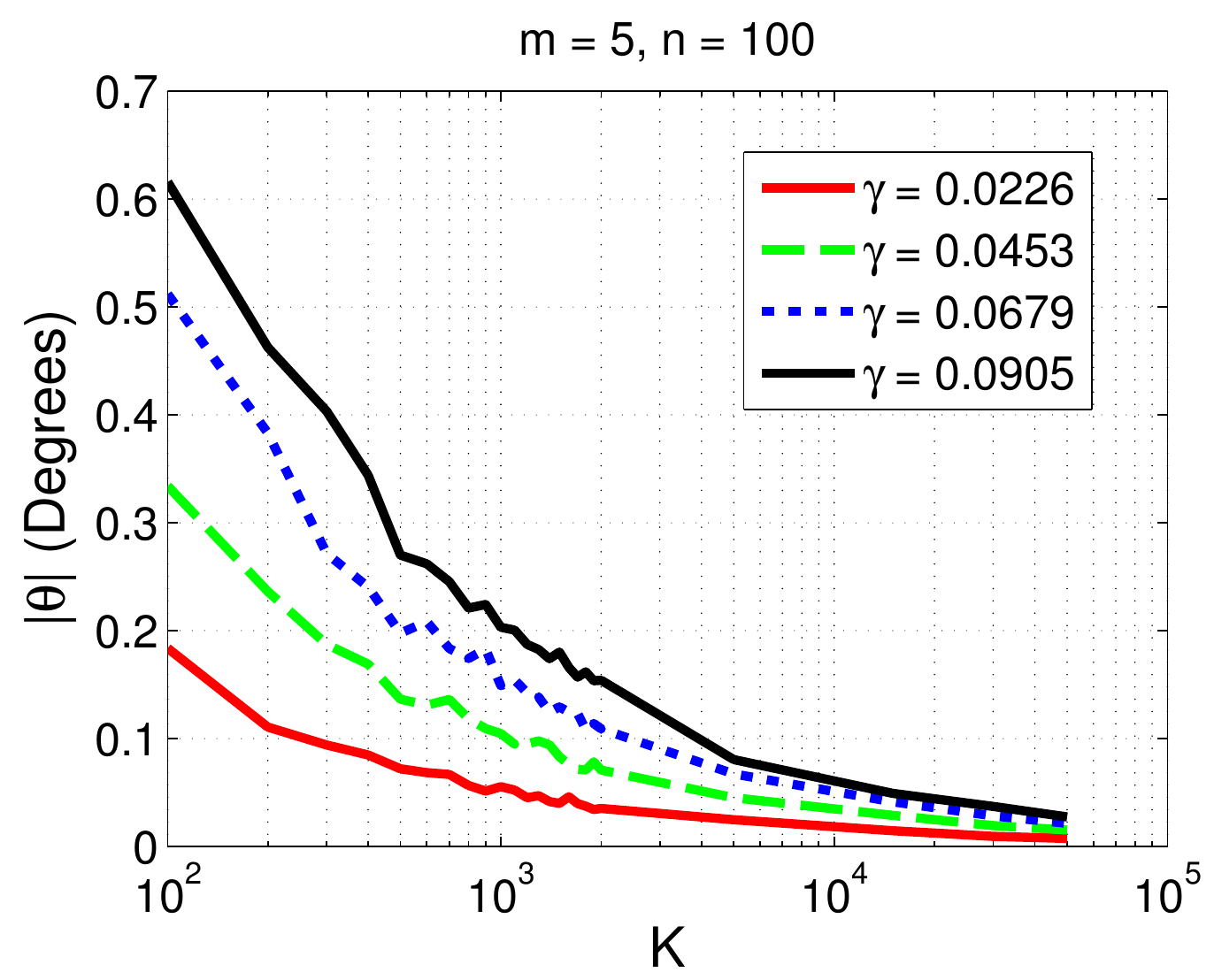}\end{minipage}}
\subfloat[\small m=5,n=500]{
\begin{minipage}[c]{0.28\linewidth}
\centering
\label{fig:smooth3_emp_m5n500} 
\noindent \includegraphics[width=1.0\linewidth]{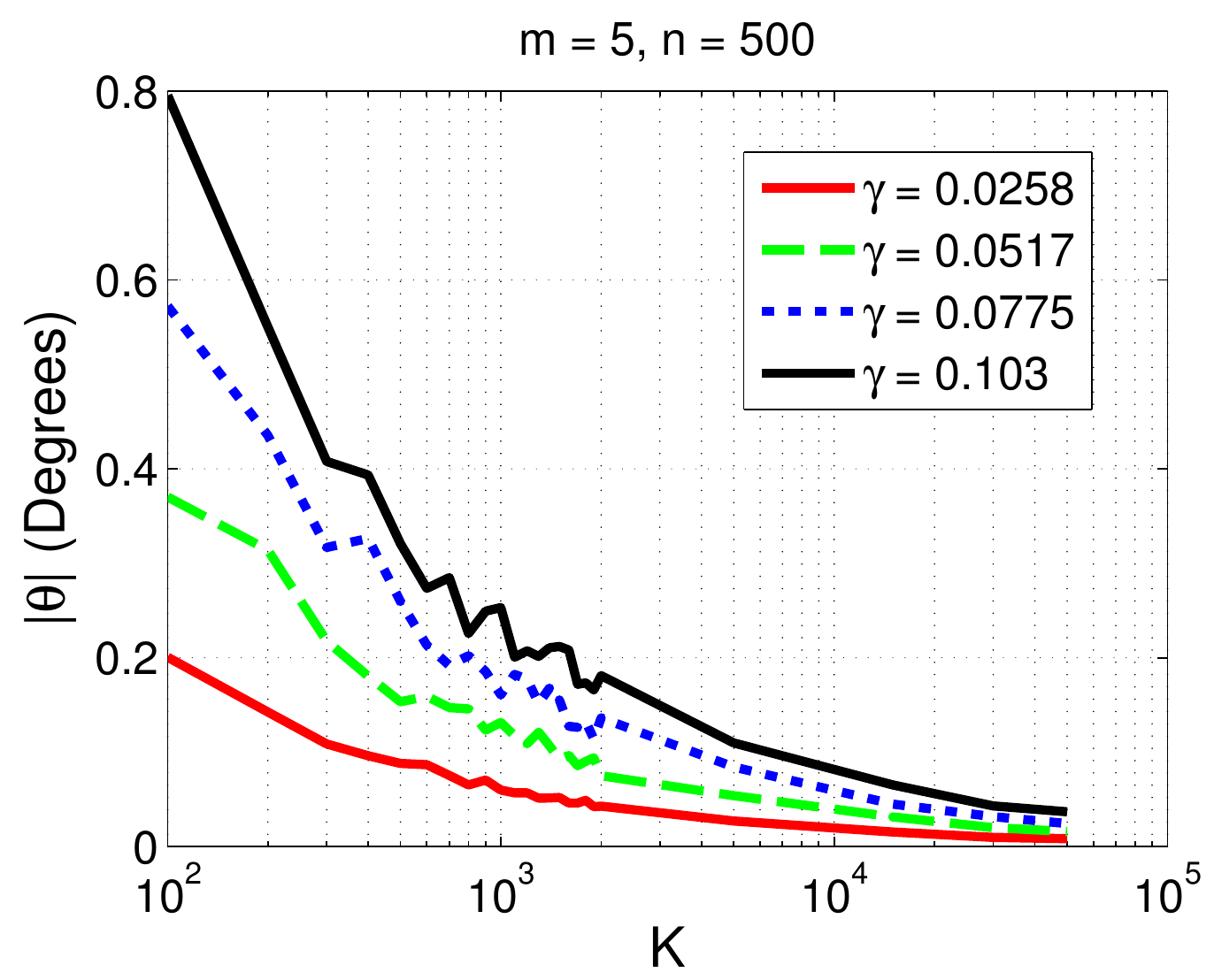}\end{minipage}}
\subfloat[\small m=5,n=1000]{
\begin{minipage}[c]{0.28\linewidth}
\centering
\label{fig:smooth3_emp_m5n1000} 
\noindent \includegraphics[width=1.0\linewidth]{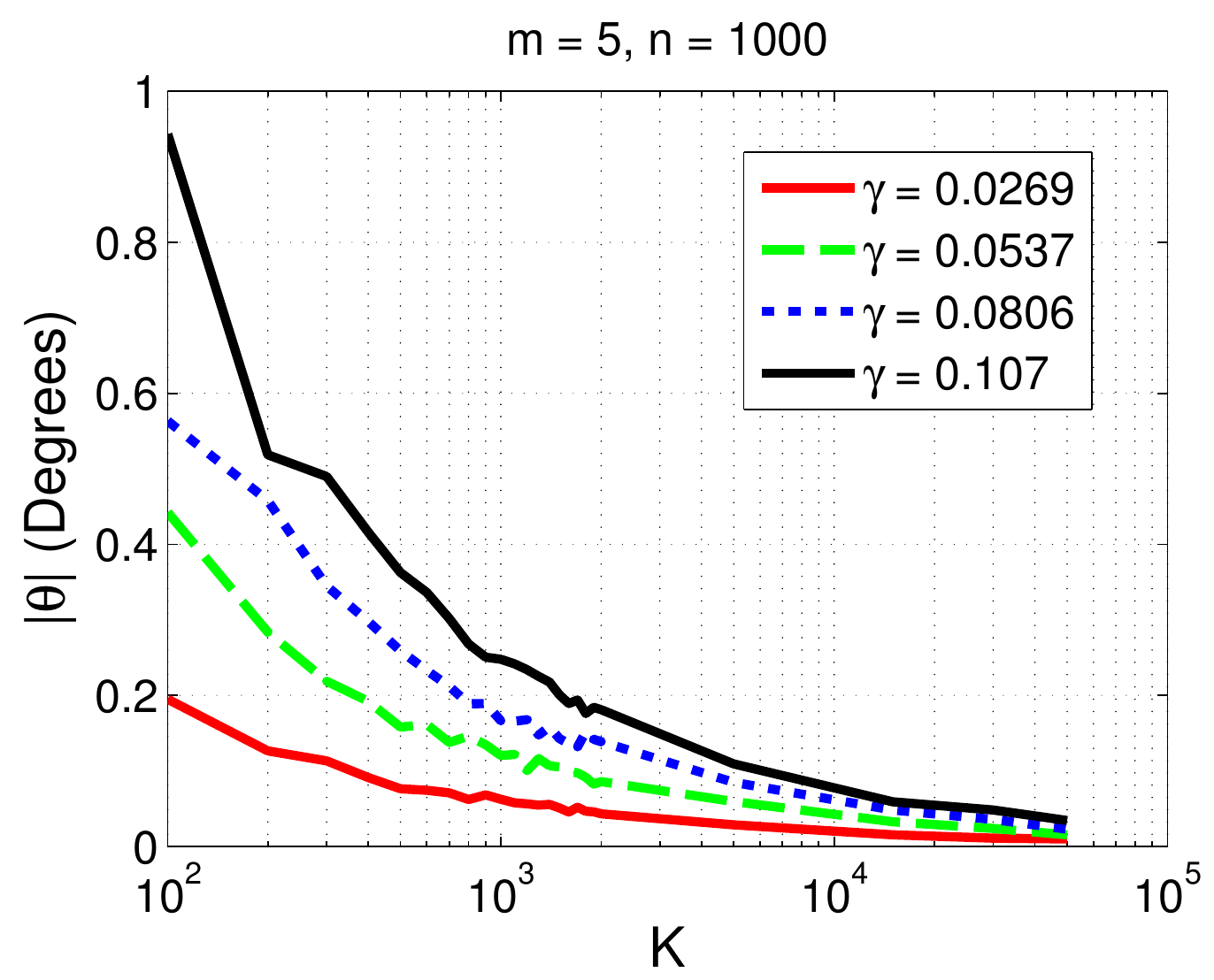}\end{minipage}}\\

\caption{\small Variation of the deviation $\abs{\theta}$ with respect to $K$ for different sampling widths $\nu = \gamma \nu_{\text{bound,quad}}$ for smooth mapping 3. Figures \ref{fig:smooth3_theo_m5n100}-\ref{fig:smooth3_theo_m5n1000} show theoretical plots while Figures \ref{fig:smooth3_emp_m5n100}-\ref{fig:smooth3_emp_m5n1000} show empirical plots.}
\label{fig:smooth3_theo_emp_exp1} 
\end{figure}

In the second set of experiments, we study the scaling of the true bound on the sampling width $\nu$ with the ambient space dimension $n$. To this end, we fix $m = 5, \, \kfmax = 10$ and the number of samples is fixed at a sufficiently large value, i.e., $K= 2000$. Then, we vary $n$ from 100 to 1000 in steps of 50. For each value of $n$, we first initialize $\nu = 3\, \nu_{\text{bound,quad}}$ and then compute the bound on the sampling width by gradually reducing $\nu$ until $\abs{\theta} < \theta_{\text{bound}}$.  The value of $\abs{\theta}$ is averaged over 25 random trials. We obtain four plots corresponding to the angle bounds $\theta_{\text{bound}} = 5^{\circ},10^{\circ},15^{\circ}$ and $20^{\circ}$, for the quadratic form and smooth mappings 1 and 2. 

Fig.~\ref{fig:md_exp2} shows the variation of $\nu$ with $n$. Importantly, we have observed that, for quadratic forms (Fig.~\ref{fig:md_exp2_quad}), the true bound $\nu$ on the sampling width is in line with its theoretical estimation $\nu_{\text{bound,quad}}$ as $\nu = \gamma \, \nu_{\text{bound,quad}}$, where $\gamma$ lies approximately between $1.38$ and $1.46$. This is also true for smooth mappings (Figures \ref{fig:md_exp2_gauss}, \ref{fig:md_exp2_sin}) indicating that the true bound on $\nu$ is approximately $O(n^{-1/2})$. It can also be observed that, at a fixed value of $n$, the bound on $\nu$ is larger when the angle bound $\theta_{\text{bound}}$ is greater, as expected.

\begin{figure}[!htbp]
\centering
\subfloat[\small Quadratic form]{
\begin{minipage}[c]{0.28\linewidth}
\centering
\label{fig:md_exp2_quad} 
\noindent \includegraphics[width=1.0\linewidth]{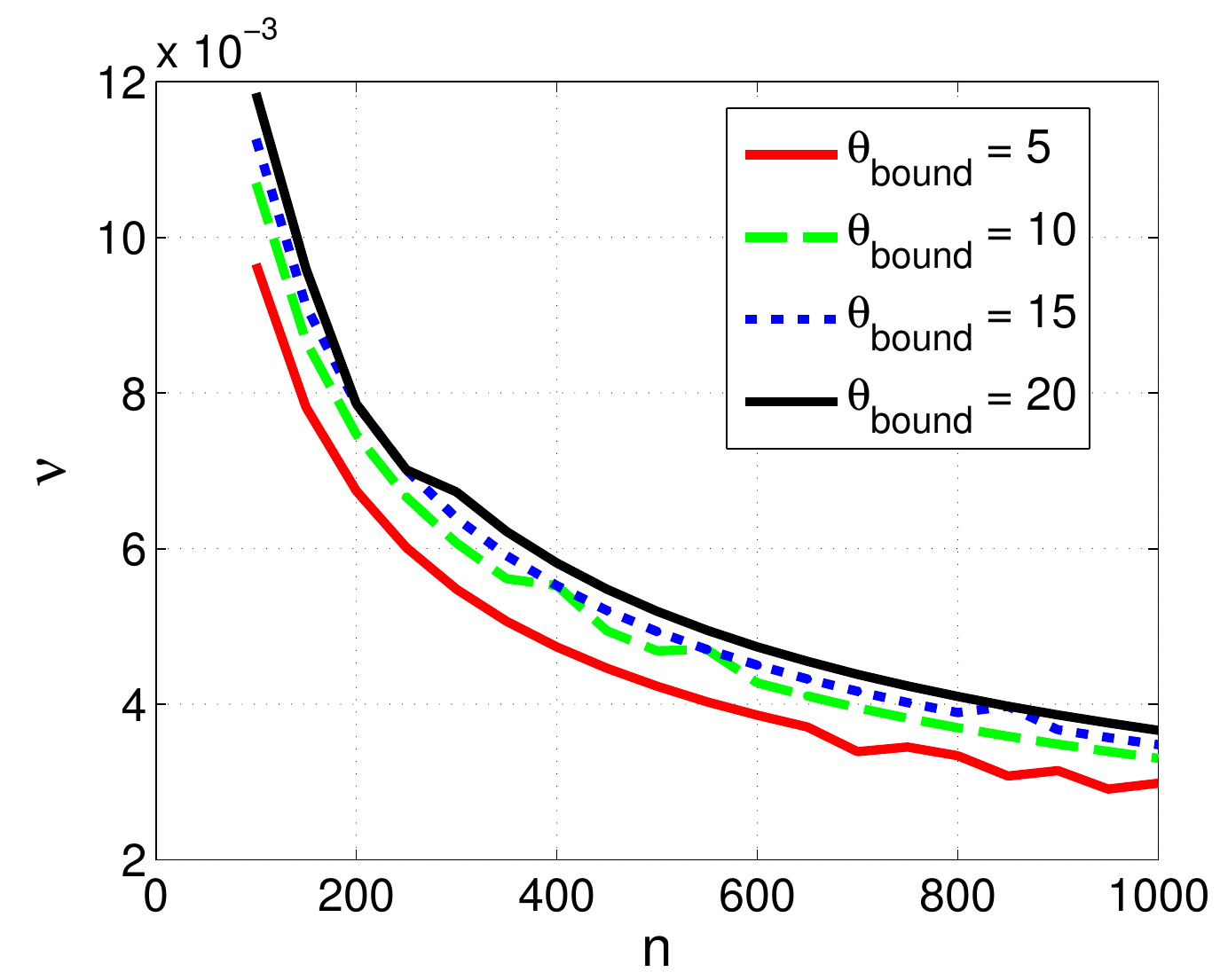} \end{minipage}}
\subfloat[Smooth mapping 1]{
\begin{minipage}[c]{0.28\linewidth}
\centering
\label{fig:md_exp2_gauss} 
\includegraphics[width=1.0\linewidth]{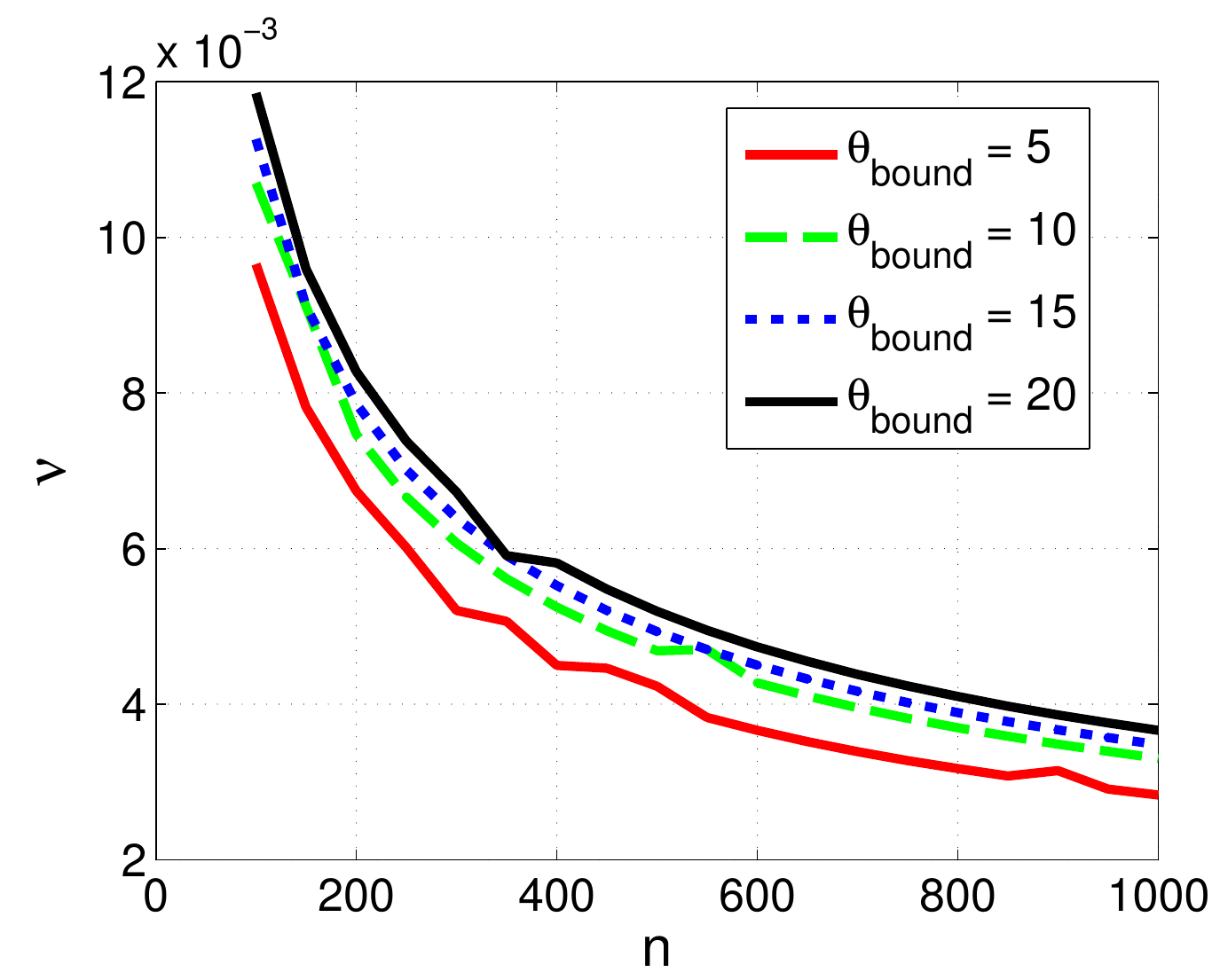} \end{minipage}}
\subfloat[Smooth mapping 2]{
\begin{minipage}[c]{0.28\linewidth}
\centering
\label{fig:md_exp2_sin} 
\includegraphics[width=1.0\linewidth]{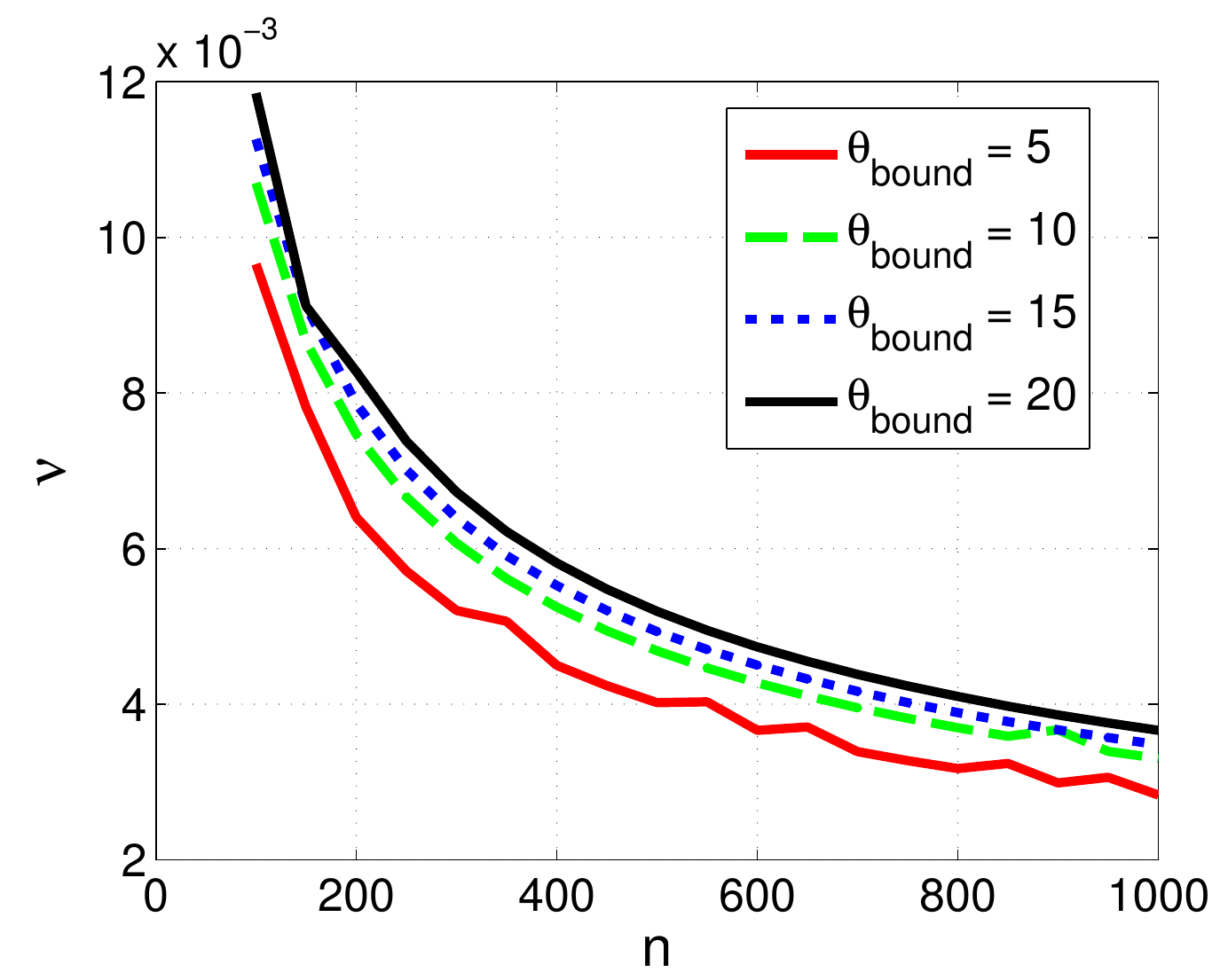} \end{minipage}}
\caption{\small Maximum sampling width $\nu$ for which the deviation $\abs{\theta} < \theta_{\text{bound}}$ is achieved for different values of $n$. Plots are shown for $\theta_{\text{bound}} = 5^{\circ},10^{\circ},15^{\circ},20^{\circ}$.}
\label{fig:md_exp2} 
\end{figure}

In the next experiment, we are interested in observing the dependency of the true sampling width bound $\nu$ on the maximum local curvature $\kfmax$. We fix $m = 5$, set $K = 2000$ and choose a fixed $n \in \set{100,500,1000}$. We vary $\kfmax$ from 0.5 to 10 in steps of 0.5. For each value of $\kfmax$, we first initialize $\nu = 3\, \nu_{\text{bound,quad}}$ and then compute the bound on $\nu$ by gradually reducing $\nu$ until $\abs{\theta} < \abs{\theta_{\text{bound}}} = 5^{\circ}$. The value of $\abs{\theta}$ is averaged over 25 random trials.

Fig.~\ref{fig:md_exp3} shows the dependency of $\nu$ on $\kfmax$ for different values of $n$. Similarly to the previous experiments, we observe that for quadratic forms, $\nu = \gamma \, \nu_{\text{bound,quad}}$, where $\gamma$ is approximately between 1.32 and 1.46 (see Fig. \ref{fig:md_exp3_quad}). We note the same behavior for the case of smooth mappings shown in Figures \ref{fig:md_exp3_gauss} and \ref{fig:md_exp3_sin}. The results indicate that for fixed values of $m$ and $n$, the true bound on $\nu$ matches the theoretical result $\nu \approx O(\abs{\kfmax}^{-1})$ derived in Section \ref{subsec:smooth_sampl_compl}.
\begin{figure}[!htbp]
\centering
\subfloat[\small Quadratic form]{
\begin{minipage}[c]{0.28\linewidth}
\centering
\label{fig:md_exp3_quad} 
\noindent \includegraphics[width=1.0\linewidth]{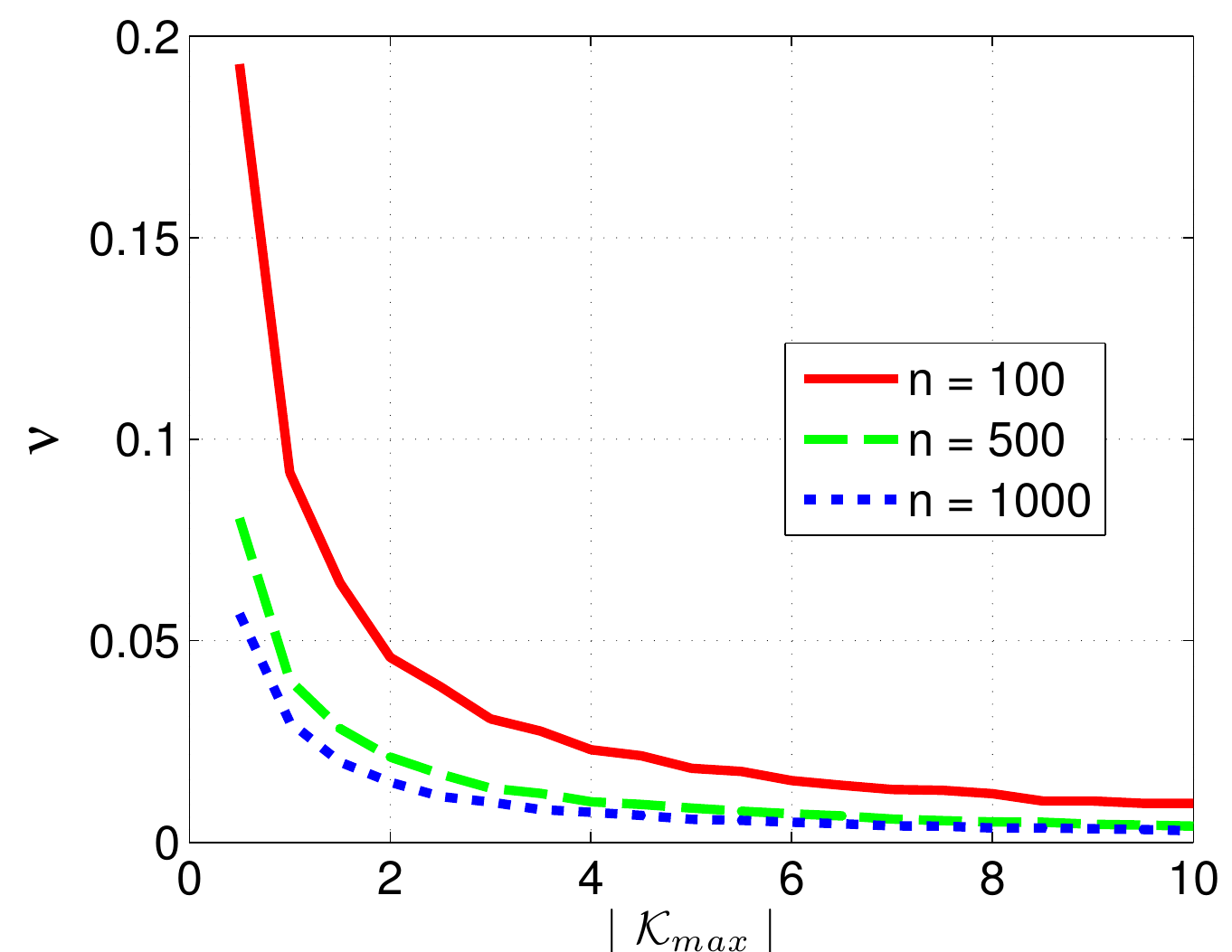} \end{minipage}}
\subfloat[Smooth mapping 1]{
\begin{minipage}[c]{0.28\linewidth}
\centering
\label{fig:md_exp3_gauss} 
\includegraphics[width=1.0\linewidth]{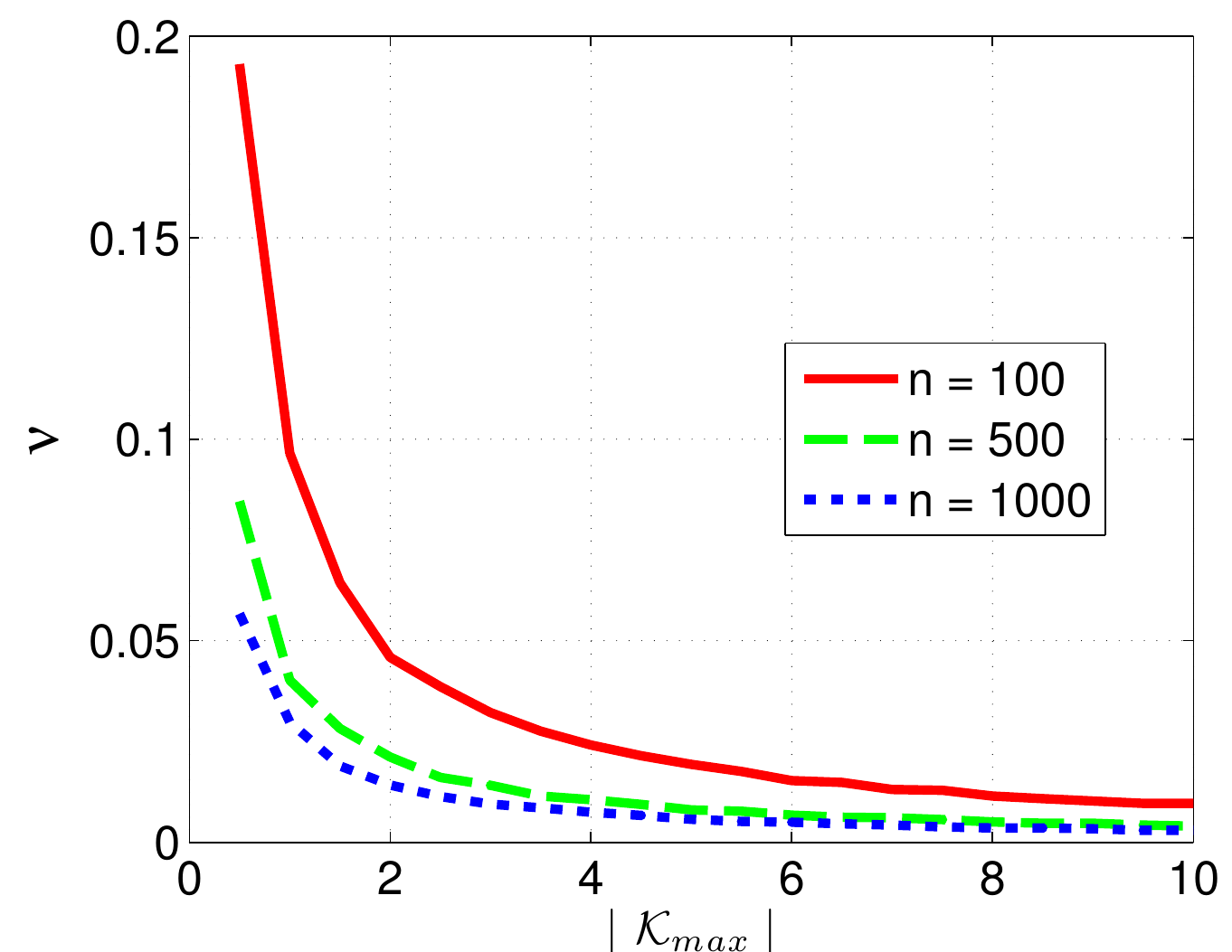} \end{minipage}}
\subfloat[Smooth mapping 2]{
\begin{minipage}[c]{0.28\linewidth}
\centering
\label{fig:md_exp3_sin} 
\includegraphics[width=1.0\linewidth]{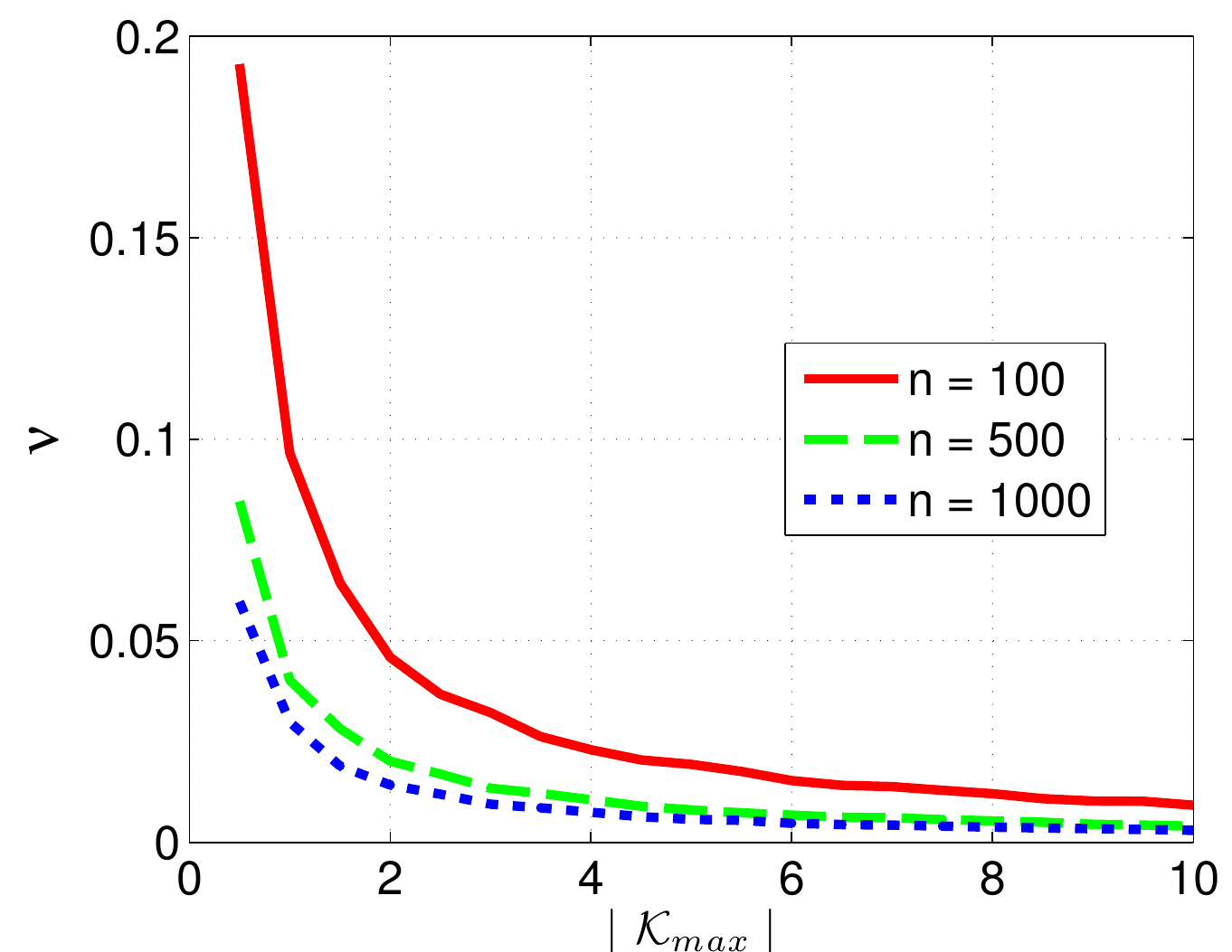} \end{minipage}}
\caption{\small Maximum sampling width $\nu$ for which $\abs{\theta} < \abs{\theta_{\text{bound}}} = 5^{\circ}$ is achieved for different values of $\kfmax$. The results are given for different dimensions of the ambient space $n$.}
\label{fig:md_exp3} 
\end{figure}

Then, we investigate the relation between the sampling density $K$ and the embedding dimension $n$ for a fixed  sampling width $\nu$. We choose $\kfmax = 10$ and select several values for the dimension of the manifold $m \in \set{5,10,15}$. We vary $n$ from 100 to 2000 in steps of 100. We set $\nu = \nu_{\text{bound,quad}}$, where $\nu_{\text{bound,quad}}$ is evaluated at the largest value of $n$ such that the fixed sampling width $\nu$ is sufficiently small for the range of $n$ under consideration.  We denote the largest value of $n$ by $n_{\text{large}}$ in this experiment. For each value of $(m,n,\kfmax)$, we compute the minimum number of samples needed in order to have $\abs{\theta} < \abs{\theta_{\text{bound}}} = 5^{\circ}$. The value of $\abs{\theta}$ is the average of 25 random trials. 
\begin{figure}[!htbp]
\centering
\subfloat[\small Quadratic form]{
\begin{minipage}[c]{0.28\linewidth}
\centering
\label{fig:md_exp4_quad} 
\noindent \includegraphics[width=1.0\linewidth]{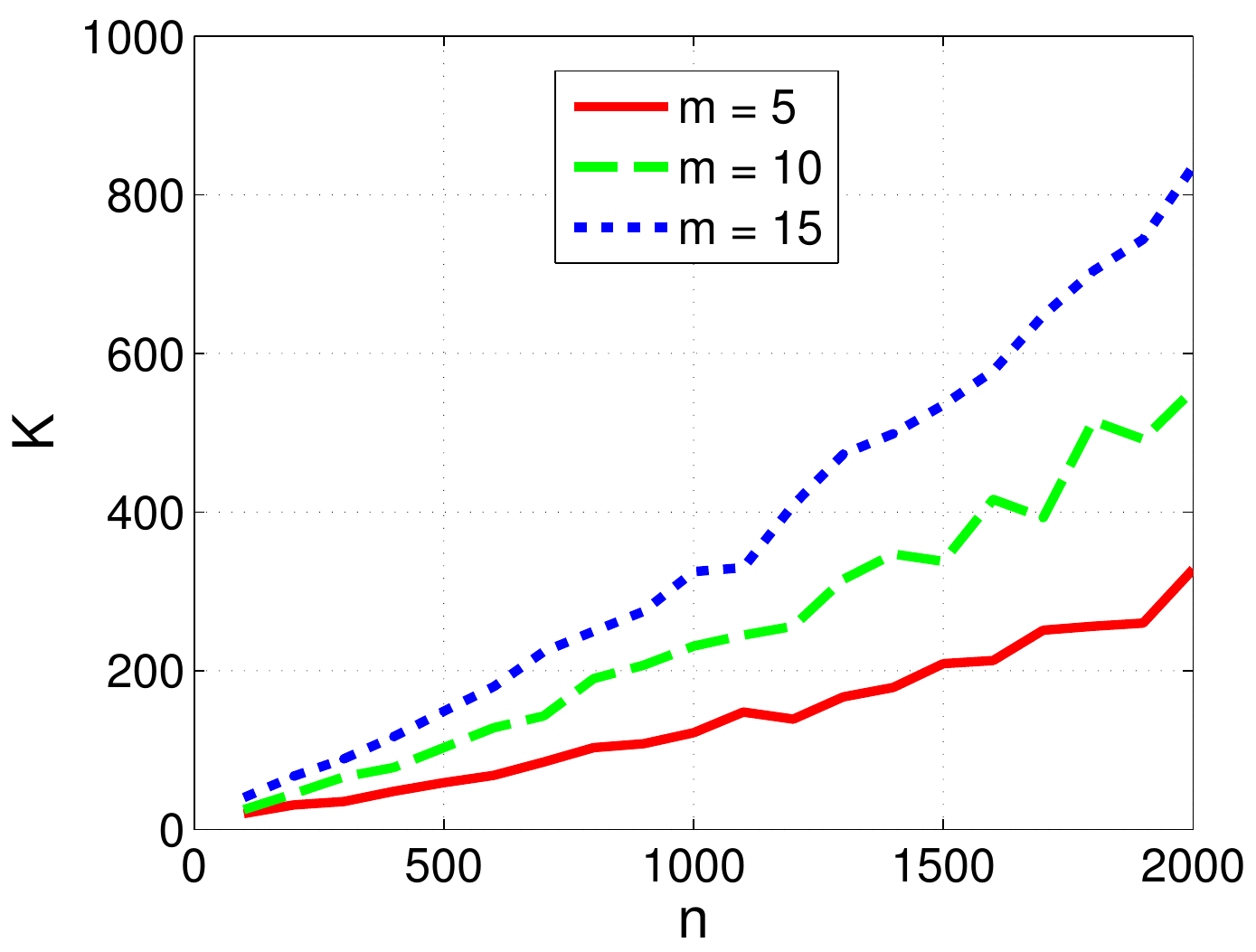} \end{minipage}}
\subfloat[Smooth mapping 1]{
\begin{minipage}[c]{0.28\linewidth}
\centering
\label{fig:md_exp4_gauss} 
\includegraphics[width=1.0\linewidth]{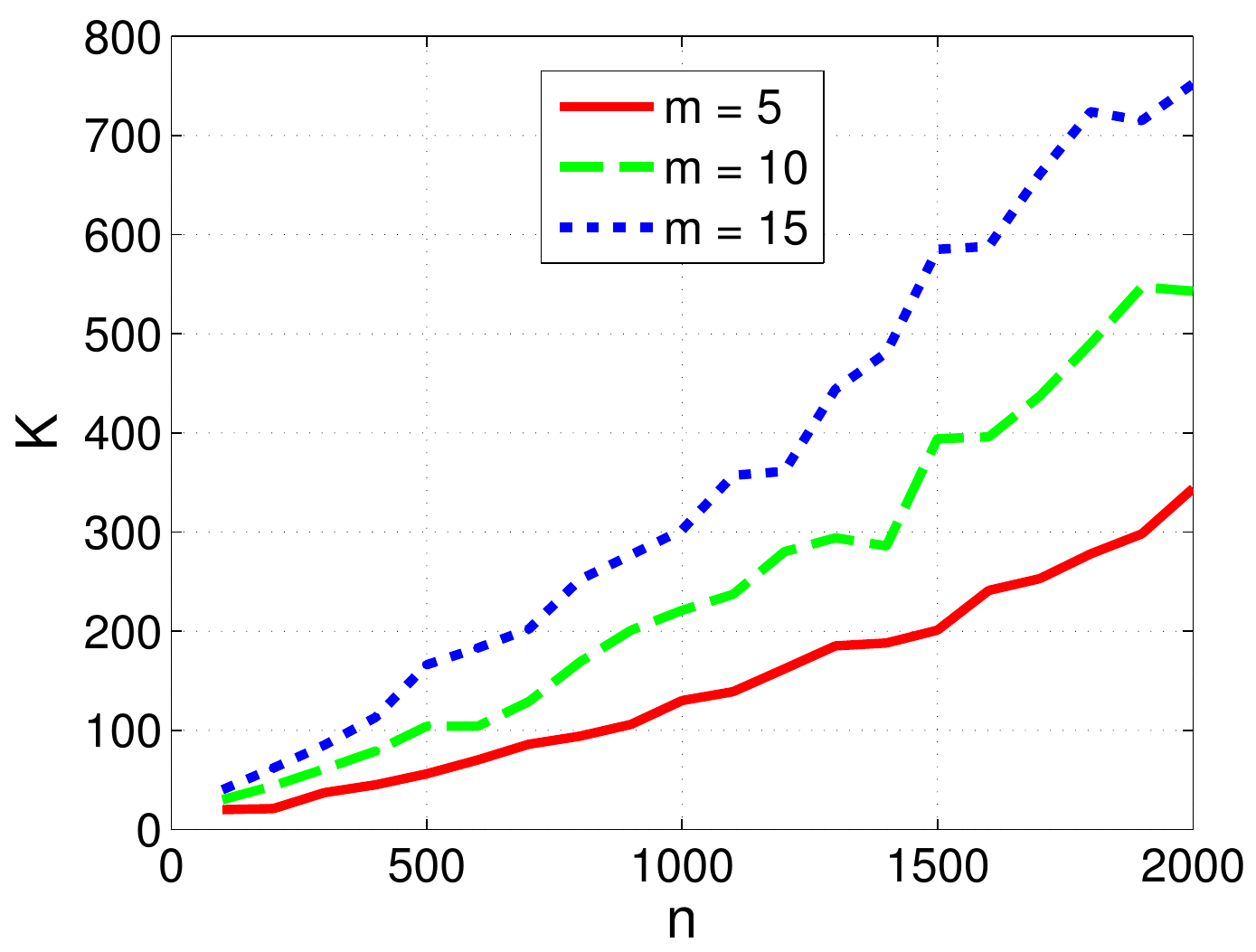} \end{minipage}}
\subfloat[Smooth mapping 2]{
\begin{minipage}[c]{0.28\linewidth}
\centering
\label{fig:md_exp4_sin} 
\includegraphics[width=1.0\linewidth]{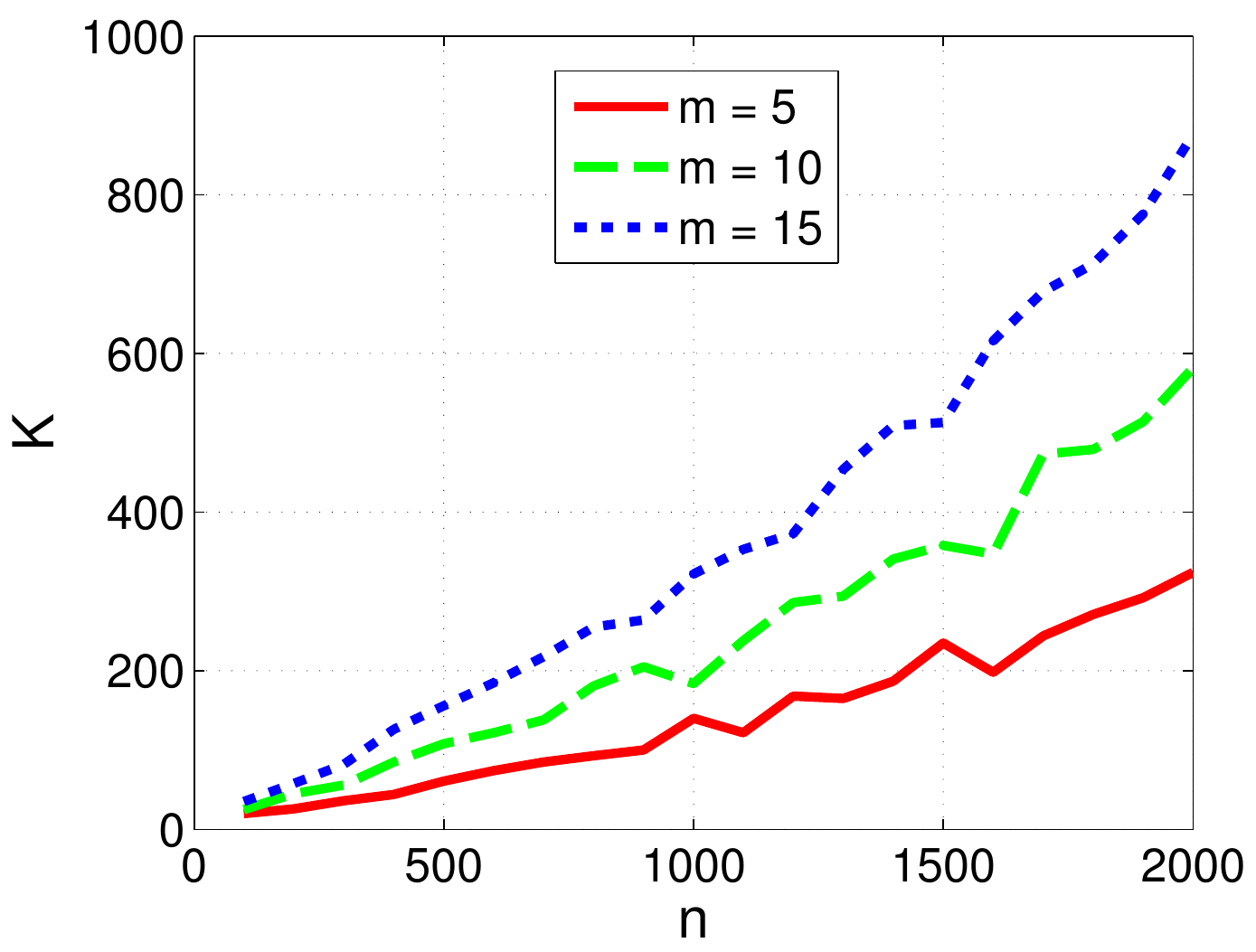} \end{minipage}}
\caption{\small Minimum number of samples $K$ for which $\abs{\theta} < \abs{\theta_{\text{bound}}} = 5^{\circ}$ is achieved as $n$ is varied. The results obtained for different dimensions of the manifold $m$.}
\label{fig:md_exp4} 
\end{figure}

Fig.~\ref{fig:md_exp4} shows the variation of $K$ with respect to $n$ for the different mappings. We see that $K$ increases with the ambient dimension $n$ as expected. Furthermore, for a given $n$, we observe that increasing the dimension of the manifold increases $K$. We now show that this behaviour is well explained by our theoretical results in Section \ref{sec:mDimSurfaces}. In order to see this, we first note that $\nu = O(n_{\text{large}}^{-1/2} m^{-1} \kfmax^{-1})$, which is due to the relation $\nu = O(n^{-1/2} m^{-1} \kfmax^{-1})$   derived in Section \ref{subsec:quad_sampl_compl} and the fact that we evaluate $\nu$ at $n=n_{\text{large}}$. Using this value of $\nu$ in the bounds on the sampling density stated in Lemma \ref{lemma:md_k_bound_eps}, one can easily verify that 
\begin{align*}
K_{\text{bound}}^{(1)} &= O\left(\left(m + \frac{n}{n_{\text{large}}}\right)\log n\right), \quad K_{\text{bound}}^{(2)} = O(\log n) \quad \text{and} \\
K_{\text{bound}}^{(3)} &= O\left(\left(\frac{nm}{n_{\text{large}}} + \sqrt{\frac{nm}{n_{\text{large}}}}\right)\log n\right) \approx O\left(m\sqrt{n/n_{\text{large}}} \log n\right).
\end{align*}
Since $n < n_{\text{large}}$, we obtain $K_{\text{bound}} = O\left(\left(m + \frac{n}{n_{\text{large}}}\right)\log n\right)$. This closely matches the behavior shown in Figures \ref{fig:md_exp4_quad}-\ref{fig:md_exp4_sin}.

Finally, in the last experiment, we would like to look into the dependency of the sampling density $K$ on the curvature term $\kfmax$ for a fixed sampling width $\nu$. We set $m = 5$ and pick several values of the ambient space dimension, i.e., $n \in \set{100,500,1000}$. We vary $\kfmax$ from 0.5 to 10 in steps of 0.5. We fix the value of the sampling width as $\nu = \nu_{\text{bound,quad}}$, where $\nu_{\text{bound,quad}}$ is evaluated at the largest value of $\kfmax$ (denoted as $\mathcal{K}_{\text{max,large}}$ in this experiment). For each value of $(m,n,\kfmax)$, we compute the minimum number of samples required in order to have $\abs{\theta} < \abs{\theta_{\text{bound}}} = 5^{\circ}$. The value of $\abs{\theta}$ is averaged over 25 random trials.
\begin{figure}[!htbp]
\centering
\subfloat[\small Quadratic form]{
\begin{minipage}[c]{0.28\linewidth}
\centering
\label{fig:md_exp5_quad} 
\noindent \includegraphics[width=1.0\linewidth]{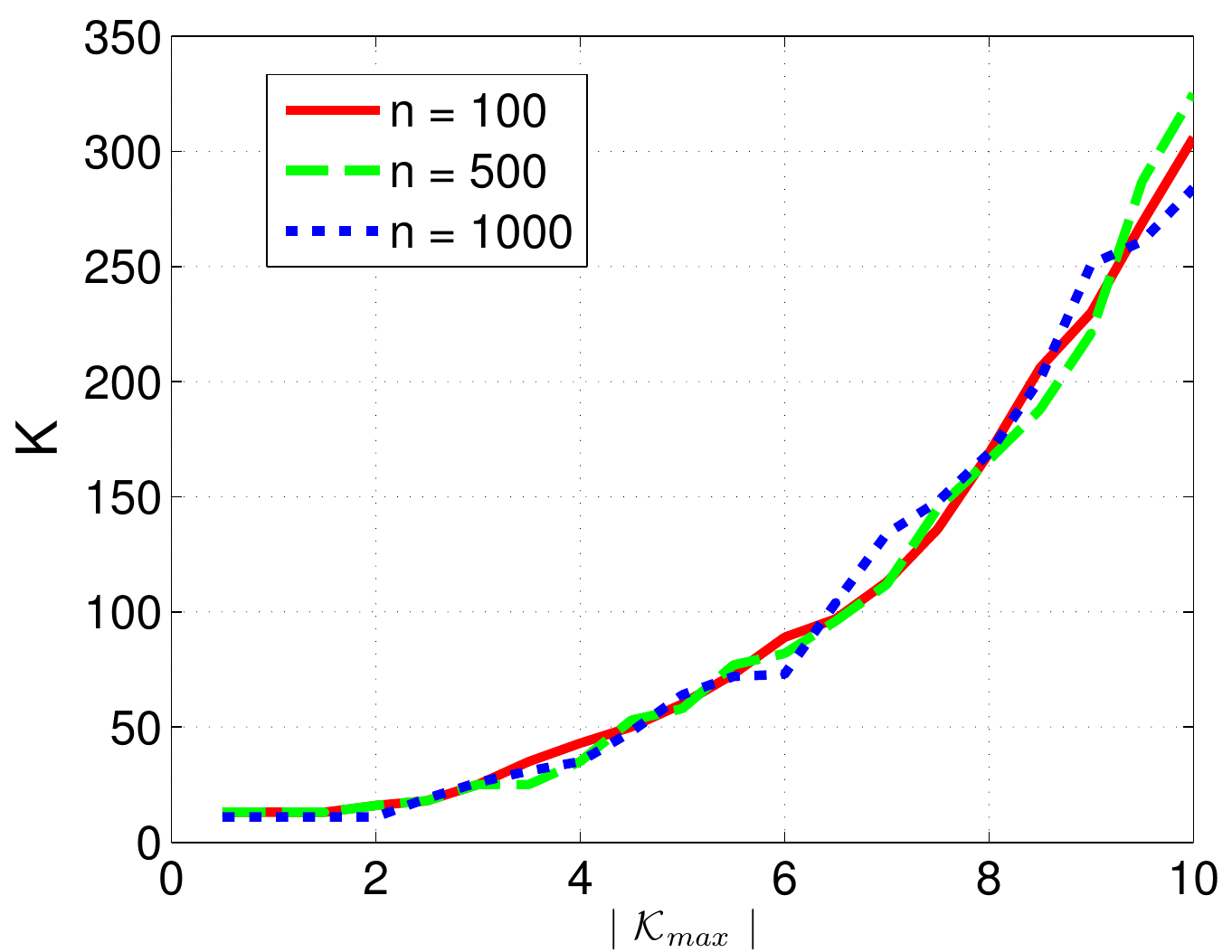} \end{minipage}}
\subfloat[Smooth mapping 1]{
\begin{minipage}[c]{0.28\linewidth}
\centering
\label{fig:md_exp5_gauss} 
\includegraphics[width=1.0\linewidth]{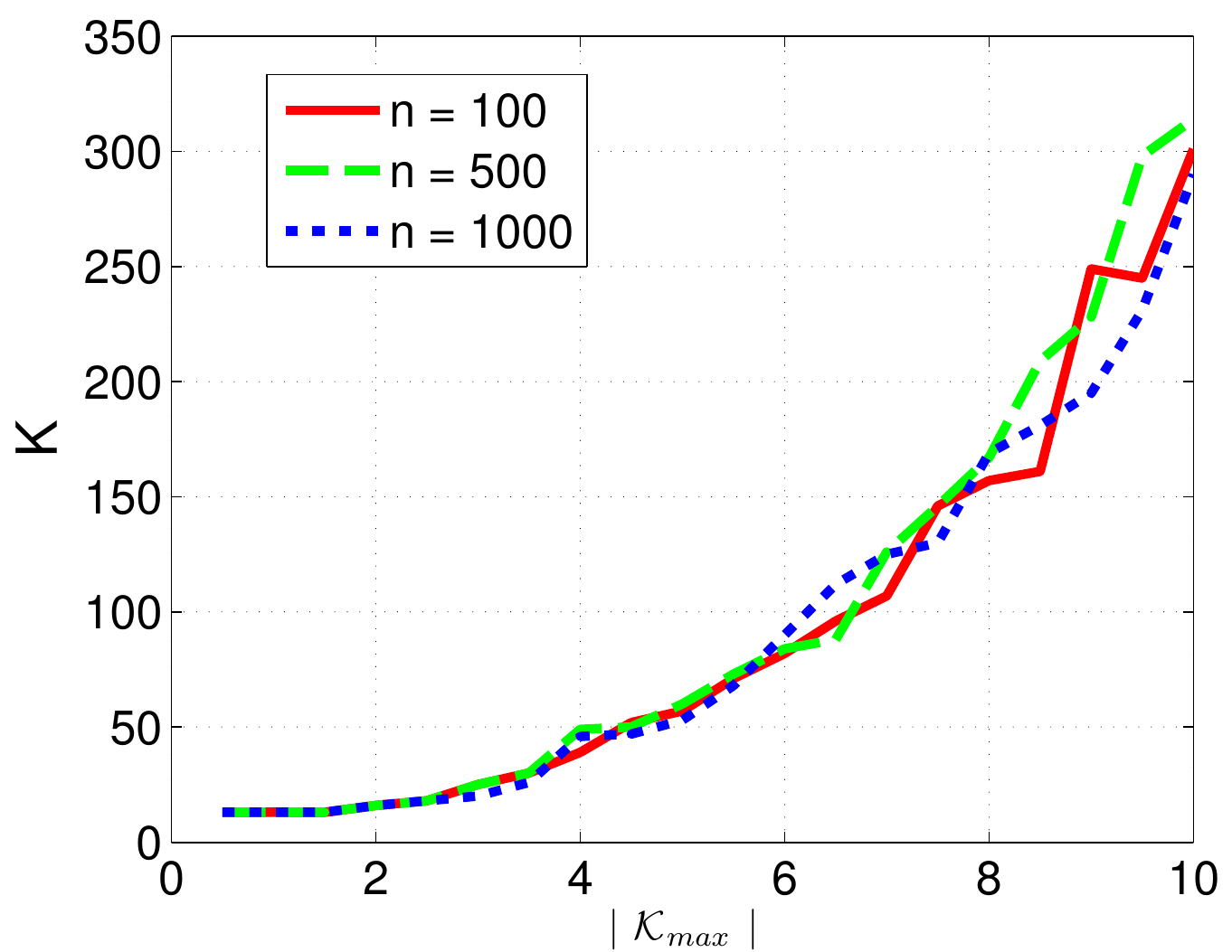} \end{minipage}}
\subfloat[Smooth mapping 2]{
\begin{minipage}[c]{0.28\linewidth}
\centering
\label{fig:md_exp5_sin} 
\includegraphics[width=1.0\linewidth]{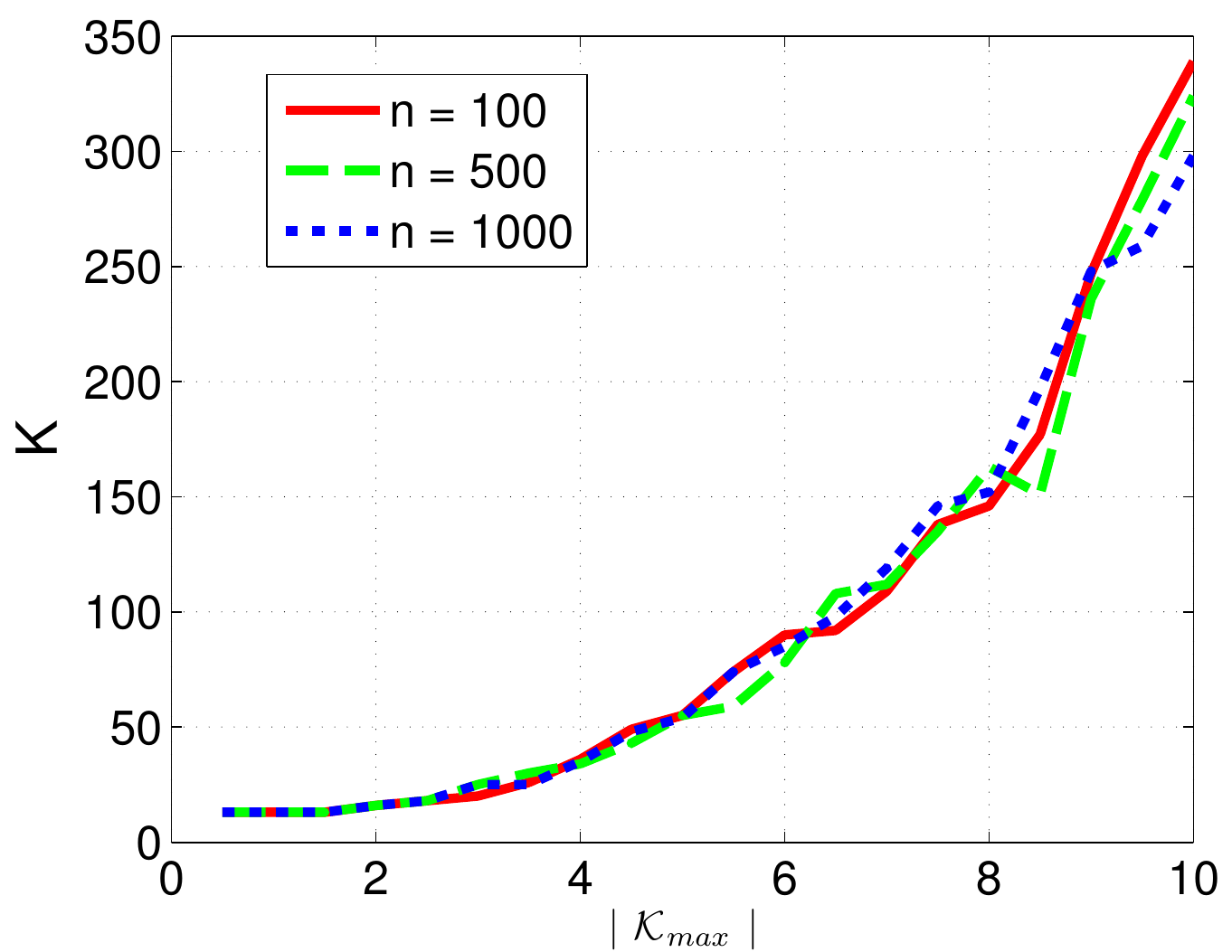} \end{minipage}}
\caption{\small Minimum number of samples $K$ for which $\abs{\theta} < \abs{\theta_{\text{bound}}} = 5^{\circ}$ is achieved as $\kfmax$ is varied. The results are given for different values of $n$.}
\label{fig:md_exp5} 
\end{figure}

Fig.~\ref{fig:md_exp5} shows the relation between $K$ and $\abs{\kfmax}$ for the different mappings. We see that $K$ increases with $\abs{\kfmax}$ as expected. Interestingly, we see that, for a fixed value of $\kfmax$, a change in the embedding dimension $n$ does not significantly affect $K$ . We now show that such a variation of $K$ with $\abs{\kfmax}$ for a fixed sampling width is explained by the theoretical results in Section \ref{sec:mDimSurfaces}. We first note that the sampling width is $\nu = O(n^{-1/2} m^{-1} \abs{\mathcal{K}_{\text{max,large}}}^{-1})$, which can be obtained from the results of Section \ref{subsec:quad_sampl_compl} by evaluating the bounds on the sampling density at $\kfmax = \mathcal{K}_{\text{max,large}} $. From Lemma \ref{lemma:md_k_bound_eps}, one can then easily verify that 
\begin{align*}
K_{\text{bound}}^{(1)} &= O\left(\left(m + \frac{\kfmax^2}{\mathcal{K}_{\text{max,large}}^2}\right)\log n\right), \quad K_{\text{bound}}^{(2)} = O(\log n) \quad \text{and} \\
K_{\text{bound}}^{(3)} &= O\left((m\kfmax^2 \mathcal{K}_{\text{max,large}}^{-2} + m^{1/2}\abs{\kfmax} \abs{\mathcal{K}_{\text{max,large}}}^{-1})\log n\right) \approx O\left(m \abs{\kfmax}/\abs{\mathcal{K}_{\text{max,large}}} \log n\right).
\end{align*}
As $\abs{\kfmax} < \abs{\mathcal{K}_{\text{max,large}}}$, we have $K_{\text{bound}} = O\left(\left(m + \frac{\kfmax^2}{\mathcal{K}_{\text{max,large}}^2}\right)\log n\right)$. Thus, for a fixed $n$, the bound on $K$ increases quadratically with $\abs{\kfmax}$, which is consistent with the curves presented in Figures \ref{fig:md_exp5_quad}-\ref{fig:md_exp5_sin}. Furthermore, $K_{\text{bound}}$ depends only logarithmically on $n$, suggesting that a change in $n$ would affect the sampling density only \textit{mildly}: this also matches the experimental results.

%
\section{Discussion} \label{sec:discuss_results} 
\noindent In this section, we first discuss our results in view of the recent works from the literature. Then we show how our results could be used in practical applications.

We first position our study relatively to the works presented in \cite{Singer2011} and \cite{Kaslovsky2011}, which are, to the best of our knowledge, the closest to our paper. In \cite{Singer2011} the authors consider a global sampling from a compact manifold and relate the size of the neighborhood $\varepsilon$ to the number of samples $K$ through the condition $\varepsilon = O(K^{-\frac{1}{m+2}})$. From this aspect, our approach is significantly different. Our bound on $\varepsilon$ is derived in the asymptotic limit where $K \rightarrow \infty$, so that it depends completely on the local manifold geometry. Furthermore, the analysis in \cite{Singer2011} gives soft bounds that do not reflect the effect of the curvature, nor of the ambient space and manifold dimensions on the sampling conditions. Meanwhile, we derive worst-case bounds on both $\varepsilon$ and $K$ by explicitly taking into account the effect of curvature and dimensions.

The work in \cite{Kaslovsky2011} is parallel to ours and addresses a similar problem. The analysis is however clearly different in two main aspects. Firstly, the analysis in \cite{Kaslovsky2011} assumes that the manifold is embedded with exactly quadratic forms and that the data consists of samples from the quadratic manifold corrupted with Gaussian noise. On the contrary, the type of the manifolds that we consider is more generic as we assume an embedding of the manifold with arbitrary smooth functions. In particular, we explicitly examine the effect of the deviation of the manifold from its second-order approximation on the accuracy of the tangent space estimation. Secondly, an important difference between both studies is that the data is already sampled in \cite{Kaslovsky2011}, where the problem consists of choosing the size of the subset of samples used in the tangent space estimation, while we assume that we have a rather direct control on the parameters of the local random sampling (sampling width and number of samples). Therefore, in \cite{Kaslovsky2011}, the number of samples $N$ (which is $K$ in our notation) and the sampling radius $r$ (which is comparable to the sampling width $\nu$ in our notation) are directly dependent on each other. As the sampling is formulated as a subset selection problem, increasing the number of samples necessarily leads to choosing samples from a larger radius. The analysis is based on the assumption $r=c \, N^{1/d}$, where $c$ is a constant and $d$ is the dimension of the manifold ($m$ with our notation); therefore, $r$ and $N$ can be represented in terms of a single parameter. Meanwhile, in our analysis, we consider a setting where we treat the sampling width $\nu$ and number of samples $K$ as two different parameters. 

Even if the frameworks in \cite{Kaslovsky2011} and in this paper are quite different, we can try to compare results. It is assumed in \cite{Kaslovsky2011} that the subset of samples selected for tangent space estimation corresponds to a sampling radius smaller than a threshold $r_{max}$, where $r_{max}$ is the largest radius within which the manifold can be accurately represented with quadratic forms. We give a characterization of such a bound on the sampling width in Lemma \ref{lemma:md_angle_asymp_smooth} for arbitrary smooth manifolds, which is very relevant to the parameter $r_{max}$ in their work. In  \cite{Kaslovsky2011}, the parameter $r_{max}$ is used as a predetermined constant and the study does not go into the analysis of $r_{max}$ for non-quadratic manifolds. A direct comparison of the main results in both papers is difficult. However, we can compare the noiseless version of the Interpretable Main Result 1 in \cite{Kaslovsky2011} and our results on quadratic manifolds in the following way. The denominator of the angle bound in Interpretable Main Result 1 quantifies the separation between the tangential and normal components of the computed eigenspace. Furthermore, the sampling radius must be small enough to guarantee that the eigenvalues corresponding to the tangential components must be larger than those corresponding to the normal components. Then, an admissible sampling radius must be below the value of $r$ that equates the denominator of the expression in Interpretable Main Result 1 to zero. Taking the noise variance as zero and observing the relation $K=O(n^{1/2} m | \kfmax |)$, where $K$ is the curvature parameter in \cite{Kaslovsky2011}, their result translates into the fact that the admissible sampling radius must be smaller than $O( n^{-1/2} m^{-1/2} |\kfmax|^{-1} )$ with our notation, where $m$, $n$ and $| \kfmax |$ are the parameters corresponding respectively to the intrinsic manifold dimension, the ambient space dimension and the curvature. This is in agreement with our result for quadratic embeddings (see Table \ref{tab:comp_main_results}), where we have calculated the admissible sampling width as $O(n^{-1/2} m^{-1} |\kfmax|^{-1})$. 

Now that our work has been properly positioned with respect to the related work, we discuss the usage of our results in practical applications. We can interpret our results in two important application areas, namely (i) the discretization of a manifold with a known parametric model - \textit{manifold sampling} and (ii) the recovery of the tangent space of a manifold from a given set of data samples - \textit{manifold learning}. 

First, in order to use our results in a real application, the intrinsic dimension $m$ of the manifold, the curvature parameter $\kfmax$, and the higher-order deviation term $C_s$ have to be known or estimated. In a manifold sampling application, $m$ is already known and it is possible to estimate $\kfmax$ in the following ways. If the manifold conforms to a known analytic model, it is easy to compute the values of the principal curvatures and the higher-order terms from the Taylor expansion of the model. If an analytic model is not known for the manifold, the curvature of a manifold of known parameterization can be estimated using results from Riemannian geometry such as \cite{Kokiopoulou2011} (Section V) and \cite{Jacques2008} (Proposition 2). The results in Section V of \cite{Kokiopoulou2011} are especially compatible with our definition of curvature, where we define $\kfmax$ as the largest of the maximum principal curvatures of the hypersurfaces $\mathcal{S}_l$, $l=1, \dots, n-m$, each of which have a single normal direction. Although the work in \cite{Kokiopoulou2011} addresses an image registration problem, the analysis in Section V of \cite{Kokiopoulou2011} is generic and it describes a procedure to compute the maximum principal curvature of a manifold corresponding to a single normal direction, which is equal to the norm of the second fundamental form corresponding to the normal direction. Applying this procedure for all $n-m$ normal directions and taking the largest one of the maximum principal curvatures, one can compute the exact value of $\kfmax$. Then, the deviation term $C_s$ is the maximum of the constants $C_{s,l}$. Once the maximum principal curvature of $ \mathcal{S}_l $ is computed as above, one can find a suitable bound for $C_{s,l}$ by looking at the deviation of $ \mathcal{S}_l$ from its second order approximation.

Second, in a manifold learning application where only data samples are available, $m$, $\kfmax$ and $C_s$ are unknown and need to be estimated. The estimation of the intrinsic dimension of a data set has been studied in several works such as \cite{Hein2005}, \cite{Levina2005} and \cite{Chen2011}. It is also possible to obtain an estimate of the curvature from data samples using results such as in \cite{Little2009}. In \cite{Little2009}, a method is proposed to estimate the intrinsic dimension of the manifold by examining the variation of the singular values of the data covariance matrix with respect to the radius of the neighborhood of samples used. It is observed that the singular values corresponding to the curvatures can be distinguished from the singular values corresponding to the tangential components by using the fact that the tangential and curvature singular values conform respectively to linear and quadratic fits as a function of the radius. In such a setting, the deviation of the curvature singular values from their quadratic fits for large values of the radius can possibly be related to the deviation term $C_s$.

Finally, in our results, we characterize the admissible sampling width for accurate tangent space estimation in terms of the tangent space distances, i.e., the distances between the projections of points on the tangent space and $P$. In a manifold sampling application, our analysis can be easily adapted to the parametric data model at hand since it assumes that the true tangent space of the manifold is aligned with the subspace generated by the first $m$ canonical basis vectors. This can be achieved by applying a Gram-Schmidt orthonormalization to the tangent vectors of the data manifold and then performing a change of coordinates in $\mathbb{R}^n$ such that the subspace spanned by the original tangent vectors is mapped to the subspace generated by the first $m$ canonical basis vectors. Meanwhile, in a manifold learning application where only data samples are available, one needs to adapt the bounds on the tangent space distance to bounds on the distance between actual data samples in the ambient space. This can be done in different ways. Based on our results, one can easily obtain some worst-case bounds on the ambient space distance by making use of the fact that the tangent space distance is upper bounded by the ambient space distance. This approach is expected to be effective if the ambient space dimension $n$ is comparable to the intrinsic dimension $m$, or if the manifold has small curvature. Alternatively, if $n\gg m$ and the manifold has significant nonlinearity, the current results involving the tangent space distance can be translated into approximate conditions on the ambient space distance with the help of the estimation $\| . \|_{ambient \, space} \approx O(\| . \|_{tangent \, space} \sqrt{n/m} )$. Note that, using this estimation, the decay of the sampling width $\nu$ in the tangent space at the rate $ O(n^{-1/2} m^{-1} |  \mathcal{K}_{max}|^{-1})$ implies that the same width measured in the ambient space must change at the rate $O(\nu \sqrt{n/m}  ) = O(m^{-3/2}   |  \mathcal{K}_{max}|^{-1})$. Therefore, the sampling width in the ambient space does not decrease with the ambient space dimension. It is of $O(1)$ with respect to $n$; meanwhile, it decreases with $m$ and $|  \mathcal{K}_{max}|$. This means that, when applying PCA, the size of the neighborhood around a reference point in the ambient space must get smaller as the intrinsic dimension or the curvature of the manifold increases.

In this work, we have focused on a noiseless data model that is perfectly representable with smooth functions. However, in real applications, one may need to work with noisy data samples that exhibit a deviation from the manifold. One can possibly extend the study presented here to include the effect of noise in the analysis. This can be achieved by first identifying the sampling region for an accurate estimation of the tangent space and then determining a sufficient sampling density in that region. The admissible sampling region highly depends on the type of noise. One would expect to have no bias in the estimation for a random noise model with spherical symmetry, while a structured noise model may bias the estimation and necessitate stricter constraints on the sampling width. Then, the sampling density is expected to be affected by the variance of the noise. These effects can be characterized by studying the additional perturbation on the correlation matrices due to the noise.

%
\section{Concluding Remarks} \label{sec:manifold_conclusion} 
\noindent We have presented a theoretical analysis of the tangent space estimation at a point on a submanifold from a set of manifold samples that are selected locally at random. We have considered a setting where the manifold is embedded smoothly in $\mathbb{R}^n$ and the tangent space is estimated with local PCA. We have derived relations between the accuracy of the tangent space estimation and the sampling conditions. In particular, we have examined the effect of the local curvature of the manifold in tangent space estimation and shown that the size of the sampling neighborhood shall be inversely proportional to the manifold curvature. We have also seen that sampling conditions are affected by the correlation between the components of the second-order approximation of the embedding. The sampling width can be chosen larger when the components of the manifold in different dimensions are less correlated. The presented study can be used for obtaining performance guarantees in the discretization of parametrizable data and in manifold learning applications. Finally, our analysis assumes that the data samples are noiseless, i.e., the data lies exactly on the manifold. A future research direction resides therefore in the extension of the current results to a scenario where data samples are corrupted with noise. 


%
%
\section{Acknowledgments} \label{sec:ack} 
The authors would like to thank Prof.~Daniel Kressner and Dr.~Bart Vandereycken for the helpful discussions and comments on the manuscripts.

\bibliographystyle{unsrt}
\bibliography{manifold_references}

\begin{thebibliography}{10}

\bibitem{Tyagi2011}
H.~Tyagi.
\newblock Local {S}ampling {A}nalysis for {Q}uadratic {E}mbeddings of
  {R}iemannian {M}anifolds.
\newblock Master's thesis, Ecole Polytechnique F\'{e}d\'{e}rale de Lausanne,
  July 2011.
\newblock Available: http://infoscience.epfl.ch/record/179897.

\bibitem{Vural2011}
E.~Vural and P.~Frossard.
\newblock Discretization of {P}arametrizable {S}ignal {M}anifolds.
\newblock {\em {IEEE} {T}ransactions on {I}mage {P}rocessing},
  20(12):3621--3633, 2011.

\bibitem{Tenenbaum2000}
J.B. Tenenbaum, V.D. Silva, and J.C. Langford.
\newblock A global geometric framework for nonlinear dimensionality reduction.
\newblock {\em Science}, 290:2319--2323, 2000.

\bibitem{Roweis2000}
S.T. Roweis and L.K. Saul.
\newblock Nonlinear dimensionality reduction by locally linear embedding.
\newblock {\em Science}, 290:2323--2326, 2000.

\bibitem{Donoho2003}
D.~L. Donoho and C.~E. Grimes.
\newblock Hessian eigenmaps: Locally linear embedding techniques for
  highdimensional data.
\newblock {\em Proc. Natl. Acad. Sci. USA}, 100:5591--5596, 2003.

\bibitem{Lin2006}
T.~Lin, H.~Zha, and S.U. Lee.
\newblock Riemannian manifold learning for nonlinear dimensionality reduction.
\newblock In {\em Proc. of Eur. Conf. Computer Vision}, 2006.

\bibitem{Zhang2011}
Z.~Zhang, J.~Wang, and H.~Zha.
\newblock {Adaptive Manifold Learning}.
\newblock {\em IEEE Transactions on Pattern Analysis and Machine Intelligence},
  34(2):253--265, February 2012.

\bibitem{Zha2009}
H.~Zha and Z.~Zhang.
\newblock Spectral properties of the alignment matrices in manifold learning.
\newblock {\em SIAM Review}, 51:545--566, 2009.

\bibitem{Zhang2002}
Z.~Zhang and H.~Zha.
\newblock Principal manifolds and nonlinear dimension reduction via local
  tangent space alignment.
\newblock {\em SIAM Journal of Scientific Computing}, 26:313--338, 2005.

\bibitem{Yang2010}
Y.~Yang, F.~Nie, S.~Xiang, Y.~Zhuang, and W.~Wang.
\newblock Local and global regressive mapping for manifold learning with
  out-of-sample extrapolation.
\newblock In {\em Proc. of the 24th AIII Conf. on Artificial Intelligence},
  2010.

\bibitem{Zhan2008}
Y.~Zhan, J.~Yin, G.~Zhang, and E.~Zhu.
\newblock Incremental manifold learning algorithm using {PCA} on overlapping
  local neighborhoods for dimensionality reduction.
\newblock In {\em Advances in Computation and Intelligence}, volume 5370 of
  {\em Lecture Notes in Computer Science}, pages 406--415. Springer
  Berlin/Heidelberg, 2008.

\bibitem{Davis1970}
C.~Davis and W.~M. Kahan.
\newblock The rotation of eigenvectors by a perturbation, {III}.
\newblock {\em SIAM J. Numer. Anal.}, 7, March 1970.

\bibitem{Wedin1972}
P.A. Wedin.
\newblock Perturbation bounds in connection with singular value decomposition.
\newblock {\em BIT Numerical Mathematics}, 12:99--111, 1972.
\newblock 10.1007/BF01932678.

\bibitem{Vu2011}
V.~Vu.
\newblock Singular vectors under random perturbation.
\newblock {\em Random Struct. Algorithms}, 39(4):526--538, December 2011.

\bibitem{Faber1995}
N.~M. Faber, M.~J. Meinders, P.~Geladi, M.~Sj\"{o}str\"{o}m, L.~M.~C. Buydens,
  and G.~Kateman.
\newblock {Random error bias in principal component analysis. Part I.
  derivation of theoretical predictions}.
\newblock {\em Analytica Chimica Acta}, 304(3):257--271, 1995.

\bibitem{Anderson1963}
T.~W. Anderson.
\newblock Asymptotic theory for principal component analysis.
\newblock {\em The Annals of Mathematical Statistics}, 34(1):122--148, 1963.

\bibitem{Lawley1956}
D.~N. Lawley.
\newblock Tests of significance for the latent roots of covariance and
  correlation matrices.
\newblock {\em Biometrika}, 43(1-2):128--136, June 1956.

\bibitem{Girshick1939}
M.~A. Girshick.
\newblock On the sampling theory of roots of determinantal equations.
\newblock {\em The Annals of Mathematical Statistics}, 10(3):203--224, 1939.

\bibitem{Singer2011}
A.~Singer and H.~Wu.
\newblock {Vector Diffusion Maps and the Connection Laplacian}.
\newblock {\em Comm. on Pure and App. Math.}, 2012.

\bibitem{Coifman2005}
R.R. Coifman, S.~Lafon, A.B. Lee, M.~Maggioni, F.~Warner, and S.~Zucker.
\newblock Geometric diffusions as a tool for harmonic analysis and structure
  definition of data: Diffusion maps.
\newblock In {\em Proceedings of the National Academy of Sciences}, pages
  7426--7431, 2005.

\bibitem{Kaslovsky2011}
D.~Kaslovsky and F.G. Meyer.
\newblock Optimal tangent plane recovery from noisy manifold samples.
\newblock Submitted to the Annals of Statistics, available at
  http://arxiv.org/abs/1111.4601v2.

\bibitem{Gittens2011}
A.~Gittens and J.~Tropp.
\newblock Tail bounds for all eigenvalues of a sum of random matrices.
\newblock {\em Preprint}, 2011.

\bibitem{Tropp2011}
J.~Tropp.
\newblock User-friendly tail bounds for sums of random matrices.
\newblock {\em Preprint}, 2011.

\bibitem{Golub1996}
G.H. Golub and van Loan~C.F.
\newblock Matrix computations.
\newblock The Johns Hopkins University Press, Baltimore, 1996.

\bibitem{Niyogi2006}
P.~Niyogi, S.~Smale, and S.~Weinberger.
\newblock {Finding the homology of submanifolds with confidence from random
  samples}.
\newblock {\em Discrete and Computational Geometry}, 2006.

\bibitem{Gunawan2005}
H.~Gunawan, O.~Neswan, and W.~Setya-Budhi.
\newblock A formula for angles between subspaces of inner product spaces.
\newblock {\em Contributions to Algebra and Geometry}, 46:311--320, 2005.

\bibitem{Kokiopoulou2011}
E.~Kokiopoulou, D.~Kressner, and P.~Frossard.
\newblock Optimal image alignment with random projections of manifolds:
  algorithm and geometric analysis.
\newblock {\em {IEEE} {T}ransactions on {I}mage {P}rocessing},
  20(6):1543--1557, 2011.

\bibitem{Jacques2008}
L.~Jacques and C.~De~Vleeschouwer.
\newblock A geometrical study of matching pursuit parametrization.
\newblock {\em IEEE Transactions on Signal Processing}, 56(7):2835--2848, July
  2008.

\bibitem{Hein2005}
M.~Hein.
\newblock Intrinsic dimensionality estimation of submanifolds in {E}uclidean
  space.
\newblock In {\em Proceedings of the $22^{nd}$ International Conference on
  Machine Learning}, pages 289--296, 2005.

\bibitem{Levina2005}
E.~Levina and P.J. Bickel.
\newblock Maximum likelihood estimation of intrinsic dimension.
\newblock In {\em Advances in Neural Information Processing Systems}, 2005.

\bibitem{Chen2011}
G.~Chen, A.V. Little, M.~Maggioni, and L.~Rosasco.
\newblock Some recent advances in multiscale geometric analysis of point
  clouds.
\newblock {\em Wavelets and Multiscale Analysis: Theory and Applications},
  March 2011.

\bibitem{Little2009}
A.V. Little, J.~Lee, Y.M. Jung, and M.~Maggioni.
\newblock Estimation of intrinsic dimensionality of samples from noisy
  low-dimensional manifolds in high dimensions with multiscale {SVD}.
\newblock In {\em Proc. of S.S.P.}, 2009.

\bibitem{Weyl1912}
H.~Weyl.
\newblock Das asymptotische verteilungsgesetz der eigenwerte linearer
  partieller differentialgleichungen (mit einer anwendung auf die theorie der
  hohlraumstrahlung).
\newblock {\em Mathematische Annalen}, 71:441--479, 1912.

\end{thebibliography}


%
%
\appendix
\section{$m$-dimensional smooth manifolds in $\mathbb{R}^n$} \label{sec:proofs_md_manifolds}
%
\subsection{Proof of Lemma ~\ref{lemma:md_width_cond}} \label{appendix:proof_md_width_cond}
\begin{proof}
Observe that each entry of $M^{(K)}$ is the sum of $K$ i.i.d. random variables. Therefore, by the Strong Law of Large Numbers as $K \rightarrow \infty$, $[M^{(K)}]_{i,j}$ converges a.s.~to $[M]_{i,j}$ for all $1 \ \leq \ i,j \ \leq \ n$, where each entry of $M$ is the expected value of the random variable involved in the summation of the corresponding entry of $M^{(K)}$. Let
\begin{equation*}
M= \begin{bmatrix}
A & B \\
B^T & D \\
\end{bmatrix}.
\end{equation*}
Consider the entries of $A$. We have for $j,k=1,\dots,m$,
\begin{equation*}
[A]_{j,k} \ = \ \mathbb{E}[x_j x_k] = \left\{
\begin{array}{rl}
0 & \text{if } j \neq k \\
\frac{\displaystyle \sampwidth^2}{3} & \text{if } j = k
\end{array} \right .
\end{equation*}
Consider the entries of $B$. We have for $j=1,\dots,m$ and $l=1,\dots,n-m$,
\begin{align*}
[B]_{j,l} \ = \ \mathbb{E}[x_j f_{l}(\bar{x})] &= \mathbb{E}[x_j \frac{1}{2}\sum_{k=1}^{m}<\bar{x},\bar{v}_{l,k}>^2 \kflk] \\
&= \frac{1}{2} \mathbb{E}[x_j \sum_{k=1}^{m} (x_1 v_{l,k,1} + \dots + x_1 v_{l,k,m})^2 \kflk] = 0.
\end{align*}
The above result follows as each term in the expansion of $x_j f_l(\bar{x})$ has at least one odd power of $x_j$, and the expected value of each term is thus 0. Now, consider the diagonal entries of $D$. We have 
\begin{align*}
[D]_{l,l} \ = \ \mathbb{E}[f^2_l(\bar{x})] &= \frac{1}{4}\mathbb{E}\left[\left(\sum_{j=1}^m <\bar{x}, \bar{v}_{l,j}>^2\kflj \right)^2\right] 
\leq \frac{1}{4} \abs{\kfmax}^2(\mathbb{E}[\norm{\bar{x}}_2^4]).
\end{align*}
Furthermore,
\begin{align*}
\mathbb{E}[\norm{\bar{x}}_2^4] = \mathbb{E}[\sum_{j=1}^m x_j^4 + 2 \sum_{k<j} x_k^2 x_j^2] = \frac{m\sampwidth^4}{5} + 2\frac{m(m-1)}{2}\left(\frac{\sampwidth^2}{3}\right)^2
= \frac{m(5m+4)\sampwidth^4}{45}.
\end{align*}
Hence
\begin{equation*}
0 \ \leq \ [D]_{l,l} \ \leq \ \frac{m(5m+4)\sampwidth^4}{180}\abs{\kfmax}^2 \quad (l = 1,\dots,n-m).
\end{equation*}
We have the following bounds for $l,k=1,\dots,n-m,\ l \neq k$ on the off-diagonal entries of $D$
\begin{align*}
[D]_{l,k} &= \mathbb{E}[f_{l}(\bar{x}) f_{k}(\bar{x})] \\
&= \frac{1}{4}\mathbb{E}[(\sum_{j=1}^m <\bar{x}, \bar{v}_{l,j}>^2\kflonej)(\sum_{j=1}^m <\bar{x}, \bar{v}_{k,j}>^2\kfltwoj)] \\
&\leq \frac{1}{4}\abs{\kfmax}^2 \mathbb{E}[\norm{\bar{x}}_2^4] = \frac{m(5m+4)\sampwidth^4}{180}\abs{\kfmax}^2.
\end{align*}
Similarly, it holds that
\begin{equation*}
[D]_{l,k} \ \geq \ -\frac{m(5m+4)\sampwidth^4}{180}\abs{\kfmax}^2.
\end{equation*}
Hence, $M$ has the form
\begin{equation*} 
M = \begin{bmatrix}
\frac{\displaystyle \sampwidth^2}{\displaystyle 3} I_{m \times m} & 0_{m \times (n-m)} \\
0_{(n-m) \times m} & D_{(n-m) \times (n-m)} \\
\end{bmatrix},
\end{equation*}
where
\begin{align*}
0 \ &\leq \ [D]_{l,l} \ \leq \ \frac{m(5m+4)\sampwidth^4}{180}\abs{\kfmax}^2, \\
-\frac{m(5m+4)\sampwidth^4}{180}\abs{\kfmax}^2 \ &\leq [D]_{l,k} \ \leq \frac{m(5m+4)\sampwidth^4}{180}\abs{\kfmax}^2 \quad (l \neq k).
\end{align*}
Therefore, for $l,k \ = \ 1,\dots,n-m$,
\begin{equation*}
\abs{[D]_{l,k}} < \frac{m(5m+4)\sampwidth^4}{180}\abs{\kfmax}^2 = [D]_{bound}.
\end{equation*}
Observe that the eigenspace of $M$ corresponding to the eigenvalue $\frac{\displaystyle \sampwidth^2}{\displaystyle 3}$ is equal to the span of $\set{\bar{e}_1,\dots,\bar{e}_m}$, which is the same as $\tanps$. Hence, as $K \rightarrow \infty$, we obtain the implication 
\begin{equation*}
\frac{\sampwidth^2}{3} \ > \ \rho(D) \ \Rightarrow \ \abs{\angle \widehat{T}_PS, \tanps} \rightarrow 0,
\end{equation*} 
where $\rho(D)$ denotes the spectral radius of $D$, which is positive definite. In the case where $D$ is diagonal, we have

\begin{equation*}
\rho(D) < [D]_{bound} = \frac{m(5m+4)\sampwidth^4}{180}\abs{\kfmax}^2.
\end{equation*}

\noindent Therefore, for this case, any value of $\sampwidth$ satisfying
\begin{align*}
\frac{\displaystyle \sampwidth^2}{\displaystyle 3} > \frac{m(5m+4)\sampwidth^4}{180}\abs{\kfmax}^2 \text{ or equivalently } \sampwidth < \sqrt{\frac{60}{m(5m+4)\abs{\kfmax}^2}}
\end{align*}
ensures that $\abs{\angle \widehat{T}_PS, \tanps} \rightarrow 0$ as $K \rightarrow \infty$. In the scenario where $D$ is dense, we have the stricter condition
\begin{align*}
\rho(D) < \norm{D}_F \leq (n-m)\frac{\sampwidth^4(5m+4)m}{180}\abs{\kfmax}^2. 
\end{align*}
Thus, for this case, any value of $\sampwidth$ satisfying
\begin{align*}
\frac{\displaystyle \sampwidth^2}{\displaystyle 3} > (n-m)\frac{m(5m+4)\sampwidth^4}{180}\abs{\kfmax}^2 \text{ or equivalently } \sampwidth < \sqrt{\frac{60}{m(n-m)(5m+4)\abs{\kfmax}^2}}
\end{align*}
ensures that $\abs{\angle \widehat{T}_PS, \tanps} \rightarrow 0$ as $K \rightarrow \infty$.
\end{proof}
%
\subsection{Proof of Lemma ~\ref{lemma:md_k_bound_eps}} \label{appendix:proof_md_k_bound_eps}
We first recall two recent results on the tail bounds for the eigenvalues of sums of independent random matrices. The first result concerns upper and lower tail bounds on all eigenvalues of a sum of independent positive semidefinite matrices as stated in Theorem 4.1 in \cite{Gittens2011}.
\begin{theorem}[Eigenvalue Chernoff Bounds] \label{thm:all_eig_tail_bounds}
Consider a finite sequence $\set{X_i}$ of independent random positive semidefinite matrices where $X_i \in \mathbb{R}^{n \times n}$ with $\norm{X_i} \leq R$ a.s. Given an integer $k \leq n$ define
\begin{equation*}
\mu_k = \lambda_k(\sum_j \mathbb{E}[X_j])
\end{equation*}
Then
\begin{align*} 
\mathbb{P}(\lambda_k(\sum_j X_j) &\geq t\mu_k) \leq (n-k+1)\left[\frac{e}{t}\right]^{t\mu_k / R} \text{where } t > e \quad \text{and} \\
\mathbb{P}(\lambda_k(\sum_j X_j) &\leq s\mu_k) \leq k e^{\frac{-(1-s)^2 \mu_k}{2R}}, \quad s \in (0,1).
\end{align*}
\end{theorem}
The second result concerns an upper tail bound on the operator norm of a sum of zero-mean independent random matrices which can moreover be rectangular. This result is stated in the form of Theorem 1.3 in \cite{Tropp2011}.
\begin{theorem}[Matrix Bernstein: Rectangular Case] \label{thm:matrix_bernstein_rect}
Consider a finite sequence $\set{Z_j}$ of independent random matrices, $Z_j \in \mathbb{R}^{d_1 \times d_2}$. Assume that each random matrix satisfies
\begin{equation*}
\mathbb{E}[Z_j] = 0 \quad \text{and} \quad \norm{Z_j} \leq R \quad \text{a.s}.
\end{equation*}
Define
\begin{equation*}
\sigma^2 := \max \set{\norm{\sum_k \mathbb{E}[Z_k Z_k^{*}]}, \norm{\sum_k \mathbb{E}[Z_k^{*} Z_k]}}.
\end{equation*}
Then for all $t \geq 0$,
\begin{equation*}
\mathbb{P}(\norm{\sum_k Z_k} \geq t) \leq (d_1 + d_2) \exp\left(\frac{-t^2 /2}{\sigma^2 + Rt/3}\right).
\end{equation*}
\end{theorem}
We now proceed to prove the Lemma \ref{appendix:proof_md_k_bound_eps}.\\
\begin{proof}
We have $M^{(K)} = \frac{1}{K}\sum_{i=1}^{K} \bar{p}_i \bar{p}_i^T$, where $\bar{p}_i = [\bar{x}_i^T \ f_1(\bar{x}_i) \dots f_{n-m}(\bar{x}_i)]^T \in \mathbb{R}^{n}$. Now,
\begin{align*}
\norm{\frac{1}{K}\bar{p}_i \bar{p}_i^T} \leq \frac{1}{K}\norm{\bar{p}_i}_2^2 &\leq \frac{1}{K}(m \sampwidth^2 + \frac{1}{4}(n-m)m^2 \sampwidth^4 \abs{\kfmax}^2) \\
&= \frac{1}{K}\sampwidth^2 R_M \quad \text{a.s.,}
\end{align*}
where $R_M = m + \frac{1}{4}(n-m)m^2 \sampwidth^2 \abs{\kfmax}^2$. Here we used the fact that
\begin{equation*}
\abs{f_l(\bar{x})} \leq \frac{1}{2} m \sampwidth^2 \abs{\kfmax} \ \text{ for } \ \bar{x} \in [-\sampwidth,\sampwidth]^m \text{and } \ l=1,\dots,n-m .
\end{equation*}
Furthermore, since $\sampwidth < \sampwidth_{\text{bound,quad}}$,
\begin{equation*}
\mu_j = \lambda_j \left(\mathbb{E}\left[\frac{1}{K}\sum_{i=1}^{K} \bar{p}_i \bar{p}_i^T \right]\right) = \frac{\sampwidth^2}{3},  \quad j=1,\dots m.
\end{equation*}
Hence, by applying Theorem \ref{thm:all_eig_tail_bounds}, we have the following for $s_1 \in (0,1)$:
\begin{align}
\mathbb{P}(\lambda_m(M^{(K)}) \leq s_1 \sampwidth^2/3) &\leq (n-m+1)\exp\left(\frac{-(1-s_1)^2 \frac{\sampwidth^2}{3}}{2 \sampwidth^2 R_M/K}\right) \nonumber \\
&= (n-m+1)\exp\left(\frac{-(1-s_1)^2 K}{6 R_M}\right). \label{eq:Mk_eig_val_bound}
\end{align}
Then, we have,  $D^{(K)} = \frac{1}{K}\sum_{i=1}^{K} \bar{q}_i \bar{q}_i^T$ where $\bar{q}_i = [f_1(\bar{x}_i) \dots f_{n-m}(\bar{x}_i)]^T \in \mathbb{R}^{n-m}$. Furthermore,
\begin{align*} 
\norm{\frac{1}{K}\bar{q}_i \bar{q}_i^T} \leq \frac{1}{K} \norm{\bar{q}_i}_2^2 \leq \frac{R_D \sampwidth^4}{K}, 
\end{align*}
where $R_D = \frac{1}{4} (n-m) m^2 \abs{\kfmax}^2$. Applying Theorem \ref{thm:all_eig_tail_bounds} for $\rho(D^{(K)}) = \lambda_1(D^{(K)})$, we can write
\begin{equation*}
\mathbb{P}(\rho(D^{(K)}) \geq s_2 \rho(D)) \leq (n-m)\left[\frac{e}{s_2}\right]^{\frac{s_2 \rho(D) K}{R_D \sampwidth^4}}, \quad s_2 > e.
\end{equation*}
We have seen in Section \ref{appendix:proof_md_width_cond} that $\rho(D) < RL \sampwidth^4$. Using this, we obtain the following tail bound:
\begin{equation} \label{eq:spec_rad_D_bound}
\mathbb{P}(\rho(D^{(K)}) \geq s_2 RL \sampwidth^4) \leq (n-m)\left[\frac{e}{s_2}\right]^{\frac{s_2 RL K}{R_D}}, \quad s_2 > e.
\end{equation}
We proceed now to derive an upper bound on $\norm{B^{(K)}}$ by applying Theorem \ref{thm:matrix_bernstein_rect}. First, observe that
\begin{equation*}
B^{(K)} = \frac{1}{K} \sum_{i=1}^{K} \bar{x}_i \bar{q}_i^T. 
\end{equation*}
By using the bounds
\begin{equation*}
\norm{\bar{x}_i}_2 \leq \sampwidth\sqrt{m} \quad \text{and} \quad \norm{\bar{q}_i}_2 \leq \frac{1}{2}m \sampwidth^2\sqrt{n-m} \abs{\kfmax},
\end{equation*}
we obtain
\begin{equation*}
\norm{\frac{1}{K} \bar{x}_i \bar{q}_i^T} \leq \frac{1}{K} \norm{\bar{x}_i}_2 \norm{\bar{q}_i}_2 \leq \frac{R_B \sampwidth^3}{K},
\end{equation*}
where $R_B = \frac{1}{2} m^{3/2} (n-m)^{1/2} \abs{\kfmax}$.
The parameter $\sigma^2$ defined in Theorem A.2 has the following form
\begin{eqnarray*}
\sigma^2 &=&  \max \left\{
\frac{1}{K^2} \| \sum_{i=1}^{K} 
 \mathbb{E}[ \bar{x}_i \bar{q}_i^T \bar{q}_i \bar{x}_i^T  ] \|,
\frac{1}{K^2} \| \sum_{i=1}^{K} 
\mathbb{E}[ \bar{q}_i \bar{x}_i^T \bar{x}_i \bar{q}_i^T  ] \| 
\right\}\\
&\leq&  \max\{
\frac{1}{K^2} \sum_{i=1}^{K} 
 \left\| \mathbb{E}\left[ \, \| \bar{q}_i \|_2^2  \bar{x}_i \bar{x}_i^T \,  \right] \right\|,
\frac{1}{K^2} \sum_{i=1}^{K} 
 \left\| \mathbb{E}\left[ \,  \| \bar{x}_i \|_2^2 \|  \bar{q}_i \bar{q}_i^T \, \right]  \right\|
\}.
\end{eqnarray*}

Now the terms $\left\| \mathbb{E}\left[ \, \| \bar{q}_i \|_2^2  \bar{x}_i \bar{x}_i^T \,  \right] \right\|$ and $\left\| \mathbb{E}\left[ \,  \| \bar{x}_i \|_2^2 \|  \bar{q}_i \bar{q}_i^T \, \right]  \right\|$ can be bounded from above as follows.
\begin{align*}
\left\|\mathbb{E}\left[ \, \| \bar{q}_i \|_2^2  \bar{x}_i \bar{x}_i^T \,\right] \right\| = 
\text{sup}_{\|\bar{y}\|_2 = 1} \bar{y}^T \mathbb{E}\left[ \, \| \bar{q}_i \|_2^2  \bar{x}_i \bar{x}_i^T \,\right] \bar{y} 
\leq (\| \bar{q}_i \|_2^2)_{\max} \left\|\mathbb{E}\left[ \, \bar{x}_i \bar{x}_i^T \,\right] \right\|, \\
\left\| \mathbb{E}\left[ \,  \| \bar{x}_i \|_2^2 \|  \bar{q}_i \bar{q}_i^T \, \right]  \right\| = 
\text{sup}_{\|\bar{y}\|_2 = 1} \bar{y}^T \mathbb{E}\left[ \, \| \bar{x}_i \|_2^2  \bar{q}_i \bar{q}_i^T \,\right] \bar{y} 
\leq (\| \bar{x}_i \|_2^2)_{\max} \left\|\mathbb{E}\left[ \, \bar{q}_i \bar{q}_i^T \,\right] \right\|.
\end{align*}

Observe that, for $i=1,\dots,K$, we have 
\begin{equation*}
\norm{\mathbb{E}[\bar{x}_i \bar{x}_i^T]} = \norm{\frac{\sampwidth^2}{3} I_m} = \frac{\sampwidth^2}{3} \quad \text{and} \quad  \norm{\mathbb{E}[\bar{q}_i\bar{q}_i^T]} = \rho(D) < RL \sampwidth^4. 
\end{equation*}
Furthermore, using the aforementioned upper bounds on $\norm{\bar{q}_i}_2$ and $\norm{\bar{x}_i}_2$, we arrive at the following:
\begin{align*}
\sigma^2 &\leq \max\set{\frac{(n-m)m^2 \sampwidth^6 \abs{\kfmax}^2}{12K}, \frac{m \sampwidth^2}{K} \rho(D)} \\
&\leq \max\set{\frac{(n-m)m^2 \sampwidth^6 \abs{\kfmax}^2}{12K}, \frac{m RL \sampwidth^6}{K}} = \frac{\sampwidth^6 R_{\sigma}}{K},
\end{align*}
where
\begin{equation*}
R_{\sigma} := \frac{m^2\abs{\kfmax}^2}{12} \max \set{n-m,\frac{R(5m+4)}{15}}. 
\end{equation*}
Employing the bounds on $\norm{\frac{1}{K} \bar{x}_i \bar{q}_i^T}$ and $\sigma^2$ in Theorem \ref{thm:matrix_bernstein_rect}, we obtain the following tail bound.
\begin{equation} \label{eq:B_op_norm_bound}
\mathbb{P}(\norm{B^{(K)}} > s_3) \leq n \exp\left(\frac{-(s_3^2/2) K}{\sampwidth^6 R_{\sigma} + \frac{R_B \sampwidth^3 s_3}{3}}\right), \quad s_3 > 0  
\end{equation}
Lastly, let $0 < p_1, p_2, p_3 < 1$ denote the upper bounds on the probabilities of the events 
\begin{equation*}
\set{\lambda_m(M^{(K)}) \leq s_1 \sampwidth^2/3}, \set{\rho(D^{(K)}) \geq s_2 RL \sampwidth^4}, \set{\norm{B^{(K)}} > s_3},
\end{equation*}
respectively. This is clearly achieved by choosing
\begin{equation*}
K > \max\set{K_{bound}^{(1)}, K_{bound}^{(2)}, K_{bound}^{(3)}} = K_{bound},
\end{equation*}
where $K_{bound}^{(1)}, K_{bound}^{(2)}, K_{bound}^{(3)}$ are as defined in the statement of Lemma \ref{lemma:md_k_bound_eps}. Applying the union bound, we arrive at the stated result.
\end{proof}
%
\subsection{Proof of Theorem ~\ref{theorem:md_angle_bound_prob}} \label{appendix:proof_md_angle_samp_quad}
\begin{proof}
We start with the following identity for $i=1,\dots,m$ 
\begin{equation} \label{eq:eig_val_vec_ident1}
M^{(K)}\bar{u}_i \ = \ \lambda_i(M^{(K)})\bar{u}_i,
\end{equation}
where
\begin{equation*}
M^{(K)} = \begin{bmatrix}
A^{(K)} & B^{(K)} \\
B^{(K)^T} & D^{(K)} \\
\end{bmatrix}.
\end{equation*}
Here $\lambda_1(M^{(K)}) \geq \lambda_2(M^{(K)}) \geq \dots \geq \lambda_n(M^{(K)})$ denote the eigenvalues of $M^{(K)}$ and $\bar{u}_i = [\bar{u}_{i,1}^T \   \bar{u}_{i,2}^T]^T$ denote its corresponding eigenvectors. Using Eq. \eqref{eq:eig_val_vec_ident1}, we obtain the following inequality.
\begin{align}
B^{(K)^T} \bar{u}_{i,1} + D^{(K)} \bar{u}_{i,2} = \lambda_i(M^{(K)}) \bar{u}_{i,2} \nonumber \\
\Rightarrow (\lambda_m(M^{(K)}) - \rho(D^{(K)})) \norm{\bar{u}_{i,2}}_2 < \norm{B^{(K)}}. \label{eq:eigevec_ineq}
\end{align}
Now, provided that $K$ is chosen such that $K > K_{bound}$, the following events hold with high probability.
\begin{equation} \label{eq:tail_bounds_eigvals}
\set{\lambda_m(M^{(K)}) > s_1\frac{\sampwidth^2}{3}}, \quad \set{\rho(D^{(K)}) < s_2 RL \sampwidth^4}, \quad \set{\norm{B^{(K)}} < s_3},
\end{equation} 
where $s_1 \in (0,1), s_2 > e$ and $s_3 > 0$. From \eqref{eq:tail_bounds_eigvals} and \eqref{eq:eigevec_ineq}, we conclude that the following inequality holds with high probability.
\begin{align*}
(s_1\frac{\sampwidth^2}{3} - s_2 RL \sampwidth^4) \norm{\bar{u}_{i,2}}_2 < s_3 .
\end{align*}
The L.H.S. of the above inequality is positive if $\sampwidth < \sqrt{s_1/(3 s_2 RL)}$. Assuming that this is satisfied, we obtain
\begin{align*}
\norm{\bar{u}_{i,2}} \ &< \ \frac{s_3}{s_1\frac{\sampwidth^2}{3} - s_2 RL \sampwidth^4} \ = \ \sigma_s \\
\Rightarrow \quad \norm{U_2}_F \ &< \ \sqrt{m} \sigma_s.
\end{align*}
Furthermore, we have from Lemma \ref{lemma:md_angle_bound_cond} that
\begin{equation} \label{eq:ang_boun_eqn}
\norm{U_2}_F < \tau < 1 \Rightarrow \abs{\angle \widehat{T}_PS, \tanps} < \cos^{-1}(\sqrt{(1-\tau^2)^m}).
\end{equation}
Lastly, we see that Eq. \eqref{eq:ang_boun_eqn} is ensured if the following holds
\begin{align*}
\sqrt{m} \sigma_s < \tau \Leftrightarrow \frac{\sqrt{m} s_3}{s_1 \frac{\sampwidth^2}{3} - s_2 RL \sampwidth^4} < \tau \Leftrightarrow s_3 < \frac{(s_1 \frac{\sampwidth^2}{3} - s_2 RL \sampwidth^4)\tau}{\sqrt{m}}.
\end{align*}
Therefore, for these choices of $\nu$ and the constants $s_1$, $s_2$, $s_3$, we get the bound on $\abs{\angle \widehat{T}_PS, \tanps}$ stated in the theorem.

\end{proof}
%
\subsection{Proof of Lemma ~\ref{lemma:md_angle_asymp_smooth}} \label{appendix:proof_md_angle_asym_smooth}
\begin{proof}
For $K = \infty$, we have $M = M_{q} + \Delta$, where
\begin{equation*}
M_q = \begin{bmatrix}
\frac{\sampwidth^2}{3}I_{m} & 0 \\
0 & D \\
\end{bmatrix}
\end{equation*}
and
\begin{align}
\norm{\Delta}_{F} &< 2\norm{B_1}_{F} + \norm{D_1}_{F}, \nonumber \\
&< 2\sqrt{m(n-m)}\sampwidth \delta(\sampwidth) + (n-m)(\delta(\sampwidth)^2 + \delta(\sampwidth) m \sampwidth^2\abs{\kfmax}), \nonumber \\
&= 2\smoothconst m^{3/2}\sqrt{m(n-m)} \sampwidth^4 + (n-m)(\smoothconst^2 m^3 \sampwidth^6 + \smoothconst m^{5/2} \sampwidth^5 \abs{\kfmax}), \label{eq:norm_pert_bound} \\
&= 2\norm{B_1}_{F,bound} + \norm{D_1}_{F,bound} = \norm{\Delta}_{F,bound}. \nonumber
\end{align}
Now, if there is no perturbation on $M_q$, the eigenvectors $\set{\bar{e}_1,\dots,\bar{e}_m}$ corresponding to $\frac{\sampwidth^2}{3}$ span $\tanps$. As $M_q$ is actually perturbed by $\Delta$, we analyze the perturbation of the space formed by the span of $\set{\bar{e}_1,\dots,\bar{e}_m}$. We first observe from Weyl's inequality \cite{Weyl1912} the following bounds on the eigenvalues $\set{\lambda_i(M)}_{i=1}^{n}$ of $M$:
\begin{equation*}
\lambda_i(M) \in \left[ \lambda_i(M_q) - \norm{\Delta}_{F,bound}, \ \lambda_i(M_q) + \norm{\Delta}_{F,bound} \right], \quad i=1,\dots,n.
\end{equation*}
Here, $\lambda_i(M_q) = \frac{\sampwidth^2}{3}$, for $ i=1,\dots,m$. Furthermore, $\set{\lambda_i(M_q)}_{i=m+1}^{n}$ are the eigenvalues of $D$. In order to analyze the perturbation on span$\set{\bar{e}_1,\dots,\bar{e}_m}$, we would like to guarantee the `separation' of $\set{\lambda_i(M)}_{i=1}^{m}$ from $\set{\lambda_i(M)}_{i=m+1}^{n}$. Denoting $\rho(D)$ to be the spectral radius of $D$, we have the following sufficient condition to guarantee this separation.
\begin{align}
\frac{\sampwidth^2}{3} - \norm{\Delta}_{F,bound} &> \rho(D) + \norm{\Delta}_{F,bound}, \nonumber \\
\Leftrightarrow \quad \frac{\sampwidth^2}{3} - \rho(D) &> 2\norm{\Delta}_{F,bound}. \label{eq:md_sep_cond_smooth}
\end{align}
Now, as shown in Section \ref{appendix:proof_md_width_cond}, $\rho(D) < R L \sampwidth^4$, where $L = \frac{m(5m+4)\abs{\kfmax}^2}{180}$, and
\begin{equation*}
R = \left\{
\begin{array}{rl}
1 & ; \ \text{if D is diagonal} \\
(n-m) & ; \ \text{if D is dense}.
\end{array} \right .
\end{equation*}
Using this fact along with Eq. \eqref{eq:norm_pert_bound} in Eq. \eqref{eq:md_sep_cond_smooth}, we arrive at the following sufficient condition that guarantees the separation of eigenvalues:
\begin{align}
\frac{\sampwidth^2}{3} - R L \sampwidth^4 > \beta_2 \sampwidth^4 + \beta_3 \sampwidth^5 + \beta_4 \sampwidth^6, \label{eq:md_quad_ineq_1} \\
\Leftrightarrow (\beta_2 + RL)\sampwidth^2 + \beta_3 \sampwidth^3 + \beta_4 \sampwidth^4 < \frac{1}{3} \label{eq:md_quad_ineq_2}
\end{align}
where $\beta_2 = 4\smoothconst m^{3/2}\sqrt{m(n-m)}$, $\beta_3 = 2(n-m)\smoothconst m^{5/2}\abs{\kfmax}$ and $\beta_4 = 2(n-m)\smoothconst^2 m^3$.
Now, clearly the solution to Eq. \eqref{eq:md_quad_ineq_1} needs to satisfy the following conditions.
\begin{align*}
(\beta_2 + RL)\sampwidth^2 < 1/3 &\Leftrightarrow \sampwidth < (3(\beta_2 + RL))^{-1/2}, \\
\beta_3 \sampwidth^3 < 1/3 &\Leftrightarrow \sampwidth < (3\beta_3)^{-1/3}, \\
\beta_4 \sampwidth^4 < \frac{1}{3} &\Leftrightarrow \sampwidth < (3\beta_4)^{-1/4}.
\end{align*}
Equivalently, the solution to Eq. \eqref{eq:md_quad_ineq_1} satisfies $\sampwidth < \alpha$, where
\begin{equation*}
\alpha = \min \set{(3(\beta_2 + RL))^{-1/2},(3\beta_3)^{-1/3},(3\beta_4)^{-1/4}}.
\end{equation*}
We thus arrive at the following sufficient condition on $\sampwidth$ in order to guarantee Eq. \eqref{eq:md_quad_ineq_2} and consequently Eq. \eqref{eq:md_sep_cond_smooth}:
\begin{align}
\sampwidth^2((\beta_2 + RL) + \beta_3 \alpha + \beta_4 \alpha^2) < \frac{1}{3}, \nonumber \\
\Leftrightarrow \sampwidth^2 < \frac{1}{3[(\beta_2 + RL) + \beta_3 \alpha + \beta_4 \alpha^2]}. \label{eq:md_a_bound_cond}
\end{align}
We now proceed to bound the angle between $\tanps$ and $\widehat{T}_PS$ by using the identity $M \bar{u}_i = \lambda_i(M) \bar{u}_i, \ \forall i=1,\dots,m$. We obtain
\begin{align*}
B_1^T \bar{u}_{i,1} + (D+D_1)\bar{u}_{i,2} = \lambda_i(M)\bar{u}_{i,2}.
\end{align*}
By taking the $l_2$-norm of both sides and using the fact that $\norm{D\bar{u}_{i,2}}_2 \leq \rho(D)\norm{\bar{u}_{i,2}}_2$, we obtain
\begin{align*}
(\lambda_i(M) - \rho(D) - \norm{D_1}_{F}) \norm{\bar{u}_{i,2}}_2 < \norm{B_1}_{F}.
\end{align*}
Now, $\lambda_i(M) > \frac{\sampwidth^2}{3} - \norm{\Delta}_{F}$ for $i=1,\dots,m$. Therefore, if $\sampwidth$ is chosen to satisfy Eq. \eqref{eq:md_a_bound_cond}, then the following holds true.
\begin{align*}
\lambda_i(M) - \rho(D) - \norm{D_1}_{F} &> \frac{\sampwidth^2}{3} - \norm{\Delta}_{F} - \rho(D) - \norm{D_1}_{F}, \\
&> \frac{\sampwidth^2}{3} - RL \sampwidth^4- 2(\norm{B_1}_{F,bound} + \norm{D_1}_{F,bound}), \\
&> 0.
\end{align*}
Using the above facts, we obtain the following upper bound on $\norm{\bar{u}_{i,2}}_2$.
\begin{align} \label{eq:eigenvec_norm_bound}
\norm{\bar{u}_{i,2}}_2 < \frac{\norm{B_1}_{F,bound}}{\frac{\sampwidth^2}{3} - RL \sampwidth^4 - 2(\norm{B_1}_{F,bound} + \norm{D_1}_{F,bound})} = \sigma_{\infty}.
\end{align}
Finally, to conclude the proof, we obtain the bound on $\abs{\angle \widehat{T}_PS,\tanps}$ by using Lemma \ref{lemma:md_angle_bound_cond}.
\begin{align*}
\norm{U_2}_{F}^2 \ = \ \sum_{i=1}^{m}\norm{\bar{u}_{i,2}}_2^2 \ < \ m \sigma_{\infty}^2, \\
\Rightarrow \cos^2(\angle \widehat{T}_PS,\tanps) \ > \ (1 - \norm{U_2}_{F}^2)^m \ > \ (1 - m\sigma_{\infty}^{2})^m.
\end{align*}
\end{proof}
%
\subsection{Proof of Theorem ~\ref{thm:md_angle_bound_smooth}} \label{appendix:proof_md_angle_samp_smooth}
\begin{proof} 
The proof follows along the lines of the proof of Theorem \ref{theorem:md_angle_bound_prob}. We start with the following identity for $i=1,\dots,m$:
\begin{equation} \label{eq:eig_val_vec_smooth_ident1}
M^{(K)}\bar{u}_i \ = \ \lambda_i(M^{(K)})\bar{u}_i.
\end{equation}
We have, $M^{(K)} = M_q^{(K)} + \Delta^{(K)}$ where
\begin{equation*}
M_q^{(K)} = \begin{bmatrix}
A^{(K)} & B^{(K)} \\
B^{(K)^T} & D^{(K)} \\
\end{bmatrix} \, \text{  and }
\Delta^{(K)} = \begin{bmatrix}
0 & B_1^{(K)} \\
B_1^{(K)^T} & D_1^{(K)} \\
\end{bmatrix}.
\end{equation*}
Let $\lambda_1(M^{(K)}) \geq \lambda_2(M^{(K)}) \dots \geq \lambda_n(M^{(K)})$ denote the eigenvalues of $M^{(K)}$ and $\bar{u}_i = [\bar{u}_{i,1}^T \  \bar{u}_{i,2}^T]^T$ denote its corresponding eigenvectors. Using Eq. \eqref{eq:eig_val_vec_smooth_ident1}, we obtain the following:
\begin{align}
(B^{(K)^T} + B_1^{(K)^T}) \bar{u}_{i,1} &+ (D^{(K)} + D_1^{(K)}) \bar{u}_{i,2} = \lambda_i(M^{(K)}) \bar{u}_{i,2} \nonumber \\
\Rightarrow (\lambda_m(M^{(K)}) - \rho(D^{(K)}) &- \norm{D_1}_{F,bound}) \norm{\bar{u}_{i,2}}_2 \label{eq:eigevec_ineq_smooth} \\ 
&< \ \norm{B^{(K)}} + \norm{B_1}_{F,bound}. \nonumber
\end{align}
We observe by Weyl's inequality \cite{Weyl1912} that the following holds true:
\begin{equation} \label{eq:m_eig_val_bound_smooth}
\lambda_m(M^{(K)}) \geq \lambda_m(M_q^{(K)}) - \norm{\Delta}_{F,bound} = \lambda_m(M_q^{(K)}) - 2\norm{B_1}_{F,bound} - \norm{D_1}_{F,bound}.
\end{equation}
If $K$ is chosen such that $K > K_{bound}$, the following events hold with high probability:
\begin{equation} \label{eq:tail_bounds_eigvals_smooth}
\set{\lambda_m(M_q^{(K)}) > s_1\frac{\sampwidth^2}{3}}, \quad \set{\rho(D^{(K)}) < s_2 RL \sampwidth^4}, \quad \set{\norm{B^{(K)}} < s_3},
\end{equation} 
where $s_1 \in (0,1)$, $s_2 > e$ and $s_3 > 0$. Thus, using Eq. \eqref{eq:m_eig_val_bound_smooth} and Eq. \eqref{eq:tail_bounds_eigvals_smooth} in Eq. \eqref{eq:eigevec_ineq_smooth}, we obtain the following:
\begin{align*}
(s_1\frac{\sampwidth^2}{3}- s_2 RL \sampwidth^4 - 2(\norm{B_1}_{F,bound} + \norm{D_1}_{F,bound})) \norm{\bar{u}_{i,2}}_2 < s_3 + \norm{B_1}_{F,bound}. 
\end{align*}
Similarly to the proof of Lemma \ref{lemma:md_angle_asymp_smooth}, one can show that the following condition is sufficient to ensure that the L.H.S. of the above inequality is strictly positive 
\begin{align} \label{eq:cond_a_smooth}
\sampwidth^2 < \frac{s_1}{3[(\beta_2+ s_2 RL) + \beta_3 \alpha + \beta_4\alpha^2]}.
\end{align}
In particular, the above condition ensures the following:
\begin{equation*}
s_1\frac{\sampwidth^2}{3}- s_2 RL \sampwidth^4 > 2\norm{\Delta}_{F,bound} = 4\norm{B_1}_{F,bound} + 2\norm{D_1}_{F,bound}.
\end{equation*}
Now, assuming that $\sampwidth$ satisfies Eq. \eqref{eq:cond_a_smooth}, we arrive at the following bound on $\norm{\bar{u}_{i,2}}_2$ for $i=1,\dots,m$:
\begin{align*}
\norm{\bar{u}_{i,2}}_2 &< \frac{s_3 + \norm{B_1}_{F,bound}}{(s_1\frac{\sampwidth^2}{3}- s_2 RL \sampwidth^4) - 2(\norm{B_1}_{F,bound} + \norm{D_1}_{F,bound})} \\
&= \sigma_s.
\end{align*}
The above bound on $\norm{\bar{u}_{i,2}}_2$ implies that $\norm{U_2}_F^2 < m \sigma_s^2$. Let 
\begin{equation*}
\sigma_{f} := \frac{\norm{B_1}_{F,bound}}{(s_1\frac{\sampwidth^2}{3}- s_2 RL \sampwidth^4) - 2(\norm{B_1}_{F,bound} + \norm{D_1}_{F,bound})}.
\end{equation*}
If for some $\tau \in (0,1)$
\begin{equation} \label{eq:ang_bound_cond_smooth}
m \sigma_s^2 < \tau^2 + m\sigma_{f}^2 \Leftrightarrow \sigma_s < (\tau^2/m + \sigma_{f}^2)^{1/2},
\end{equation}
then from Lemma \ref{lemma:md_angle_bound_cond} we obtain
\begin{equation*}
\cos^2(\angle \widehat{T}_PS, \tanps) > 1 - \tau^2 - m\sigma_{f}^2. 
\end{equation*}
Finally, we see that Eq. \eqref{eq:ang_bound_cond_smooth} is ensured if the following holds.
\begin{align*}
\frac{s_3 + \norm{B_1}_{F,bound}}{(s_1\frac{\sampwidth^2}{3}- s_2 RL \sampwidth^4) - 2(\norm{B_1}_{F,bound} + \norm{D_1}_{F,bound})} &< \left(\frac{\tau^2}{m} + \sigma_{f}^2\right)^{1/2} \\
\Leftrightarrow s_3 < [(s_1\frac{\sampwidth^2}{3}- s_2 RL \sampwidth^4) - 2(\norm{B_1}_{F,bound} &+ \norm{D_1}_{F,bound})]\left(\frac{\tau^2}{m} + \sigma_{f}^2\right)^{1/2} \\
&- \norm{B_1}_{F,bound}.
\end{align*}
This completes the proof.
\end{proof}

\end{document}